\documentclass[11pt,a4paper]{article}
\usepackage{amsfonts}
\usepackage[dvips]{graphicx}
\usepackage{amsmath}
\usepackage{mathrsfs}
\usepackage{makeidx}
\usepackage{xypic}
\usepackage[nottoc]{tocbibind}
\usepackage{pstricks}
\usepackage{bbm}
\usepackage[arrow, matrix, curve]{xy}
\usepackage{amsthm}
\usepackage{amssymb}
\usepackage{epic}
\usepackage{eepic}
\usepackage{times}
\usepackage{epsfig}
\usepackage{hyperref}
\usepackage[numbers,square,sort&compress]{natbib}

\xymatrixcolsep{0.5cm}
\xymatrixrowsep{0.5cm}
\newtheorem{theorem}{Theorem}[section]
\newtheorem{proposition}{Proposition}[section]
\newtheorem{lemma}{Lemma}[section]

\newtheorem{corollary}{Corollary}[section]

\newcommand{\beqa}{\begin{eqnarray}}
\newcommand{\eeqa}{\end{eqnarray}}

\makeindex
\numberwithin{equation}{section}

\begin{document}

\begin{flushright}
LPENSL-TH-13-08\\
YITP-SB-13-18
\end{flushright}

\par \vskip .1in \noindent

\vspace{-12pt}

\begin{center}
\begin{LARGE}
\vspace*{1cm}
{On the form factors of local operators in the\\ 

\medskip

\noindent
Bazhanov-Stroganov and  chiral Potts models}
\end{LARGE}

\vspace{50pt}

\begin{large}

{\bf N.~Grosjean}\footnote[1]{ LPTM, UMR 8089
du CNRS, Univesit\'e de Cergy-Pontoise, France, nicolas.grosjean@u-cergy.fr},~~
{\bf J.~M.~Maillet}\footnote[2]{ Laboratoire de Physique, UMR 5672
du CNRS, ENS Lyon,  France,
 maillet@ens-lyon.fr},~~
{\bf G.~Niccoli}\footnote[3]{ YITP, Stony Brook University, New York, USA, niccoli@max2.physics.sunysb.edu}
\par

\end{large}

\vspace{40pt}


\vspace{60pt}

\centerline{\bf Abstract} \vspace{1cm}
\parbox{12cm}{\small   We consider general cyclic
representations of the 6-vertex Yang-Baxter algebra and analyze the
associated quantum integrable systems, the Bazhanov-Stroganov model  and the corresponding chiral Potts model on finite size lattices. We first determine the {\it propagator operator} in terms of the chiral Potts transfer matrices and we compute the scalar product of \textit{separate states} (including the transfer matrix eigenstates) as a single determinant formulae in the framework of Sklyanin's quantum separation of variables. Then, we solve the quantum inverse problem and reconstruct the local operators in terms of the separate variables. We also determine a basis of operators  whose form factors are characterized by a single determinant formulae. This implies that  the form factors of any local operator are expressed as finite sums of determinants. Among these form factors written in determinant form are in particular those which will reproduce the chiral Potts order parameters in the thermodynamic limit. The results presented here are the generalization to the present models associated to the most general cyclic representations of the 6-vertex Yang-Baxter algebra of those we derived for the lattice sine-Gordon model.}
\end{center}
\newpage
\begin{small}
\tableofcontents
\end{small}
\newpage

\section{Introduction\label{NJG-BS-ChP-INTR}}

In the article \cite{NJG-BS-ChP-GMN12-SG} we developed an approach in the framework
of the quantum inverse scattering method (QISM) \cite{NJG-BS-ChP-SF78, NJG-BS-ChP-FT79, NJG-BS-ChP-KS79, NJG-BS-ChP-FST80, NJG-BS-ChP-F80, NJG-BS-ChP-S82, NJG-BS-ChP-KS, NJG-BS-ChP-F82, NJG-BS-ChP-F95, NJG-BS-ChP-J90, NJG-BS-ChP-Sh85, NJG-BS-ChP-Th81, NJG-BS-ChP-IK82} to
achieve the complete solution of lattice integrable quantum models by the
exact characterization of their spectrum and the computation of the matrix
elements of local operators in the eigenstates basis. This approach is addressed to the large class
of integrable quantum models whose spectrum (eigenvalues and eigenstates)
can be determined by implementing Sklyanin's quantum separation of variables
(SOV) method \cite{NJG-BS-ChP-Sk1, NJG-BS-ChP-Sk2, NJG-BS-ChP-Sk3}.  It can be considered as the
generalization to this SOV framework of the Lyon group method\footnote{This method has been introduced in \cite{NJG-BS-ChP-KMT99} for the spin-1/2 XXZ
quantum chain \cite{NJG-BS-ChP-H28, NJG-BS-ChP-Be31, NJG-BS-ChP-Hul38, NJG-BS-ChP-Or58, NJG-BS-ChP-W59, NJG-BS-ChP-YY661, NJG-BS-ChP-YY662, NJG-BS-ChP-Ga83, NJG-BS-ChP-LM66} with periodic  boundaries and further developed in \cite{NJG-BS-ChP-MT00, NJG-BS-ChP-IKitMT99, NJG-BS-ChP-KitMT00, NJG-BS-ChP-KitMST02a, NJG-BS-ChP-KitMST02b, NJG-BS-ChP-KitMST02c, NJG-BS-ChP-KitMST02d, NJG-BS-ChP-KitMST04a, NJG-BS-ChP-KitMST04b, NJG-BS-ChP-KitMST04c, NJG-BS-ChP-KitMST05a,  NJG-BS-ChP-KitMST05b, NJG-BS-ChP-KKMST07}. Its generalization to the higher spin XXX quantum chains and to the open spin-1/2 XXZ
quantum chains \cite{NJG-BS-ChP-Skly88, NJG-BS-ChP-Che84, NJG-BS-ChP-KS91, NJG-BS-ChP-MN91, NJG-BS-ChP-KS92, NJG-BS-ChP-GZ94a, NJG-BS-ChP-GZ94} with diagonal boundary conditions
has been respectively implemented in \cite{NJG-BS-ChP-K01, NJG-BS-ChP-CM07} and \cite{NJG-BS-ChP-KKMNST07, NJG-BS-ChP-K08, NJG-BS-ChP-KKMNST08}.} for the computation of matrix elements of local
operators in the algebraic Bethe ansatz settings. In \cite{NJG-BS-ChP-GMN12-SG} the approach has been developed for the lattice quantum sine-Gordon model \cite{NJG-BS-ChP-FST80,NJG-BS-ChP-IK82}\
associated by QISM to particular cyclic representations \cite{NJG-BS-ChP-Ta} of the 6-vertex
Yang-Baxter algebra. More in detail, in \cite{NJG-BS-ChP-NT, NJG-BS-ChP-GN10, NJG-BS-ChP-GN11} the complete
SOV spectrum characterization has been constructed for the lattice quantum
sine-Gordon model while in \cite{NJG-BS-ChP-GMN12-SG} the scalar product of separate states and the matrix elements of local
operators have been computed. In the present article we implement this
approach for the quantum models associated by QISM to the most general
cyclic representations of the 6-vertex Yang-Baxter algebra, i.e. the
inhomogeneous Bazhanov-Stroganov model and subsequently the chiral Potts (chP)
model \cite{NJG-BS-ChP-BS, NJG-BS-ChP-BBP, NJG-BS-ChP-B04, NJG-BS-ChP-AMcP0, NJG-BS-ChP-AMcP1, NJG-BS-ChP-AMcP2, NJG-BS-ChP-auYMcCPTY87, NJG-BS-ChP-BaPauY, NJG-BS-ChP-vGR85, NJG-BS-ChP-P87, NJG-BS-ChP-Ba89-1, NJG-BS-ChP-Ba89, NJG-BS-ChP-BaxBP90, NJG-BS-ChP-BazBR1996, NJG-BS-ChP-Baz2011, NJG-BS-ChP-BazS2012, NJG-BS-ChP-auYMcPTY, NJG-BS-ChP-McPTS, NJG-BS-ChP-auYMcPT, NJG-BS-ChP-TarasovSChP}, by exploiting the well known links
between these two models \cite{NJG-BS-ChP-BS}. We first build our two central tools for
computing matrix elements of local operators, i.e.  the expression of the scalar
products of separate states in terms of a determinant formula and the local fields
reconstruction in terms of quantum separate variables (by solving the so called quantum inverse scattering problem). Then, we use these results to compute the form factors of local operators on the transfer
matrix eigenstates and to express them as sums of determinants given by
simple deformations of the ones giving the scalar product of separate states.

\subsection{\label{NJG-BS-ChP-summary}Literature summary}

Let us first  summarize some known results  concerning 
these quantum integrable models and that are relevant for our present work. In \cite{NJG-BS-ChP-BS} the Bazhanov-Stroganov model was
introduced from its Lax operator built as a general solution to the
Yang-Baxter equation associated to the 6-vertex R-matrix. For a specific
subset of cyclic representations, in which the parameters lie on the
algebraic curves associated to the chP-model, the construction of the Baxter
Q-operator allowed for the analysis of the spectrum (eigenvalues). This 
Q-operator was shown to coincide with the transfer matrix of the integrable $%
Z_{p}$ chP-model \cite{NJG-BS-ChP-AMcP0, NJG-BS-ChP-AMcP1, NJG-BS-ChP-AMcP2, NJG-BS-ChP-auYMcCPTY87, NJG-BS-ChP-BaPauY, NJG-BS-ChP-vGR85, NJG-BS-ChP-P87, NJG-BS-ChP-Ba89-1, NJG-BS-ChP-Ba89, NJG-BS-ChP-BaxBP90}; in this way a first remarkable
connection between these two apparently very different models\footnote{Note that in a 2-dimensional statistical mechanics formulation both 
models have Boltzmann weights which satisfy the star-triangle equations.
However, while the weights of the Bazhanov-Stroganov model satisfy the difference
property in the rapidities those of the chP-model do not. 
In this respect, the link to classical integrable discrete models is quite illuminating \cite{NJG-BS-ChP-BazBR1996,NJG-BS-ChP-Baz2011,NJG-BS-ChP-BazS2012}. It is worth
recalling that the first solutions of the star-triangle equations with this
non-difference property were obtained in \cite{NJG-BS-ChP-auYMcPTY,NJG-BS-ChP-McPTS,NJG-BS-ChP-auYMcPT} while
in \cite{NJG-BS-ChP-BaPauY} the general solutions for the chP-model were derived.} was
established. Additional functional equations of fusion hierarchy type%
\footnote{%
The approach of fusion hierarchy of commuting transfer matrices was first
introduced in \cite{NJG-BS-ChP-KRS81,NJG-BS-ChP-KR}.} for commuting transfer matrices\footnote{%
The transfer matrix of the Bazhanov-Stroganov model is the second element in this hierarchy,
this explains the name $\tau _{2}$ given some times to this model.} were then exhibited in 
\cite{NJG-BS-ChP-BBP}. Bethe ansatz type equations play an important role in the
special sub-variety of the super-integrable chP-model as it was first shown in 
\cite{NJG-BS-ChP-AMcP0,NJG-BS-ChP-AMcP1,NJG-BS-ChP-AMcP2}. The connection between the Bazhanov-Stroganov model and the
chP-model allowed to introduce rigorously \cite{NJG-BS-ChP-TarasovSChP} the description
of the super-integrable chP spectrum using algebraic Bethe ansatz. The Bethe
ansatz construction was applied to the transfer matrix $\tau_2$ of the Bazhanov-Stroganov model, thus 
obtaining in a different way the Baxter results \cite{NJG-BS-ChP-Ba89-1} on the subset
of the translation-invariant eigenvectors of the super-integrable chP-model%
\footnote{%
For further analysis of the eigenstates of super-integrable chP-model see
also \cite{NJG-BS-ChP-auYP08, NJG-BS-ChP-auYP09, NJG-BS-ChP-NisD2008, NJG-BS-ChP-Roa2010}. It is interesting to mention here also that in all these analysis the underlying Onsager algebra \cite{NJG-BS-ChP-Ons1944} and  realizations of the sl2 loop algebra \cite{NJG-BS-ChP-FabM2001}, which are  symmetries for these super-integrable representations \cite{NJG-BS-ChP-BaPauY, NJG-BS-ChP-vGR85} and \cite{NJG-BS-ChP-Dav1990, NJG-BS-ChP-DatR2000, NJG-BS-ChP-Roa2005, NJG-BS-ChP-NisD2006, NJG-BS-ChP-Roa2007} have played fundamental roles.}. More recently, the extension of the eigenvalue
analysis of the Bazhanov-Stroganov model to completely general cyclic representations
was done by Baxter \cite{NJG-BS-ChP-B04}. The main tool used there was the construction
of a generalized Q-operator which satisfies the Baxter equation with the transfer matrix $\tau_2$  and the extension to these
representations of the functional relations of the fused transfer matrices.

Another important feature of the chP-model which has been the subject of
recent attention is the spontaneous magnetization. This order parameter
was first described in \cite{NJG-BS-ChP-AMcCPT89} on the basis of perturbativet
calculations developed for the special class of super-integrable
representations\footnote{%
This case both obeys Yang-Baxter integrability \cite{NJG-BS-ChP-BaPauY} and has an
underlying Onsager algebra \cite{NJG-BS-ChP-auYMcCPTY87}.}. The first non perturbative
derivation of this order parameter was achieved only recently by Baxter
under some natural analyticity assumptions and the use of a technique
introduced by Jimbo et al. \cite{NJG-BS-ChP-JMN93}. More classical techniques, like the
corner transfer matrix \cite{NJG-BS-ChP-Bax82}, could not be used, mainly because of
the very nature of the chP-model \cite{NJG-BS-ChP-Ba06}. The proof of the spontaneous magnetization
formula \cite{NJG-BS-ChP-AMcCPT89} starting from direct computations on the finite
lattice of matrix elements of the spin operators could only be achieved
after the recent introduction by Baxter \cite{NJG-BS-ChP-Ba08,NJG-BS-ChP-Ba09} of a generalized
version of the Onsager algebra for the special class of super-integrable
representations of chP-model. The matrix elements used for this proof have
been first analyzed by Au-Yang and Perk in a series of papers \cite{NJG-BS-ChP-auYP08,NJG-BS-ChP-auYP09}, \cite{NJG-BS-ChP-auYP10, NJG-BS-ChP-auYP11, NJG-BS-ChP-auYP11-2} for the case of the
super-integrable chP-model. Their factorized form, first conjectured by
Baxter \cite{NJG-BS-ChP-Ba08-1}, has been proven\footnote{%
Note that factorized formulas for the spin matrix elements exist also for
the 2D Ising model \cite{NJG-BS-ChP-BL04} and for the quantum XY-chain \cite{NJG-BS-ChP-Iorgov11}.}
by Iorgov et al \cite{NJG-BS-ChP-IPSTvG09} and used to derive the spontaneous
magnetization formula conjectured in \cite{NJG-BS-ChP-AMcCPT89}. Finally, it is worth
recalling that,  in the algebraic framework of generalized Onsager
algebra, Baxter has also first conjectured \cite{NJG-BS-ChP-Ba10} and successively
proven in \cite{NJG-BS-ChP-Ba10-2} a determinant formula for the spontaneous magnetization
of the super-integrable chP-model; this result is also used for a
further derivation of the known formula of the order parameter in the
thermodynamical limit.

\subsection{Motivations for the use of SOV}

Let us comment that in the literature we just recalled, the spectral analysis
has usually one or more of the following problems: there is no eigenstates
construction for the functional methods based only on the Baxter Q-operator
and the fusion of transfer matrices. The ABA applies only to very special
representations of the Bazhanov-Stroganov model as well as the algebraic framework
of the generalized Onsager algebra is proven to exist only in the class of
super-integrable representations of chiral Potts model. The proof of the
completeness of eigenstates is not ensured by these methods and it was so far missing in
the general p-state chP-model and Bazhanov-Stroganov model. Existing results about
this issue are mainly restricted to the case of the 3-state super-integrable
chP-model \cite{NJG-BS-ChP-DKMcCoy93} and to the reduction of the 3-state Potts model
to the trivial algebraic curve case \cite{NJG-BS-ChP-ADMcCoy92}, i.e. the
Fateev-Zamolodchikov model \cite{NJG-BS-ChP-FZ}, see also \cite{NJG-BS-ChP-FMcCoy01} and \cite{NJG-BS-ChP-NR03} for further applications of this method.

The circumstance interesting for us is that, in the case of the cyclic
representations of the Bazhanov-Stroganov model for which the algebraic Bethe 
ansatz does not apply, Sklyanin's quantum SOV can be developed to analyze the
system. This means that, for most\footnote{%
The values of the parameters of the representations for which ABA applies
define a proper sub-variety in the full space of the parameters of the
representations of the Bazhanov-Stroganov model.} of the representations of this
model, we have the opportunity to use the SOV method, which appears quite 
promising as it leads to both the eigenvalues and the eigenstates of the transfer matrix of the Bazhanov-Stroganov model with a complete spectrum construction if some
simple conditions are satisfied. The SOV analysis of these representations
was first introduced\footnote{%
There the eigenvector analysis developed in \cite{NJG-BS-ChP-I06} was used to obtain
the SOV representations of the Bazhanov-Stroganov model. See also the series of works 
\cite{NJG-BS-ChP-GIPST07,NJG-BS-ChP-GIPST08,NJG-BS-ChP-GIPS09} where the form factors of local spin operators
were computed by SOV for the special case ($p$=2) of the generalized Ising
model.} in \cite{NJG-BS-ChP-GIPS06} and further developed in \cite{NJG-BS-ChP-GN12}.  Here we will
use these SOV results as setup for the computation of the form factors of
local operators. Let us recall that in \cite{NJG-BS-ChP-GN12}, the functional equation
characterization of the transfer matrix spectrum has been derived purely on the
basis of the SOV spectrum characterization\footnote{
Note that for cyclic representations the SOV does not lead directly to the
spectrum characterization by functional equations and so, in particular, it
does not lead to Bethe equations.} together with a first proof of the
completeness of the system of equations of Bethe ansatz type\footnote{%
For Bethe ansatz methods, as the coordinate Bethe ansatz \cite{NJG-BS-ChP-Be31, NJG-BS-ChP-Bax82,
NJG-BS-ChP-ABBQ87}, the algebraic Bethe ansatz \cite{NJG-BS-ChP-FT79, NJG-BS-ChP-KS79, NJG-BS-ChP-FST80} and the
analytic Bethe ansatz \cite{NJG-BS-ChP-R83-I,NJG-BS-ChP-R83-II}, a proof of the completeness was
achieved only for few integrable quantum models, see as concrete
examples \cite{NJG-BS-ChP-MTV} for the $XXX$ Heisenberg model, \cite{NJG-BS-ChP-ORR09} for the
infinite $XXZ$ spin chain with domain wall boundary conditions and \cite{NJG-BS-ChP-Korff-11} for the nonlinear quantum Schroedinger model.} for some classes
of representations of Bazhanov-Stroganov model and chP-model and the simplicity of
these transfer matrix spectra in the inhomogeneous models.

Beyond these motivations on the spectrum analysis, the summary presented in
the previous subsection makes clear that the computations of matrix elements
of local operators are so far mainly confined to the special class of
super-integrable representations of chP-model as they were derived in the
algebraic framework of the generalized Onsager algebra. This stresses the
relevance of our approach using quantum separation of variables which leads to
form factors of local operators and applies to generic representations of Bazhanov-Stroganov model and chiral Potts model to which the methods based on
generalized Onsager algebra do not apply up to now.

\subsection{Paper organization}
In order to make the paper self-contained we dedicate Sections 2 and 3 to review the material presented in \cite{NJG-BS-ChP-GN12} simultaneously integrating it with the presentation of new results needed for our purposes. In particular, Section 2 provides the definition of the Bazhanov-Stroganov model and the main results of \cite{NJG-BS-ChP-GN12} on SOV while Subsection 2.3.1 and 2.4.2 contain new results on the SOV decomposition of the identity and the characterization of the transfer matrix eigenstates. Section 3 provides the definition of the chiral Potts model and the main results obtained by SOV method in \cite{NJG-BS-ChP-GN12}. The scalar products of separate states and the decomposition of the identity w.r.t. the transfer matrix eigenbasis are derived in Section 4. Section 5 contains the characterization of the  propagator operator of the Bazhanov-Stroganov model in terms of the chiral Potts transfer matrices. The reconstruction of local operators in terms of separate variables is given in Section 6 while their form factors are expressed in terms of finite size determinants in Section 7. The last section addresses some comments on these results and a comparison with the existing literature.
\section{The Bazhanov-Stroganov model}

We use this section to give our notations and  to briefly  recall the main results derived in \cite{NJG-BS-ChP-GN12} on the
spectrum description by SOV of the Bazhanov-Stroganov model and chiral Potts model that are useful for our purposes.

\subsection{The Bazhanov-Stroganov model: definitions and first properties}

We define in  the $\mathsf{N}$ sites of the chain $\mathsf{N}$ local Weyl
algebras $\mathcal{W}_n$ and denote by $\mathsf{u}_{n}$ and $\mathsf{v}_{n}$
their generators: 
\begin{equation}
\mathsf{u}_{n}\mathsf{v}_{m}=q^{\delta _{n,m}}\mathsf{v}_{m}\mathsf{u}_{n}%
\text{ \ \ }\forall n,m\in \{1,...,\mathsf{N}\}.
\end{equation}
The Lax operator of the Bazhanov-Stroganov model reads\footnote{%
Up to different notations, this Lax operator coincides with the one
introduced in \cite{NJG-BS-ChP-BS}.}: 
\begin{equation}
\mathsf{L}_{n}(\lambda )\equiv \left( 
\begin{array}{cc}
\lambda \alpha _{n}\mathsf{v}_{n}-\beta _{n}\lambda ^{-1}\mathsf{v}_{n}^{-1}
& \mathsf{u}_{n}\left( q^{-1/2} \mathbbm{a}_{n}\mathsf{v}_{n}+q^{1/2} %
\mathbbm{b}_{n}\mathsf{v}_{n}^{-1}\right) \\ 
\mathsf{u}_{n}^{-1}\left( q^{1/2} \mathbbm{c}_{n}\mathsf{v}_{n}+q^{-1/2} %
\mathbbm{d}_{n}\mathsf{v}_{n}^{-1}\right) & \gamma _{n}\mathsf{v}%
_{n}/\lambda -\delta _{n}\lambda /\mathsf{v}_{n}%
\end{array}%
\right) ,
\end{equation}%
where $\alpha_n$, $\beta_n$, $\gamma_n$, $\delta_n$, $\mathbbm{a}_n$, $%
\mathbbm{b}_n$, $\mathbbm{c}_n$ and $\mathbbm{d}_n$ are constants associated
to the site $n$ of the chain subject to the relations : 
\begin{equation}
\alpha_n \gamma _{n}=\mathbbm{a}_{n}\mathbbm{c}_{n},\text{ \ \ \ \ }\beta_n
\delta _{n}=\mathbbm{b}_{n}\mathbbm{d}_{n}.
\end{equation}
The monodromy matrix of the model is defined in terms of the Lax operators
by: 
\begin{equation}
\mathsf{M}(\lambda )=\left( 
\begin{array}{cc}
\mathsf{A}(\lambda ) & \mathsf{B}(\lambda ) \\ 
\mathsf{C}(\lambda ) & \mathsf{D}(\lambda )%
\end{array}%
\right) \equiv \mathsf{L}_{\mathsf{N}}(\lambda )\cdots \mathsf{L}%
_{1}(\lambda ).  \label{NJG-BS-ChP-Mdef}
\end{equation}%
It satisfies the quadratic Yang-Baxter relation :%
\begin{equation}
R(\lambda /\mu )\,(\mathsf{M}(\lambda )\otimes 1)\,(1\otimes \mathsf{M}(\mu
))\,=\,(1\otimes \mathsf{M}(\mu ))\,(\mathsf{M}(\lambda )\otimes 1)R(\lambda
/\mu )\,,  \label{NJG-BS-ChP-YBA}
\end{equation}%
driven by the six-vertex (standard) $R$-matrix:%
\begin{equation}
R(\lambda )=\left( 
\begin{array}{cccc}
q\lambda -q^{-1}\lambda ^{-1} &  &  &  \\[-1mm] 
& \lambda -\lambda ^{-1} & q-q^{-1} &  \\[-1mm] 
& q-q^{-1} & \lambda -\lambda ^{-1} &  \\[-1mm] 
&  &  & q\lambda -q^{-1}\lambda ^{-1}%
\end{array}%
\right) \,.  \label{NJG-BS-ChP-Rlsg}
\end{equation}%
Then the elements of $\mathsf{M}(\lambda )$ generate a representation $%
\mathcal{R}_{\mathsf{N}}$ of the so-called Yang-Baxter algebra. In
particular, (\ref{NJG-BS-ChP-YBA}) yields the relation $\left[ \mathsf{B}(\lambda), 
\mathsf{B}(\mu) \right]=0$, for all $\lambda$ and $\mu$, and the mutual
commutativity of the elements of the one parameter family of transfer matrix operators: 
\begin{equation}
\tau _{2}(\lambda )\equiv \,\mathrm{tr}_{\mathbb{C}^{2}}\mathsf{M}(\lambda
)\,=\mathsf{A}(\lambda )+\mathsf{D}(\lambda ). \label{NJG-BS-ChP-Tdef}
\end{equation}%
Let us introduce the operator: 
\begin{equation}
\Theta =\prod_{n=1}^{\mathsf{N}}\mathsf{v}_{n},  \label{NJG-BS-ChP-topological-charge}
\end{equation}%
which plays the role of a \textit{grading operator} in the Yang-Baxter
algebra\footnote{%
The proof of the lemma is given following the same steps of that of
Proposition 6 of \cite{NJG-BS-ChP-NT}.}:

\begin{lemma}
\label{NJG-BS-ChP-YB-G-T} \textbf{(Lemma 1 of \cite{NJG-BS-ChP-GN12})} $\Theta $ commutes with the
transfer matrix $T(\lambda)$. More precisely, its commutation relations with the
elements of the monodromy matrix are: 
\begin{eqnarray}
\Theta \mathsf{C}(\lambda ) &=&q\mathsf{C}(\lambda )\Theta \text{, \ \ \ }[%
\mathsf{A}(\lambda ),\Theta ]=0, \\
\mathsf{B}(\lambda )\Theta &=&q\Theta \mathsf{B}(\lambda ),\text{ \ \ }[%
\mathsf{D}(\lambda ),\Theta ]=0.
\end{eqnarray}
\end{lemma}

Besides, the $\Theta $-charge allows to express the following
asymptotics in both $\lambda \to 0$ and $\lambda \to \infty$ of the leading
operators of the Yang-Baxter algebras: 
\begin{align}
\mathsf{A}(\lambda )& =\left( \lambda ^{\mathsf{N}}\Theta \prod_{a=1}^{%
\mathsf{N}}\alpha _{n}+(-1)^{\mathsf{N}}\lambda ^{-\mathsf{N}}\Theta
^{-1}\prod_{a=1}^{\mathsf{N}}\beta _{a}\right) +\sum_{i=1}^{\mathsf{N}-1} 
\mathsf{A}_i \lambda^{\mathsf{N}-2i},  \label{NJG-BS-ChP-asymp-A} \\
\mathsf{D}(\lambda )& =\left( \lambda ^{-\mathsf{N}}\Theta \prod_{a=1}^{%
\mathsf{N}}\gamma _{a}+(-1)^{\mathsf{N}}\lambda ^{\mathsf{N}}\Theta
^{-1}\prod_{a=1}^{\mathsf{N}}\delta _{a}\right) +\sum_{i=1}^{\mathsf{N}-1} 
\mathsf{D}_i \lambda^{\mathsf{N}-2i},  \label{NJG-BS-ChP-asymp-D}
\end{align}%
with $\mathsf{A}_i$ and $\mathsf{D}_i$ being operators, and so%
\begin{equation}
\lim_{\log \lambda \rightarrow \mp \infty }\lambda ^{\pm \mathsf{N}}\tau
_{2}(\lambda )=\left( \Theta ^{\mp 1}a_{\mp }+\Theta ^{\pm 1}d_{\mp }\right)
,  \label{NJG-BS-ChP-asymptotics-t}
\end{equation}%
where $\lim_{\log \lambda \to - \infty}$ means $\lim_{\lambda \to 0}$, $%
\lim_{\log \lambda \to + \infty}$ means $\lim_{\lambda \to \infty}$ and:%
\begin{equation}
a_{+}\equiv \prod_{a=1}^{\mathsf{N}}\alpha _{a},\text{ \ \ }a_{-}\equiv
(-1)^{\mathsf{N}}\prod_{a=1}^{\mathsf{N}}\beta _{a},\text{\ \ }d_{+}\equiv
(-1)^{\mathsf{N}}\prod_{a=1}^{\mathsf{N}}\delta _{a},\text{ \ }d_{-}\equiv
\prod_{a=1}^{\mathsf{N}}\gamma _{a}.  \label{NJG-BS-ChP-Asymptotic-A-D}
\end{equation}%
We only consider here representations for which the Weyl algebra generators $%
\mathsf{u}_{n}$ and $\mathsf{v}_{n}$ are unitary operators; then the
following Hermitian conjugation properties of the generators of Yang-Baxter
algebra hold:

\begin{lemma}
\textbf{(Lemma 2 of \cite{NJG-BS-ChP-GN12})} Let $\epsilon \in \{ +1,-1\}$, then under
the following constrains on the parameters: 
\begin{equation}
\mathbbm{c}_{n}=-\epsilon \mathbbm{b}_{n}^{\ast },\text{ \ }\mathbbm{d}%
_{n}=-\epsilon \mathbbm{a}_{n}^{\ast },\text{ \ }\beta _{n}=\epsilon \left( %
\mathbbm{a}_{n}^{\ast }\mathbbm{b}_{n}\right) /\alpha _{n}^{\ast },\qquad ,
\label{NJG-BS-ChP-Self-adjointness Condition}
\end{equation}%
the generators of the Yang-Baxter algebra satisfy the following
transformations under Hermitian conjugation: 
\begin{equation}
\mathsf{M}(\lambda )^{\dagger }\equiv \left( 
\begin{array}{cc}
\mathsf{A}^{\dagger }(\lambda ) & \mathsf{B}^{\dagger }(\lambda ) \\ 
\mathsf{C}^{\dagger }(\lambda ) & \mathsf{D}^{\dagger }(\lambda )%
\end{array}%
\right) =\left( 
\begin{array}{cc}
\mathsf{D}(\lambda ^{\ast }) & -\epsilon \mathsf{C}(\lambda ^{\ast }) \\ 
-\epsilon \mathsf{B}(\lambda ^{\ast }) & \mathsf{A}(\lambda ^{\ast })%
\end{array}%
\right) ,  \label{NJG-BS-ChP-Hermit-Monodromy}
\end{equation}%
which, in particular, imply the self-adjointness of the transfer matrix ${%
\tau _{2}}(\lambda )$ for real $\lambda $.
\end{lemma}

\subsection{General cyclic representations}

Here, we will consider general cyclic representations for which $\mathsf{v}_{n}$ and $%
\mathsf{u}_{n}\,$ have discrete spectra, and we will restrict our study to the case where $q$ is a root of unity:%
\begin{equation}  \label{NJG-BS-ChP-beta}
q=e^{-i\pi \beta ^{2}},\text{ \ \ \ }\beta ^{2}\,=\,\frac{p^{\prime }}{p}%
\,,\qquad p,p^{\prime }\in \mathbb{Z}^{>0}\,,
\end{equation}%
with $p$ odd and $p^{\prime }$ even being two co-prime numbers so that $q^{p}=1$. The condition (\ref{NJG-BS-ChP-beta}) implies that the powers $p$ of the generators $\mathsf{u}_{n}$ and $%
\mathsf{v}_{n}$ are central elements of each Weyl algebra $\mathcal{W}_{n}$.
In this case, we fix them to the identity:%
\begin{equation}
\mathsf{v}_{n}^{p}=1,\text{ \ }\mathsf{u}_{n}^{p}=1.
\end{equation}
We associate to any site $n$ of the chain a $p$-dimensional linear space R$_{n}$ ; we can define on it the following cyclic representation of $%
\mathcal{W}_{n}$:%
\begin{equation}
\mathsf{v}_{n}|k_{n}\rangle \equiv q^{k_{n}}|k_{n}\rangle ,\text{ \ }\mathsf{%
u}_{n}|k_{n}\rangle \equiv |k_{n}-1\rangle ,\text{\ \ \ \ }\forall k_{n}\in
\{0,...,p-1\},  \label{NJG-BS-ChP-v-eigenbasis}
\end{equation}%
with the following cyclic condition:%
\begin{equation}
|k_{n}+p\rangle \equiv |k_{n}\rangle .
\end{equation}%
The vectors $|k_{n}\rangle $ give a $\mathsf{v}_{n}$-eigenbasis of the
local space R$_{n}$. Let L$_{n}$ be the linear space dual of R$_{n}$ and let 
$\langle k_{n}|$ be the vectors of the dual basis defined by: 
\begin{equation}
\langle k_{n}|k_{n}^{\prime }\rangle =(|k_{n}\rangle ,|k_{n}^{\prime
}\rangle )\equiv \delta _{k_{n},k_{n}^{\prime }}\text{ \ \ }\forall
k_{n},k_{n}^{\prime }\in \{0,...,p-1\}.
\end{equation}%
The generators $\mathsf{u}_{n}$ and $\mathsf{v}_{n}$ being unitary,
the covectors $\langle k_{n}|$ define a $\mathsf{v}_{n}$-eigenbasis in the
dual space L$_{n}$. This induces the following left representation of Weyl algebra $%
\mathcal{W}_{n}$:%
\begin{equation}
\langle k_{n}|\mathsf{v}_{n}=q^{k_{n}}\langle k_{n}|,\text{ \ }\langle k_{n}|%
\mathsf{u}_{n}=\langle k_{n}+1|,\text{\ \ \ \ }\forall k_{n}\in
\{0,...,p-1\},
\end{equation}%
with the cyclic condition: 
\begin{equation}
\langle k_{n}|=\langle k_{n}+p|.
\end{equation}

In the \textit{left} and \textit{right} linear spaces:%
\begin{equation}
\mathcal{L}_{\mathsf{N}}\equiv \otimes _{n=1}^{\mathsf{N}}\text{L}_{n},\text{
\ \ \ \ }\mathcal{R}_{\mathsf{N}}\equiv \otimes _{n=1}^{\mathsf{N}}\text{R}%
_{n},
\end{equation}%
these representations of the Weyl algebras $\mathcal{W}_{n}$ determine left and right cyclic
representations of dimension $p^{\mathsf{N}}$ of the
monodromy matrix elements, and therefore of the Yang-Baxter
algebra. In the following, we will denote with $\mathcal{R}_{\mathsf{N}}^{%
{\small \text{S-adj}}}$ the sub-variety of the space of representations $%
\mathcal{R}_{\mathsf{N}}$ defined by the condition (\ref{NJG-BS-ChP-Self-adjointness
Condition}).

\subsubsection{Centrality of operator averages}

We define the average value $\mathcal{O}$ of any operator matrix element  $\mathsf{O}$ of the
monodromy matrix $\mathsf{M}(\lambda )$ by
\begin{equation}
\mathcal{O}(\Lambda )\,=\,\prod_{k=1}^{p}\mathsf{O}(q^{k}\lambda )\,,\qquad
\Lambda \,=\,\lambda ^{p},  \label{NJG-BS-ChP-avdef}
\end{equation}
then the commutativity of each family of operators $\mathsf{A}(\lambda )$, $%
\mathsf{B}(\lambda )$, $\mathsf{C}(\lambda )$ and $\mathsf{D}(\lambda )$
implies that the corresponding average values are functions of $\Lambda $.

\begin{proposition}
\label{NJG-BS-ChP-central} \textbf{(Proposition 1 of \cite{NJG-BS-ChP-GN12})}

\begin{itemize}
\item[a)] The average values of the monodromy matrix
entries, $\mathcal{A}(\Lambda )$, $\mathcal{B}(\Lambda )$%
, $\mathcal{C}(\Lambda )$, $\mathcal{D}(\Lambda )$, are central elements. They also satisfy, in the case of self-adjoint representations $\mathcal{R}_{\mathsf{N}}^{{\small \text{S-adj}}}$, the following relations under complex conjugation: 
\begin{equation}
(\mathcal{A}(\Lambda ))^{\ast }\equiv \mathcal{D}(\Lambda ^{\ast }),\ \ \ \
\ (\mathcal{B}(\Lambda ))^{\ast }\equiv -\epsilon \mathcal{C}(\Lambda ^{\ast
}),  \label{NJG-BS-ChP-H-cj-A-D}
\end{equation}%

\item[b)] Let%
\begin{equation}
\mathcal{M}(\Lambda )\,\equiv \,\left( 
\begin{array}{cc}
\mathcal{A}(\Lambda ) & \mathcal{B}(\Lambda ) \\ 
\mathcal{C}(\Lambda ) & \mathcal{D}(\Lambda )%
\end{array}%
\right)
\end{equation}%
be the 2$\times $2 matrix made of the average values of the elements of the
monodromy matrix $\mathsf{M}(\lambda )$, then it holds: 
\begin{equation}
\mathcal{M}(\Lambda )\,=\,\mathcal{L}_{\mathsf{N}}(\Lambda )\,\mathcal{L}_{%
\mathsf{N}-1}(\Lambda )\,\dots \,\mathcal{L}_{1}(\Lambda )\,,
\end{equation}%
where: 
\begin{equation}
\mathcal{L}_{n}(\Lambda )\equiv \left( 
\begin{array}{cc}
\Lambda \alpha _{n}^{p}-\beta _{n}^{p}/\Lambda & q^{p/2}(\mathbbm{a}_{n}^{p}+%
\mathbbm{b}_{n}^{p}) \\ 
q^{p/2}(\mathbbm{c}_{n}^{p}+\mathbbm{d}_{n}^{p}) & \gamma _{n}^{p}/\Lambda
-\Lambda \delta _{n}^{p}%
\end{array}%
\right) ,  \label{NJG-BS-ChP-Average-L}
\end{equation}%
is the 2$\times $2 matrix made of the average values of the elements of the
Lax matrix $\mathsf{L}_{n}(\lambda )$.
\end{itemize}
\end{proposition}

\subsubsection{Quantum determinant}

The following linear combination of products of the Yang-Baxter generators: 
\begin{equation}  \label{NJG-BS-ChP-q-det-f}
\mathrm{det_{q}}\mathsf{M}(\lambda )\,\equiv \,\mathsf{A}(\lambda )\mathsf{D}%
(\lambda /q)-\mathsf{B}(\lambda )\mathsf{C}(\lambda /q),
\end{equation}
is called quantum determinant and it is central\footnote{%
The centrality of the quantum determinant in the Yang-Baxter algebra was
first discovered in \cite{NJG-BS-ChP-IK81}, see also \cite{NJG-BS-ChP-IK09}.
} in this algebra. It admits the following factorized form: 
\begin{equation}
\text{det}_{\text{q}}\mathsf{M}(\lambda )=\prod_{n=1}^{\mathsf{N}}\text{det}%
_{\text{q}}\mathsf{L}_{n}(\lambda ),
\end{equation}%
in terms of the local quantum determinants: 
\begin{equation}
\mathrm{det_{q}}\mathsf{L}_{n}(\lambda )\equiv \left( \mathsf{L}_{n}(\lambda
)\right) _{11}\left( \mathsf{L}_{n}(\lambda /q)\right) _{22}-\left( \mathsf{L%
}_{n}\right) _{12}\left( \mathsf{L}_{n}\right) _{21}.
\end{equation}%
In the Bazhanov-Stroganov model it reads: 
\begin{eqnarray}
\mathrm{det_{q}}\mathsf{M}(\lambda ) &=&\prod_{n=1}^{\mathsf{N}}k_{n}(\frac{%
\lambda }{\mu _{n,+}}-\frac{\mu _{n,+}}{\lambda })(\frac{\lambda }{\mu _{n,-}%
}-\frac{\mu _{n,-}}{\lambda })  \notag  \label{NJG-BS-ChP-explicit-q-det} \\
&=&(-q)^{\mathsf{N}}\prod_{n=1}^{\mathsf{N}}\frac{\beta _{n}\mathbbm{a}_{n}%
\mathbbm{c}_{n}}{\alpha _{n}}(\frac{1}{\lambda }+q^{-1}\frac{\mathbbm{b}%
_{n}\alpha _{n}}{\mathbbm{a}_{n}\beta _{n}}\lambda )(\frac{1}{\lambda }%
+q^{-1}\frac{\mathbbm{d}_{n}\alpha _{n}}{\mathbbm{c}_{n}\beta _{n}}\lambda ),
\end{eqnarray}%
where: 
\begin{equation}
k_{n}\equiv \left( \mathbbm{a}_{n}\mathbbm{b}_{n}\mathbbm{c}_{n}\mathbbm{d}%
_{n}\right) ^{1/2},\text{ \ }\mu _{n,h}\equiv \left\{ 
\begin{array}{c}
iq^{1/2}\left( \mathbbm{a}_{n}\beta _{n}/\alpha _{n}\mathbbm{b}_{n}\right)
^{1/2}\text{ \ \ }h=+, \\ 
iq^{1/2}\left( \mathbbm{c}_{n}\beta _{n}/\alpha _{n}\mathbbm{d}_{n}\right)
^{1/2}\text{ \ \ }h=-.%
\end{array}%
\right.
\end{equation}%
Moreover, for the representations that satisfy (\ref{NJG-BS-ChP-Self-adjointness
Condition}) the quantum determinant reads\footnote{%
Remark that it depends on the parameters in Lax operators only through their
modules.}: 
\begin{equation}
\mathrm{det_{q}}\mathsf{M}(\lambda )=q^{\mathsf{N}}\prod_{n=1}^{\mathsf{N}}%
\frac{|\mathbbm{a}_{n}|^{2}|\mathbbm{b}_{n}|^{2}}{|\alpha _{n}|^{2}}(\frac{1%
}{\lambda }+\epsilon q^{-1}\frac{|\alpha _{n}|^{2}}{|\mathbbm{a}_{n}|^{2}}%
\lambda )(\frac{1}{\lambda }+\epsilon q^{-1}\frac{|\alpha _{n}|^{2}}{|%
\mathbbm{b}_{n}|^{2}}\lambda ).
\end{equation}%
Let us define the following functions that will be crucial in the rest of
the paper:%
\begin{equation}
\bar{\textsc{a}}(\lambda )\equiv \alpha (\lambda )\text{\textsc{a}}(\lambda
),\qquad \,\,\bar{\textsc{d}}(\lambda )\equiv \alpha ^{-1}(q\lambda )\text{%
\textsc{d}}(\lambda )
\end{equation}%
where:%
\begin{equation}
\text{\textsc{a}}(\lambda )\equiv \prod_{n=1}^{\mathsf{N}}(\beta _{n}\alpha
_{n})^{1/2}(\frac{\lambda }{\mu _{n,+}}-\frac{\mu _{n,+}}{\lambda }),\qquad 
\text{\textsc{d}}(\lambda )\equiv \prod_{n=1}^{\mathsf{N}}(\frac{\mathbbm{a}%
_{n}\mathbbm{b}_{n}\mathbbm{c}_{n}\mathbbm{d}_{n}}{\alpha _{n}\beta _{n}}%
)^{1/2}(\frac{q\lambda }{\mu _{n,-}}-\frac{\mu _{n,-}}{q\lambda }).
\label{NJG-BS-ChP-L-poly-a-d}
\end{equation}%
They always satisfy the condition: 
\begin{equation}
\text{det}_{\text{q}}\mathsf{M}(\lambda )=\bar{\textsc{a}}(\lambda )\bar{%
\textsc{d}}(\lambda /q),  \label{NJG-BS-ChP-relation-q-det}
\end{equation}%
while the function $\alpha (\lambda )$ is defined by the requirement: 
\begin{equation}
\prod_{n=1}^{p}\bar{\textsc{a}}(\lambda q^{n})+\prod_{n=1}^{p}\bar{\textsc{%
d}}(\lambda q^{n})=\mathcal{A}(\Lambda )+\mathcal{D}(\Lambda ).
\label{NJG-BS-ChP-relation-averages}
\end{equation}%
Note that this last condition is a second order equation in the average $%
\prod_{n=1}^{p}\alpha (q^{n}\lambda )$ and then we have only two possible
choices for the averages of the functions $\bar{\textsc{a}}(\lambda )$ and $%
\bar{\textsc{d}}(\lambda )$: 
\begin{equation}
\prod_{n=1}^{p}\bar{\textsc{a}}(\lambda q^{n})=\Omega _{\epsilon }\left(
\Lambda \right) ,\text{ \ \ }\prod_{n=1}^{p}\bar{\textsc{d}}(\lambda
q^{n})=\Omega _{-\epsilon }\left( \Lambda \right) ,
\end{equation}%
where $\epsilon =\mp $ and $\Omega _{\pm }$ are the two eigenvalues of the $%
2\times 2$ matrix $\mathcal{M}(\Lambda )$ composed by the averages of the
Yang-Baxter generators.

\subsection{SOV-representations and the Yang-Baxter algebra}

\label{NJG-BS-ChP-SOV-Gen} 
The spectral problem of the transfer matrix $\tau _{2}(\lambda )$ admits a separate variables representation in the basis which diagonalize the commutative family of operators $\mathsf{B}(\lambda )$ as generally argued by Sklyanin in \cite{NJG-BS-ChP-Sk1,NJG-BS-ChP-Sk2,NJG-BS-ChP-Sk3}. In \cite{NJG-BS-ChP-GN12} it has been proven:

\begin{theorem}
\textbf{(Theorem 1 of \cite{NJG-BS-ChP-GN12})} For almost all the values of the
parameters of the representation, there exists a SOV representation for the Bazhanov-Stroganov model; in this case $\mathsf{B}(\lambda )$\ is diagonalizable and
has simple spectrum.
\end{theorem}

Let us recall here the left SOV-representations of the generators of the
Yang-Baxter algebra for the Bazhanov-Stroganov model. Let $\langle \,\boldsymbol\eta_{\mathbf{k}}\,|$ be the generic element of a basis of eigenvectors of $%
\mathsf{B}(\lambda )$: 
\begin{equation}
\langle \,{\boldsymbol\eta _{\mathbf{k}}}\,|\mathsf{B}(\lambda )\,=\,\eta _{\mathsf{N}%
}\,b_{\boldsymbol\eta _{\mathbf{k}}}(\lambda )\,\langle \,{\boldsymbol\eta _{\mathbf{k}}}%
\,|\,,\qquad b_{\boldsymbol\eta _{\mathbf{k}}}(\lambda )\,\equiv \,\prod_{a=1}^{\mathsf{%
N}-1}\left( \lambda /{\eta _a^{(k_a)}}-{\eta _a^{(k_a)}}/\lambda \right) \,,  \label{NJG-BS-ChP-Bdef}
\end{equation}%
and 
\begin{equation}
\boldsymbol\eta_{\mathbf{k}}\in {\mathsf{Z}_{\mathsf{B}}}\,\equiv \,\left\{
\,({\eta _1^{(k_1)}}\equiv q^{k_{1}}\eta _{1}^{(0)},\dots ,{\eta _\mathsf{N}^{(k_\mathsf{N})}}\equiv q^{k_{\mathsf{N}}}\eta _{\mathsf{N}%
}^{(0)})\,;\,{\mathbf{k}}\equiv (k_{1},\dots ,k_{\mathsf{N}})\in \mathbbm{Z}%
_{p}^{\mathsf{N}}\,\right\} \,,
\end{equation}%
where $\eta _{a}^{(0)}$ are fixed constants\footnote{Here, the simplicity of the spectrum of $\mathsf{B}(\lambda )$ is equivalent
to the requirement $\left( \eta _{a}^{(0)}\right) ^{p}\neq \left( \eta
_{b}^{(0)}\right) ^{p}$ for any $a\neq b\in \{1,\dots ,\mathsf{N}-1\}$.} of the representations.
For simplicity, when possible we will omit the subscript ${\mathbf{k}}$ in $\langle {\boldsymbol\eta _{\mathbf{k}}}\,|$. The action of the remaining generators of the
Yang-Baxter algebra on arbitrary states $\langle \,\boldsymbol\eta |$ reads: 
\begin{align}
\langle \,\boldsymbol\eta \,|\mathsf{A}(\lambda )\,=\,& \,b_{\boldsymbol\eta }(\lambda )\left[
\lambda \eta _{\mathsf{A}}^{(+)}\langle \,q^{-\delta _{\mathsf{N}}}\boldsymbol\eta
\,|+\lambda ^{-1}\eta _{\mathsf{A}}^{(-)}\langle \,q^{\delta _{\mathsf{N}%
}}\boldsymbol\eta \,|\right] +\sum_{a=1}^{\mathsf{N}-1}\prod_{b\neq a}\frac{\lambda
/\eta _{b}-\eta _{b}/\lambda }{\eta _{a}/\eta _{b}-\eta _{b}/\eta _{a}}\,%
\mathtt{a}^{(SOV)}(\eta _{a})\,\langle \,q^{-\delta _{a}}\boldsymbol\eta \,|\,,
\label{NJG-BS-ChP-SAdef} \\
\langle \,\boldsymbol\eta \,|\mathsf{D}(\lambda )\,=\,& \,b_{\boldsymbol\eta }(\lambda )\left[
\lambda \eta _{\mathsf{D}}^{(+)}\langle \,q^{\delta _{\mathsf{N}}}\boldsymbol\eta
\,|+\lambda ^{-1}\eta _{\mathsf{D}}^{(-)}\langle \,q^{-\delta _{\mathsf{N}%
}}\boldsymbol\eta \,|\right] +\sum_{a=1}^{\mathsf{N}-1}\prod_{b\neq a}\frac{\lambda
/\eta _{b}-\eta _{b}/\lambda }{\eta _{a}/\eta _{b}-\eta _{b}/\eta _{a}}\,%
\mathtt{d}^{(SOV)}(\eta _{a})\,\langle \,q^{\delta _{a}}\boldsymbol\eta \,|\,,
\label{NJG-BS-ChP-SDdef}
\end{align}%
where: 
\begin{equation}
\eta _{\mathsf{A}}^{(\pm )}=(\pm 1)^{\mathsf{N}-1}a_{\pm }\prod_{n=1}^{%
\mathsf{N}-1}\eta _{n}^{\pm 1},\,\,\,\,\,\,\,\,\,\,\eta _{\mathsf{D}}^{(\pm
)}=(\pm 1)^{\mathsf{N}-1}d_{\pm }\prod_{n=1}^{\mathsf{N}-1}\eta _{n}^{\pm 1},
\label{NJG-BS-ChP-ZAD-asymp}
\end{equation}%
and the states $\langle \,q^{\pm \delta _{a}}\boldsymbol\eta \,|$ are defined by: 
\begin{equation}
\langle \,q^{\pm \delta _{a}}\boldsymbol\eta \,|\equiv \langle \,\eta _{1},\dots
,q^{\pm 1}\eta _{a},\dots ,\eta _{\mathsf{N}}\,|\,.
\end{equation}%
Finally, the quantum determinant relation defines uniquely $\mathsf{C}(\lambda )$. The expressions (\ref{NJG-BS-ChP-SAdef}) and (\ref{NJG-BS-ChP-SDdef})
contain complex-valued coefficients $\mathtt{a}^{(SOV)}(\eta _{a})$ and $%
\mathtt{d}^{(SOV)}(\eta _{a})$ which completely characterize the SOV
representation. These coefficients have to be solution of the quantum
determinant conditions: 
\begin{equation}
\mathrm{det_{q}}\mathsf{M}(\eta _{r})\,=\,\mathtt{a}^{(SOV)}(\eta _{r})%
\mathtt{d}^{(SOV)}(q^{-1}\eta _{r})\,,\quad \forall r=1,\dots ,\mathsf{N}%
-1\,,  \label{NJG-BS-ChP-addet}
\end{equation}%
and of the average conditions: 
\begin{equation}
\mathcal{A}(Z_{r}\equiv\eta
_{r}^p)\,\equiv \,\prod_{k=1}^{p}\mathtt{a}^{(SOV)}(q^{k}\eta
_{r})\,,\qquad \mathcal{D}(Z_{r})\,\equiv \,\prod_{k=1}^{p}\mathtt{d}%
^{(SOV)}(q^{k}\eta _{r})\,,\qquad \forall r\in \{1,\dots ,\mathsf{N}-1\}.
\label{NJG-BS-ChP-ADaver}
\end{equation}%
In a SOV representation, some freedom is left in the choice of $\mathtt{a}%
^{(SOV)}(\eta _{r})$ and $\mathtt{d}^{(SOV)}(\eta _{r})$. It  can be
parametrized by the gauge transformation written in terms of an arbitrary function $f$: 
\begin{equation}
\mathtt{\tilde{a}}^{(SOV)}(\eta _{r})\,=\,\mathtt{a}^{(SOV)}(\eta _{r})\frac{%
f(\eta _{r}q^{-1})}{f(\eta _{r})}\,,\qquad \mathtt{\tilde{d}}^{(SOV)}(\eta
_{r})\,=\,\mathtt{d}^{(SOV)}(\eta _{r})\frac{f(\eta _{r}q)}{f(\eta _{r})}\,;
\label{NJG-BS-ChP-gauge}
\end{equation}%
which just amounts to the following change of normalization for the states of the $\mathsf{B}$%
-eigenbasis: 
\begin{equation}
\langle \,\boldsymbol\eta \,|\,\rightarrow \,\prod_{r=1}^{\mathsf{N}-1}f^{-1}(\eta
_{r})\langle \,\boldsymbol\eta \,|\,.
\end{equation}%
Similarly, we can construct a right SOV-representation of the Yang-Baxter
generators by the following actions:%
\begin{align}
\mathsf{B}(\lambda )|\boldsymbol\eta \rangle & =|\boldsymbol\eta \rangle \eta _{\mathsf{N}}b_{\boldsymbol\eta
}(\lambda )\,,  \label{NJG-BS-ChP-SOV-B-R} \\
&  \notag \\
\mathsf{A}(\lambda )|\boldsymbol\eta \rangle & =[|q^{\delta _{\mathsf{N}}}\boldsymbol\eta \rangle
\eta _{\mathsf{A}}^{(+)}\lambda +|q^{-\delta _{\mathsf{N}}}\boldsymbol\eta \rangle 
\frac{\eta _{\mathsf{A}}^{(-)}}{\lambda }]b_{\boldsymbol\eta }(\lambda )\,+\sum_{a=1}^{%
\mathsf{N}-1}|q^{\delta _{a}}\boldsymbol\eta \rangle \prod_{b\neq a}\frac{(\lambda
/\eta _{b}-\eta _{b}/\lambda )}{(\eta _{a}/\eta _{b}-\eta _{b}/\eta _{a})}%
\mathtt{\bar{a}}^{(SOV)}(\eta _{a}),  \label{NJG-BS-ChP-SOV-A-R} \\
&  \notag \\
\mathsf{D}(\lambda )|\boldsymbol\eta \rangle & =[|q^{-\delta _{\mathsf{N}}}\boldsymbol\eta \rangle
\eta _{\mathsf{D}}^{(+)}\lambda +|q^{\delta _{\mathsf{N}}}\boldsymbol\eta \rangle \frac{%
\eta _{\mathsf{D}}^{(-)}}{\lambda }]b_{\boldsymbol\eta }(\lambda )\,+\sum_{a=1}^{%
\mathsf{N}-1}|q^{-\delta _{a}}\boldsymbol\eta \rangle \prod_{b\neq a}\frac{(\lambda
/\eta _{b}-\eta _{b}/\lambda )}{(\eta _{a}/\eta _{b}-\eta _{b}/\eta _{a})}%
\mathtt{\bar{d}}^{(SOV)}(\eta _{a}),  \label{NJG-BS-ChP-SOV-D-R}
\end{align}%
where $|\boldsymbol\eta \rangle \in \mathcal{R}_{\mathsf{N}}$ is the right $\mathsf{B}$%
-eigenstate corresponding to the generic $\boldsymbol{\eta }\in \mathsf{Z}{_{\mathsf{B}}%
}$. The coefficients $\mathtt{\bar{a}}^{(SOV)}(\eta _{a})$ and $\mathtt{\bar{%
d}}^{(SOV)}(\eta _{a})$ are solutions of the same average $\left( \ref{NJG-BS-ChP-ADaver}\right) $ and quantum determinant:  
\begin{equation}
\mathrm{det_{q}}\mathsf{M}(\eta _{r})\,=\,\mathtt{\bar{d}}^{(SOV)}(\eta _{r})%
\mathtt{\bar{a}}^{(SOV)}(q^{-1}\eta _{r})\,,\quad \forall r=1,\dots ,\mathsf{N}%
-1\,  \label{NJG-BS-ChP-o-addet}
\end{equation} conditions while $\mathsf{C}%
(\lambda )$ is uniquely defined by the quantum determinant relation $(\ref{NJG-BS-ChP-q-det-f})$.

\subsubsection{SOV-decomposition of the identity}

The diagonalizability of the Yang-Baxter generator $\mathsf{B}(\lambda )$\
and the simplicity of its spectrum imply the following spectral
decomposition of the identity $\mathbb{I}$ in terms of the $\mathsf{B}$%
-eigenbasis: 
\begin{equation}
\mathbb{I}\equiv \sum_{{\mathbf{k}}\in \mathbb{Z}_{p}^{\mathsf{N}}}\mu {_{%
\mathbf{k}}}|\boldsymbol{\eta _{\mathbf{k}}}\rangle \langle \boldsymbol{\eta _{\mathbf{k}}}|,
\end{equation}%
where:%
\begin{equation}
\mu {_{\mathbf{k}}}\equiv \langle {\boldsymbol\eta _{\mathbf{k}}}|{\boldsymbol\eta _{\mathbf{k}}}%
\rangle ^{-1}\text{ \ \ }\forall\, {\mathbf{k}}\in \mathbb{Z}_{p}^{\mathsf{N}},
\end{equation}%
is the equivalent of the so-called Sklyanin's measure\footnote{%
Sklyanin's measure has been first introduced by Sklyanin in his article \cite{NJG-BS-ChP-Sk1} on quantum Toda chain, \cite{NJG-BS-ChP-GM}-\cite{NJG-BS-ChP-KL99}; see also \cite{NJG-BS-ChP-Sm98} and 
\cite{NJG-BS-ChP-BT06} for further discussions on the measure in the quantum Toda chain
and in the sinh-Gordon model, respectively.}. The non-Hermitian character of
the operator family $\mathsf{B}(\lambda )$\ clearly implies that, for
generic ${\mathbf{k}}\in \mathbb{Z}_{p}^{\mathsf{N}}$, $\left( |\boldsymbol{\eta _{%
\mathbf{k}}}\rangle \right) ^{\dag }$ and $\langle {\boldsymbol\eta _{\mathbf{k}}}|$
are not proportional covectors in $\mathcal{L}_{\mathsf{N}}$; then, $\mu {_{%
\mathbf{k}}}$ is not a standard positive definite measure in our cyclic
representations. Nevertheless, we will show that the above formula defines a proper orthogonal decomposition of the identity operator. 

Now we compute\footnote{%
Let us recall that this measure has been first derived in \cite{NJG-BS-ChP-GIPST07} for
cyclic representations of Bazhanov-Stroganov model \cite{NJG-BS-ChP-BS,NJG-BS-ChP-BBP,NJG-BS-ChP-B04} through the
recursion in the construction of left and right SOV-basis.} this
\textquotedblleft measure\textquotedblright\ $\mu {_{\mathbf{k}}}$ and we
show that up to an overall constant (i.e. a constant w.r.t. ${\mathbf{k}}\in 
\mathbb{Z}_{p}^{\mathsf{N}}$) it is completely fixed by the given left and
right SOV-representations of the Yang-Baxter algebras when the gauges are
fixed.

\begin{proposition}
The following identities hold:%
\begin{eqnarray}
\langle {\boldsymbol\eta _{\mathbf{k}}}|{\boldsymbol\eta _{\mathbf{h}}}\rangle &=&\langle {\boldsymbol\eta _{%
\mathbf{h}}}|{\boldsymbol\eta _{\mathbf{h}}}\rangle \prod_{j=1}^{\mathsf{N}}\delta
_{k_{i},h_{i}}\text{, \ }\forall\, {\mathbf{k}},{\mathbf{h}}\in \mathbb{Z}%
_{p}^{\mathsf{N}},  \label{NJG-BS-ChP-1M_jj} \\
\mu {_{\mathbf{h}}} &=&\frac{\prod_{1\leq a<b\leq \mathsf{N}-1}(({\eta }%
_{a}^{(h_{a})})^{2}-({\eta }_{b}^{(h_{b})})^{2})}{C_{\mathsf{N}}\prod_{a=1}^{%
\mathsf{N}-1}\omega _{a}({\eta }_{a}^{(h_{a})})},\text{ \ }\forall\, {\mathbf{h%
}}\in \mathbb{Z}_{p}^{\mathsf{N}},  \label{NJG-BS-ChP-2M_jj}
\end{eqnarray}%
where:%
\begin{equation}
\omega _{a}({\eta }_{a}^{(h_{a})})\equiv \left( {\eta }_{a}^{(h_{a})}\right)
^{\mathsf{N}-1}\prod_{l_{a}=1}^{h_{a}}\mathtt{a}^{(SOV)}({\eta }%
_{a}^{(l_{a})})/\mathtt{\bar{a}}^{(SOV)}({\eta }_{a}^{(l_{a}-1)})
\end{equation}%
are gauge dependent parameters and  $C_{\mathsf{N}}$ in the formula for  $\mu {_{\mathbf{h}%
}}$ is a constant w.r.t.\ ${\mathbf{h}}\in \mathbb{Z}_{p}^{\mathsf{N}}$.
Then, the SOV-decomposition of the identity explicitly reads: 
\begin{equation}
\mathbb{I}\equiv \sum_{h_{1},...,h_{\mathsf{N}}=1}^{p}\prod_{1\leq a<b\leq 
\mathsf{N}-1}(({\eta }_{a}^{(h_{a})})^{2}-({\eta }_{b}^{(h_{b})})^{2})\frac{|%
{\eta }_{1}^{(h_{1})},...,{\eta }_{\mathsf{N}}^{(h_{\mathsf{N}})}\rangle
\langle {\eta }_{1}^{(h_{1})},...,{\eta }_{\mathsf{N}}^{(h_{\mathsf{N}})}|}{%
C_{\mathsf{N}}\prod_{b=1}^{\mathsf{N}-1}\omega _{b}({\eta }_{b}^{(h_{b})})},
\label{NJG-BS-ChP-Decomp-Id}
\end{equation}
Note that the constant $C_{\mathsf{N}}$ can be put equal to one by a trivial (constant) gauge transformation that does not affect the functions $\mathtt{a}^{(SOV)}$ and $\mathtt{\bar{a}}^{(SOV)}$. 
\end{proposition}

\begin{proof}
Computing $\langle {\boldsymbol\eta _{\mathbf{k}}}|\mathsf{B}(\lambda )|{\boldsymbol\eta _{\mathbf{%
h}}}\rangle $, we get:%
\begin{equation}
(b_{\boldsymbol{\eta _{\mathbf{k}}}}(\lambda )-b_{\boldsymbol{\eta _{\mathbf{h}}}}(\lambda
))\langle {\boldsymbol\eta _{\mathbf{k}}}|{\boldsymbol\eta _{\mathbf{h}}}\rangle =0\text{ \ \ }%
\forall \lambda \in \mathbb{C},\text{ }\forall\, {\mathbf{k}},{\mathbf{h}}\in 
\mathbb{Z}_{p}^{\mathsf{N}}
\end{equation}%
and then the simplicity of the spectrum of $\mathsf{B}(\lambda )$ implies $%
\left( \ref{NJG-BS-ChP-1M_jj}\right) $. To compute $\mu {_{\mathbf{h}}}$, we compute
the following matrix elements $\theta _{a}\equiv \langle {\eta }%
_{1}^{(h_{1})},...,{\eta }_{a}^{(h_{a}-1)},...,{\eta }_{\mathsf{N}}^{(h_{%
\mathsf{N}})}|\mathsf{A}(\eta _{a}^{(h_{a}-1)})|{\eta }_{1}^{(h_{1})},...,{%
\eta }_{a}^{(h_{a})},...,{\eta }_{\mathsf{N}}^{(h_{\mathsf{N}})}\rangle $,
by using first the left action of $\mathsf{A}(\eta _{a}^{(h_{a}-1)})$, then
the right action of $\mathsf{A}(\eta _{a}^{(h_{a}-1)})$ together with $%
\left( \ref{NJG-BS-ChP-1M_jj}\right) $ and finally equating the two results we get:%
\begin{align}
\frac{\langle {\eta }_{1}^{(h_{1})},...,{\eta }_{a}^{(h_{a})},...,{\eta }_{%
\mathsf{N}}^{(h_{\mathsf{N}})}|{\eta }_{1}^{(h_{1})},...,{\eta }%
_{a}^{(h_{a})},...,{\eta }_{\mathsf{N}}^{(h_{\mathsf{N}})}\rangle }{\langle {%
\eta }_{1}^{(h_{1})},...,{\eta }_{a}^{(h_{a}-1)},...,{\eta }_{\mathsf{N}%
}^{(h_{\mathsf{N}})}|{\eta }_{1}^{(h_{1})},...,{\eta }_{a}^{(h_{a}-1)},...,{%
\eta }_{\mathsf{N}}^{(h_{\mathsf{N}})}\rangle }& =\delta _{a,\mathsf{N}%
}+(1-\delta _{a,\mathsf{N}})\frac{\mathtt{a}^{(SOV)}({\eta }_{a}^{(h_{a})})}{%
\mathtt{\bar{a}}^{(SOV)}({\eta }_{a}^{(h_{a}-1)})}  \notag \\
& \times \prod_{b\neq a,b=1}^{\mathsf{N}-1}\frac{({\eta }_{a}^{(h_{a}-1)}/{%
\eta }_{b}^{(h_{b})}-{\eta }_{b}^{(h_{b})}/{\eta }_{a}^{(h_{a}-1)})}{({\eta }%
_{a}^{(h_{a})}/{\eta }_{b}^{(h_{b})}-{\eta }_{b}^{(h_{b})}/{\eta }%
_{a}^{(h_{a})})},  \label{NJG-BS-ChP-F1}
\end{align}%
from which $\left( \ref{NJG-BS-ChP-2M_jj}\right) $ simply follows.
\end{proof}

\subsection{SOV-characterization of the spectrum}

Let us denote with $\Sigma _{\tau _{2}}$ the set of eigenvalue functions 
$t(\lambda )$ of the transfer matrix $\tau _{2}(\lambda )$. We have then: 
\begin{equation}
\Sigma _{\tau _{2}}\subset \mathbb{C}_{even}[\lambda ,\lambda ^{-1}]_{%
\mathsf{N}}\text{ for }\mathsf{N}\text{ even, \ \ \ }\Sigma _{\tau
_{2}}\subset \mathbb{C}_{odd}[\lambda ,\lambda ^{-1}]_{\mathsf{N}}\text{ for 
}\mathsf{N}\text{ odd},  \label{NJG-BS-ChP-set-t}
\end{equation}%
where $\mathbb{C}_{\epsilon }[x,x^{-1}]_{\mathsf{M}}$ denotes the linear
space in the field $\mathbb{C}$ of the Laurent polynomials of degree $%
\mathsf{M}$ in the variable $x$ which are even or odd as stated in the index 
$\epsilon $. The $\Theta $-charge naturally induces the grading $\Sigma
_{\tau _{2}}=\bigcup_{k=0}^{2l}\Sigma _{\tau _{2}}^{k}$, where: 
\begin{equation}
\Sigma _{\tau _{2}}^{k}\equiv \left\{ t(\lambda )\in \Sigma _{\tau
_{2}}:\lim_{\log \lambda \rightarrow \mp \infty }\lambda ^{\pm \mathsf{N}%
}t(\lambda )=\left( q^{\mp k}a_{\mp }+q^{\pm k}d_{\mp }\right) \right\} .
\end{equation}%
This simply follows from the commutativity of $\tau _{2}(\lambda )$ with $\Theta $ and from
its asymptotics. In particular, any $t_{k}(\lambda )\in
\Sigma _{\tau _{2}}^{k}$ is a $\tau _{2}$-eigenvalue corresponding to
simultaneous eigenstates of ${\tau _{2}}(\lambda )$ and $\Theta $ with $%
\Theta $-eigenvalue $q^{k}$.

\subsubsection{Eigenvalues and wave-funtions}
In the SOV representations, the spectral problem for ${\tau _{2}}(\lambda)$
is reduced to the following discrete system of Baxter-like equations in the
wave-function $\Psi _{t}(\boldsymbol\eta )\equiv \langle \,\boldsymbol\eta \,|\,t\,\rangle $ of a $%
\tau _{2}$-eigenstate $|\,t\,\rangle $: 
\begin{equation}
t(\eta _{r})\Psi_t(\boldsymbol\eta)\,=\,{\tt a}^{(SOV)}(\eta _{r})\Psi_t( q^{-\delta_{r}} \boldsymbol\eta)+{\tt d}^{(SOV)}(\eta _{r})\Psi_t(q^{\delta_{r}} \boldsymbol\eta)\, \qquad \text{ \ }\forall r\in \{1,...,\mathsf{N}-1\},  \label{NJG-BS-ChP-SOVBax1}
\end{equation}%
plus the following equation in the variable $\eta _{\mathsf{N}}$: 
\begin{equation}
\Psi_{t}(q^{\delta_{\mathsf{N}}}\boldsymbol\eta)\,=\,q^{-k}\Psi_{t}(\boldsymbol\eta),  \text{ \ where \ } q^{\pm\delta_{r}} \boldsymbol\eta \equiv(\eta_1,\dots,q^{\pm 1}\eta_r,\dots,\eta_\mathsf{N}), 
\label{NJG-BS-ChP-SOVBax2}
\end{equation}%
for $t(\lambda )\in \Sigma _{{\tau _{2}}}^{k}\ $with $k\in \{0,...,2l\}$. Let us introduce the one parameter family $D(\lambda )$ of $p\times p$
matrix:%
\begin{equation}
D(\lambda )\equiv 
\begin{pmatrix}
t(\lambda ) & -\bar{\text{\textsc{d}}}(\lambda ) & 0 & \cdots  & 0 & -\bar{%
\text{\textsc{a}}}(\lambda ) \\ 
-\bar{\text{\textsc{a}}}(q\lambda ) & t(q\lambda ) & -\bar{\text{\textsc{d}}}%
(q\lambda ) & 0 & \cdots  & 0 \\ 
0 & {\quad }\ddots  &  &  &  & \vdots  \\ 
\vdots  &  & \cdots  &  &  & \vdots  \\ 
\vdots  &  &  & \cdots  &  & \vdots  \\ 
t\vdots  &  &  &  & \ddots {\qquad } & 0 \\ 
0 & \ldots  & 0 & -\bar{\text{\textsc{a}}}(q^{2l-1}\lambda ) & 
t(q^{2l-1}\lambda ) & -\bar{\text{\textsc{d}}}(q^{2l-1}\lambda ) \\ 
-\bar{\text{\textsc{d}}}(q^{2l}\lambda ) & 0 & \ldots  & 0 & -\bar{\text{%
\textsc{a}}}(q^{2l}\lambda ) & t(q^{2l}\lambda )%
\end{pmatrix}
\label{NJG-BS-ChP-D-matrix}
\end{equation}%
then when we make the following choice of gauge for the left
SOV-representation: 
\begin{equation}
\mathtt{a}^{(SOV)}(\lambda )\equiv \bar{\text{\textsc{a}}}(\lambda ),\text{
\ \ \ \ \ \ \ \ \ \ \ \ }\mathtt{d}^{(SOV)}(\lambda )\equiv \bar{\text{%
\textsc{d}}}(\lambda ),  \label{NJG-BS-ChP-L-gauge}
\end{equation}%
it holds:

\begin{theorem}
\label{NJG-BS-ChP-C:T-eigenstates} \textbf{(Theorems 2, 3 and 4 of \cite{NJG-BS-ChP-GN12})} For
almost all the values of the parameters of a Bazhanov-Stroganov representation, the spectrum of $\tau _{2}(\lambda )$ is simple. Moreover:

\begin{itemize}
\item[\textsf{I)}] $\Sigma _{{\tau _{2}}}$ coincides with the set of
functions in (\ref{NJG-BS-ChP-set-t}) which are solutions of the functional equation:%
\begin{equation}
\det_{p}\text{$D$}(\Lambda )=0,\text{ \ \ }\forall \Lambda \in \mathbb{C}.
\label{NJG-BS-ChP-I-Functional-eq}
\end{equation}%
Then, up to an overall normalization, we can fix the $\tau _{2}$-eigenstate
corresponding to $t_{k}(\lambda )\in \Sigma _{{\tau _{2}}}^{k}$ by: 
\begin{equation}
\Psi _{t_{k}}(\boldsymbol\eta )\equiv \langle \,\eta _{1},...,\eta _{\mathsf{N}%
}\,|\,t_{k}\,\rangle =\eta _{\mathsf{N}}^{-k}\prod_{r=1}^{\mathsf{N}%
-1}Q_{t_{k}}(\eta _{r}),  \label{NJG-BS-ChP-Qeigenstate-even}
\end{equation}%
where $Q_{t_{k}}(\lambda )$ is the only solution (up to quasi-constants)
corresponding to $t_{k}(\lambda )$ of the Baxter equation:%
\begin{equation}
t_{k}(\lambda )Q_{t_{k}}(\lambda )=\bar{\text{\textsc{a}}}(\lambda
)Q_{t_{k}}(\lambda /q)+\bar{\text{\textsc{d}}}(\lambda )Q_{t_{k}}(q\lambda ).
\label{NJG-BS-ChP-Baxter-eq-eigenvalues}
\end{equation}

\item[\textsf{II)}] In the self-adjoint representations of the Bazhanov-Stroganov 
model under the further constrains: 
\begin{equation}
\prod_{h=1}^{N}\frac{\alpha _{h}^{\ast }}{\alpha _{h}}=1,\quad \quad \text{
\ \ }\frac{\mathbbm{b}_{n}}{\mathbbm{b}_{n}^{\ast }}=\frac{\mathbbm{a}_{n}}{%
\mathbbm{a}_{n}^{\ast }},\quad \quad \text{ \ \ }\frac{\alpha _{n+1}^{\ast
}\alpha _{n}^{\ast }}{\alpha _{n+1}\alpha _{n}}=\frac{\mathbbm{b}%
_{n+1}^{\ast }\mathbbm{b}_{n}}{\mathbbm{b}_{n+1}\mathbbm{b}_{n}^{\ast }}%
,\quad \text{ \ \ \ \ \ }\forall n\in \{1,...,\mathsf{N}\},  \label{NJG-BS-ChP-S-Adj-c}
\end{equation}the functions $\bar{\text{\textsc{a}}}(\lambda )$ and $\bar{\text{\textsc{d}}%
}(\lambda )$ are gauge equivalent to the Laurent polynomials:%
\begin{equation}
\mathtt{a}(\lambda )\equiv i^{\mathsf{N}}\prod_{n=1}^{\mathsf{N}}\frac{\beta
_{n}}{\lambda }(1-i^{(1+\epsilon )/2}q^{-1/2}\frac{|\alpha _{n}|}{|%
\mathbbm{a}_{n}|}\lambda )(1-i^{(1+\epsilon )/2}q^{-1/2}\frac{|\alpha _{n}|}{%
|\mathbbm{b}_{n}|}\lambda ),\qquad \mathtt{d}(\lambda )\equiv q^{\mathsf{N}}%
\mathtt{a}(-\lambda q),  \label{NJG-BS-ChP-S-adj-coeff}
\end{equation}%
respectively, and for any $t_{k}(\lambda )\in \Sigma _{{\tau _{2}}}^{k}$, we
can construct uniquely up to quasi-constants a $\epsilon $-real polynomial%
\footnote{%
i.e. it satisfies the following complex-conjugation conditions: $\left(
Q_{t}(\lambda )\right) ^{\ast }\equiv Q_{t}(\epsilon \lambda ^{\ast })\text{
\ }\forall \lambda \in \mathbb{C}$.}$^{,}$\footnote{%
Note that $Q_{t}(\lambda )$ has been constructed in terms of the cofactors
of the matrix $D$$(\Lambda )$ in Theorem 3 of \cite{NJG-BS-ChP-GN12}.}:%
\begin{equation}
Q_{t_{k}}(\lambda )=\lambda ^{a_{t_{k}}}\prod_{h=1}^{2l\mathsf{N}%
-(b_{t_{k}}+a_{t_{k}})}(\lambda _{h}-\lambda ),\,\,\,\,\,\,\,\,0\leq
a_{t_{k}}\leq 2l,\,\,0\leq b_{t_{k}}+a_{t_{k}}\leq 2l\mathsf{N},
\label{NJG-BS-ChP-Q_t-definition}
\end{equation}%
which is a solution of the Baxter functional equation (\ref{NJG-BS-ChP-Baxter-eq-eigenvalues}) in the gauge (\ref{NJG-BS-ChP-S-adj-coeff}) and: 
\begin{equation}
a_{t_{k}}=\pm k\,\,\mathsf{mod}\,p,\,\,\,\,\,\,\,\,\,b_{t_{k}}=\pm k\,\,%
\mathsf{mod}\,p\text{.}  \label{NJG-BS-ChP-asymptitics-Q}
\end{equation}
\end{itemize}
\end{theorem}

\subsubsection{Eigenvectors and eigencovectors\label{NJG-BS-ChP-LReigenstates}}

The SOV-decomposition of the identity $\left( \ref{NJG-BS-ChP-Decomp-Id}\right) $ and
the results of the previous subsections imply that the state:%
\begin{equation}
|t_{k}\rangle =\sum_{h_{1},...,h_{\mathsf{N}}=1}^{p}\frac{q^{kh_{\mathsf{N}}}%
}{p^{1/2}}\prod_{a=1}^{\mathsf{N}-1}Q_{t_{k}}({\eta }_{a}^{(h_{a})})\prod_{1%
\leq a<b\leq \mathsf{N}-1}(({\eta }_{a}^{(h_{a})})^{2}-({\eta }%
_{b}^{(h_{b})})^{2})\frac{|{\eta }_{1}^{(h_{1})},...,{\eta }_{\mathsf{N}%
}^{(h_{\mathsf{N}})}\rangle }{\prod_{b=1}^{\mathsf{N}-1}\omega _{b}({\eta }%
_{b}^{(h_{b})})},
\end{equation}%
is, up to an overall normalization, the only right $\tau _{2}$-eigenstate
associated to $t_{k}(\lambda )\in \Sigma _{\mathsf{T}}^{k}$. Here, $%
Q_{t_{k}}(\lambda )$ is the only solution (up to quasi-constants) of the
Baxter equation:%
\begin{equation}
t_{k}(\lambda )Q_{t_{k}}(\lambda )=\bar{\text{\textsc{a}}}(\lambda
)Q_{t_{k}}(\lambda q^{-1})+\bar{\text{\textsc{d}}}(\lambda
)Q_{t_{k}}(\lambda q),  \label{NJG-BS-ChP-EQ-Baxter-R}
\end{equation}%
as defined in Theorem \ref{NJG-BS-ChP-C:T-eigenstates}. Similarly, we can prove that
the state:%
\begin{equation}
\langle t_{k}|=\sum_{h_{1},...,h_{\mathsf{N}}=1}^{p}\frac{q^{kh_{\mathsf{N}}}%
}{p^{1/2}}\prod_{a=1}^{\mathsf{N}-1}\bar{Q}_{t_{k}}({\eta }%
_{a}^{(h_{a})})\prod_{1\leq a<b\leq \mathsf{N}-1}(({\eta }%
_{a}^{(h_{a})})^{2}-({\eta }_{b}^{(h_{b})})^{2})\frac{\langle {\eta }%
_{1}^{(h_{1})},...,{\eta }_{\mathsf{N}}^{(h_{\mathsf{N}})}|}{\prod_{b=1}^{%
\mathsf{N}-1}\omega _{b}({\eta }_{b}^{(h_{b})})},
\end{equation}%
is, up to an overall normalization, the only left $\tau _{2}$-eigenstate
associated to $t_{k}(\lambda )\in \Sigma _{\mathsf{T}}^{k}$. Here, $\bar{Q}%
_{t_{k}}(\lambda )$ is the only solution (up to quasi-constants) of the
Baxter equation:%
\begin{equation}
t_{k}(\lambda )\bar{Q}_{t_{k}}(\lambda )=\bar{\text{\textsc{d}}}(\lambda /q)%
\bar{Q}_{t_{k}}(\lambda /q)+\bar{\text{\textsc{a}}}(\lambda q)\bar{Q}%
_{t_{k}}(\lambda q),  \label{NJG-BS-ChP-EQ-Baxter-L}
\end{equation}%
when we make the following choice of gauge for the right SOV-representation:%
\begin{equation}
\mathtt{\bar{a}}^{(SOV)}(\lambda )\equiv \bar{\text{\textsc{a}}}(\lambda q),%
\text{ \ \ \ \ \ \ \ \ \ \ \ \ }\mathtt{\bar{d}}^{(SOV)}(\lambda )\equiv 
\bar{\text{\textsc{d}}}(\lambda /q).  \label{NJG-BS-ChP-R-gauge}
\end{equation}

\section{The inhomogeneous chiral Potts model}

\subsection{Definitions and first properties\label{NJG-BS-ChP-chP-def}}

The connections between the integrable chiral Potts model and the Bazhanov-Stroganov model restricted to parametrization by points on the algebraic curves $%
\mathcal{C}_{k}$ were first remarked in \cite{NJG-BS-ChP-BS}. We can summarize them as follows: \\
\textsf{I)} the fundamental R-matrix intertwining the Bazhanov-Stroganov Lax operator in the quantum space is given by the product of four chiral
Potts Boltzmann weights; \\
\textsf{II)} the transfer matrix of the chiral
Potts model is a Baxter $\mathsf{Q}$-operator for the Bazhanov-Stroganov model.\\
\noindent
Let us recall here how the spectrum of the inhomogeneous chiral Potts
transfer matrix is characterized by SOV construction thanks to the property 
\textsf{(II)}. The algebraic curve $\mathcal{C}_{k}$ of modulus $k$ is by
definition the locus of the points p $\equiv (a_{\text{p}},b_{\text{p}},c_{%
\text{p}},d_{\text{p}})\in \mathbb{C}^{4}$ which satisfy the equations: 
\begin{equation}
x_{\text{p}}^{p}+y_{\text{p}}^{p}=k(1+x_{\text{p}}^{p}y_{\text{p}}^{p}),%
\text{ \ \ }kx_{\text{p}}^{p}=1-k^{^{\prime }}s_{\text{p}}^{-p},\text{ \ \ }%
ky_{\text{p}}^{p}=1-k^{^{\prime }}s_{\text{p}}^{p}, \label{NJG-BS-ChP-curve-eq}
\end{equation}%
where:%
\begin{equation}
x_{\text{p}}\equiv a_{\text{p}}/d_{\text{p}},\text{ \ }y_{\text{p}}\equiv b_{%
\text{p}}/c_{\text{p}},\text{ \ }s_{\text{p}}\equiv d_{\text{p}}/c_{\text{p}%
},\,t_{\text{p}}\equiv x_{\text{p}}y_{\text{p}},\text{ \ }k^{2}+(k^{^{\prime
}})^{2}=1.
\end{equation}%
Let us introduce the following cyclic dilogarithm functions\footnote{%
They are the Boltzmann weights of the chiral Potts model \cite{NJG-BS-ChP-BaPauY}, see
also \cite{NJG-BS-ChP-BT06,NJG-BS-ChP-FK2}-\cite{NJG-BS-ChP-V2} for the study of the properties of dilogarithm
functions.}; here we use the notation: 
\begin{equation}
\frac{W_{\text{qp}}(z(n))}{W_{\text{qp}}(z(0))}=(\frac{s_{\text{q}}}{s_{%
\text{p}}})^{n}\prod_{k=1}^{n}\frac{y_{\text{p}}-q^{-2k}x_{\text{q}}}{y_{%
\text{q}}-q^{-2k}x_{\text{p}}},\text{ \ \ \ \ }\frac{\bar{W}_{\text{qp}%
}(z(n))}{\bar{W}_{\text{qp}}(z(0))}=(s_{\text{p}}s_{\text{q}%
})^{n}\prod_{k=1}^{n}\frac{q^{-2}x_{\text{q}}-q^{-2k}x_{\text{p}}}{y_{\text{p%
}}-q^{-2k}y_{\text{q}}},  \label{NJG-BS-ChP-W-function}
\end{equation}%
where $z(n)=q^{-2n}\text{,\ }n\in \{0,...,2l\}$. They are solutions of the
following recursion relations:%
\begin{equation}
\frac{W_{\text{qp}}(zq)}{W_{\text{qp}}(zq^{-1})}=-z\frac{s_{\text{p}}}{s_{%
\text{q}}}\frac{x_{\text{p}}}{y_{\text{p}}}q^{-1}\frac{1-\frac{y_{\text{q}}}{%
x_{\text{p}}}qz^{-1}}{1-\frac{x_{\text{q}}}{y_{\text{p}}}q^{-1}z},\text{ \ \ 
}\frac{\bar{W}_{\text{qp}}(zq)}{\bar{W}_{\text{qp}}(zq^{-1})}=-\frac{qz^{-1}%
}{s_{\text{p}}s_{\text{q}}}\frac{y_{\text{p}}}{x_{\text{p}}}\frac{1-\frac{y_{%
\text{q}}}{y_{\text{p}}}q^{-1}z}{1-\frac{x_{\text{q}}}{x_{\text{p}}}%
q^{-1}z^{-1}}.  \label{NJG-BS-ChP-Recursion-w-1}
\end{equation}%
If the points p and q belong to the curves $\mathcal{C}_{k}$, they are well
defined functions of $z\in \mathbb{S}_{p}\equiv \{q^{2n};$ $n=0,...,2l\}$
which satisfy the cyclicity condition:%
\begin{equation}
\frac{\bar{W}_{\text{qp}}(z(p))}{\bar{W}_{\text{qp}}(z(0))}=1,\text{ \ \ \ \ 
}\frac{W_{\text{qp}}(z(p))}{W_{\text{qp}}(z(0))}=1.
\end{equation}%
Then, in the left and right $\mathsf{u}_{n}$-eigenbasis, the transfer matrix 
\textsf{T}$_{\lambda }^{{\small \text{chP}}}$ of the inhomogeneous chiral
Potts model\footnote{%
For a direct comparison see formula (4.12) of \cite{NJG-BS-ChP-Ba08} with the following
identifications:%
\begin{equation*}
z_{j}\equiv q^{2\sigma _{j}^{\prime }},\text{ \ }z_{j}^{\prime }\equiv
q^{2\sigma _{j}}\text{ \ }\forall j\in \{1,...,\mathsf{N}\}.
\end{equation*}%
Note that $\mathsf{T}_{\lambda }^{{\small \text{chP}}}$ is well defined
since the $W$-functions (\ref{NJG-BS-ChP-W-function})\ are cyclic functions of their
arguments.} \cite{NJG-BS-ChP-BS} is characterized by the following kernel:%
\begin{equation}
\mathsf{T}_{\lambda }^{{\small \text{chP}}}(\text{z},\text{z}^{\prime
})\equiv \langle \text{z}|\mathsf{T}_{\lambda }^{{\small \text{chP}}}|\text{z%
}^{\prime }\rangle =\prod_{n=1}^{\mathsf{N}}W_{\text{q}_{n}\text{p}%
}(z_{n}/z_{n}^{\prime })\bar{W}_{\text{r}_{n}\text{p}}(z_{n}/z_{n+1}^{\prime
}),  \label{NJG-BS-ChP-kernel}
\end{equation}%
where:%
\begin{equation}
\lambda =t_{\text{p}}^{-1/2}\mathsf{c}_{0},\text{ \ \ p, r}_{n}\text{, q}%
_{n}\in \mathcal{C}_{k},\mathsf{c}_{0}\in \mathbb{C}\text{.}
\end{equation}%
Let us denote with $\mathcal{R}_{\mathsf{N}}^{{\small \text{chP}}}$ the
sub-variety of the representations defined by the following
parametrization of the Bazhanov-Stroganov Lax operator in terms of points of the
curve:%
\begin{align}
\alpha _{n}& =-b_{\text{q}_{n}}^{2}/\mathsf{c}_{0},\text{ \ \ \ }\mathbbm{b}%
_{n}=-\mathbbm{d}_{n}/q=-a_{\text{q}_{n}}d_{\text{q}_{n}}/q^{3/2},
\label{NJG-BS-ChP-Par-on-chP-inho-1} \\
\beta _{n}& =-\mathsf{c}_{0}d_{\text{q}_{n}}^{2},\text{ \ \ \ \ }\mathbbm{c}%
_{n}=-\mathbbm{a}_{n}q=b_{\text{q}_{n}}c_{\text{q}_{n}}q^{1/2},
\label{NJG-BS-ChP-Par-on-chP-inho-SP}
\end{align}%
and $\text{q}_{n}\in \mathcal{C}_{k}$, $k\in \mathbb{C}$. \textsf{T}$%
_{\lambda }^{{\small \text{chP}}}$ is then a Baxter $\mathsf{Q}$-operator%
\footnote{%
It is worth pointing out that while the Baxter equation (\ref{NJG-BS-ChP-Bax-chP-T-II})
holds in the general inhomogeneous representations the commutativity
properties are proven only under the further restrictions q$_{n}\equiv $ r$%
_{n}$ \ \ $\forall n\{1,...,\mathsf{N}\}$ under which is characterized $%
\mathcal{R}_{\mathsf{N}}^{{\small \text{chP}}}$.} w.r.t. the transfer matrix of the Bazhanov-Stroganov model in $\mathcal{R}_{\mathsf{N}}^{{\small \text{chP}}}$: 
\begin{equation}
\tau _{2}(\lambda )\mathsf{T}_{\lambda }^{{\small \text{chP}}}=a_{\text{BS}%
}(\lambda )\mathsf{T}_{\lambda /q}^{{\small \text{chP}}}+d_{\text{BS}%
}(\lambda )\mathsf{T}_{q\lambda }^{{\small \text{chP}}},
\label{NJG-BS-ChP-Bax-chP-T-II}
\end{equation}%
\begin{equation}
\lbrack \tau _{2}(\lambda ),\mathsf{T}_{\lambda }^{{\small \text{chP}}}]=0,%
\text{ \ \ \ }[\Theta ,\mathsf{T}_{\lambda }^{{\small \text{chP}}}]=0,\text{
\ \ \ }[\mathsf{T}_{\lambda }^{{\small \text{chP}}},\mathsf{T}_{\mu }^{%
{\small \text{chP}}}]=0\text{ \ \ \ }\forall \lambda ,\mu \in \mathbb{C},
\label{NJG-BS-ChP-chP-TII.C.M.R.}
\end{equation}%
with $a_{\text{BS}}$ and $d_{\text{BS}}$ defined in (5.8) and (5.9) of \cite{NJG-BS-ChP-GN12}.

\subsection{SOV-spectrum characterization\label{NJG-BS-ChP-SOV-Char-chP}}

\begin{theorem}
\textbf{(Proposition 3, Theorem 5 and Lemma 13 of \cite{NJG-BS-ChP-GN12})} For almost
all the representations in $\mathcal{R}_{\mathsf{N}}^{{\small \text{chP}}}$
the spectrum of the chiral Potts transfer matrix \textsf{T}$_{\lambda }^{%
{\small \text{chP}}}$ is simple. Moreover:

\textsf{I)} All right and left eigenstates of the chiral Potts transfer
matrix \textsf{T}$_{\lambda }^{{\small \text{chP}}}$ are eigenstates of $%
\tau _{2}(\lambda )$ and they admit the SOV construction presented in point 
\textsf{I)} of Theorem \ref{NJG-BS-ChP-C:T-eigenstates}. The solution $Q_{t}(\lambda )$
of the functional Baxter equation $(\ref{NJG-BS-ChP-Baxter-eq-eigenvalues})$ is gauge
equivalent to the corresponding \textsf{T}$_{\lambda }^{{\small \text{chP}}}$%
-eigenvalue \textsf{q}$_{\lambda }^{{\small \text{chP}}}$ being the
coefficients $a_{\text{BS}}(\lambda )$ and $d_{\text{BS}}(\lambda )$ of $(%
\ref{NJG-BS-ChP-Bax-chP-T-II})$ gauge equivalent to the SOV-ones:%
\begin{equation}
a_{\text{BS}}(\lambda )=h_{\text{BS}}(\lambda )\bar{\textsc{a}}(\lambda )%
\text{ \ \ \ \ \ }d_{\text{BS}}(\lambda )=h_{\text{BS}}^{-1}(\lambda q)\bar{%
\textsc{d}}(\lambda ).
\end{equation}%
Here $h_{\text{BS}}(\lambda )$ is a function whose average value is 1 for
any $\lambda \in \mathbb{C}$.

\textsf{II)} In the sub-variety $\mathcal{R}_{\mathsf{N}}^{{\small \text{%
chP,S-adj}}}\equiv \mathcal{R}_{\mathsf{N}}^{{\small \text{chP}}}\cap 
\mathcal{R}_{\mathsf{N}}^{{\small \text{S-adj}}}$, characterized by (\ref{NJG-BS-ChP-Par-on-chP-inho-1})-(\ref{NJG-BS-ChP-Par-on-chP-inho-SP}) under the following
constrains:%
\begin{equation}
\text{q}_{n}=(a_{\text{q}_{n}},\epsilon q\epsilon _{0,n}a_{\text{q}%
_{n}}^{\ast },\epsilon _{0,n}d_{\text{q}_{n}}^{\ast },d_{\text{q}_{n}})\in 
\mathcal{C}_{k},\ \ \ \epsilon _{0,n}=\pm 1,\ \ \ k^{\ast }=\epsilon k,
\label{NJG-BS-ChP-q-S-adj}
\end{equation}%
the operator \textsf{T}$_{\lambda }^{{\small \text{chP}}}$ is normal and $%
\tau _{2}(\lambda )$ is self-adjoint. Then, point \textsf{I)} of Theorem \ref{NJG-BS-ChP-C:T-eigenstates} allows to construct the full simultaneous (\textsf{T}$%
_{\lambda }^{{\small \text{chP}}},\tau _{2}(\lambda ),\Theta $)-eigenbasis
associating to any $t(\lambda )\in \Sigma _{\tau _{2}}$ the corresponding
eigenstate.
\end{theorem}

\section{Decomposition of the identity in the transfer matrix eigenbasis}

\subsection{Action of left separate states on right separate states}

Here we compute the action of covectors on vectors which in the left and
right SOV-basis have a \textit{separate form} similar to that of the
transfer matrix eigenstates. To be more precise, let us give the
following definition of a left $\langle \alpha _{k}|$ and a right $|\beta
_{k}\rangle $ separate states characterized by the given arbitrary set of functions $\alpha _{a}$ and $\beta _{a}$:%
\begin{align}
\langle \alpha _{k}|& =\sum_{h_{1},...,h_{\mathsf{N}}=1}^{p}\frac{q^{kh_{%
\mathsf{N}}}}{p^{1/2}}\prod_{a=1}^{\mathsf{N}-1}\alpha _{a}({\eta }%
_{a}^{(h_{a})})\prod_{1\leq a<b\leq \mathsf{N}-1}(\left( {\eta }%
_{a}^{(h_{a})}\right) ^{2}-\left( {\eta }_{b}^{(h_{b})}\right) ^{2})\frac{%
\langle {\eta }_{1}^{(h_{1})},...,{\eta }_{\mathsf{N}}^{(h_{\mathsf{N}})}|}{%
\prod_{b=1}^{\mathsf{N}-1}\omega _{b}({\eta }_{b}^{(h_{b})})},
\label{NJG-BS-ChP-Fact-left-SOV} \\
|\beta _{k}\rangle & =\sum_{h_{1},...,h_{\mathsf{N}}=1}^{p}\frac{q^{-kh_{%
\mathsf{N}}}}{p^{1/2}}\prod_{a=1}^{\mathsf{N}-1}\beta _{a}({\eta }%
_{a}^{(h_{a})})\prod_{1\leq a<b\leq \mathsf{N}-1}(\left( {\eta }%
_{a}^{(h_{a})}\right) ^{2}-\left( {\eta }_{b}^{(h_{b})}\right) ^{2})\frac{|{%
\eta }_{1}^{(h_{1})},...,{\eta }_{\mathsf{N}}^{(h_{\mathsf{N}})}\rangle }{%
\prod_{b=1}^{\mathsf{N}-1}\omega _{b}({\eta }_{b}^{(h_{b})})}.
\label{NJG-BS-ChP-Fact-right-SOV}
\end{align}

\begin{proposition}
The action of the left separate state $\langle \alpha _{k}|$ of form $(\ref{NJG-BS-ChP-Fact-left-SOV})$ on the right separate state $|\beta _{h}\rangle $ of form $(\ref{NJG-BS-ChP-Fact-right-SOV})$ reads:%
\begin{equation}
\langle \alpha _{k}|\beta _{h}\rangle =\delta _{k,h}\det_{\mathsf{N}-1}||%
\mathcal{M}_{a,b}^{\left( \alpha ,\beta \right) }||\text{ \ \ with \ }%
\mathcal{M}_{a,b}^{\left( \alpha ,\beta \right) }\equiv \left( {\eta }%
_{a}^{(0)}\right) ^{2(b-1)}\sum_{h=1}^{p}\frac{\alpha _{a}({\eta }%
_{a}^{(h)})\beta _{a}({\eta }_{a}^{(h)})}{\omega _{a}({\eta }_{a}^{(h)})}%
q^{2(b-1)h}.
\end{equation}
\end{proposition}

\begin{proof}
The SOV-decomposition of these states implies:
\begin{equation}
\langle \alpha _{k}|\beta _{h}\rangle =\sum_{h_{\mathsf{N}}=1}^{p}\frac{%
q^{(k-h)h_{\mathsf{N}}}}{p}\sum_{h_{1},...,h_{\mathsf{N}-1}=1}^{p}V(({\eta }%
_{1}^{(h_{1})})^{2},...,({\eta }_{\mathsf{N}-1}^{(h_{\mathsf{N}%
-1})})^{2})\prod_{a=1}^{\mathsf{N}-1}\frac{\alpha _{a}({\eta }%
_{a}^{(h_{a})})\beta _{a}({\eta }_{a}^{(h_{a})})}{\omega _{a}({\eta }%
_{a}^{(h_{a})})},
\end{equation}%
where $V(x_{1},...,x_{\mathsf{N}})\equiv \prod_{1\leq a<b\leq \mathsf{N}%
-1}(x_{a}-x_{b})$ is the Vandermonde determinant. Then from the identity:%
\begin{equation}
\delta _{k,h}=\sum_{h_{\mathsf{N}}=1}^{p}\frac{q^{(k-h)h_{\mathsf{N}}}}{p}%
\text{ \ when }q\text{ is a }p\text{-root of unit and }h,k\in \mathbb{Z}_{p}
\end{equation}%
and by using the multilinearity of the determinant w.r.t. the rows we prove
the proposition.
\end{proof}

It is worth remarking that the previous determinant formulae define also
scalar products for vectors in $\mathcal{R}_{\mathsf{N}}$ which have a
separate form in the right $\mathsf{B}$-eigenbasis and in the dual of the
left $\mathsf{B}$-eigenbasis. Indeed, $\left( \langle \alpha _{k}|\right)
^{\dagger }\in \mathcal{R}_{\mathsf{N}}$ is a separate vector\ in the basis
of $\mathcal{R}_{\mathsf{N}}$ formed out of the $\left( \langle {\eta _{%
\mathbf{k}}}|\right) ^{\dagger }$\ dual states of the left $\mathsf{B}$%
-eigenbasis. Then these results represent the SOV analogue of
the scalar product formulae \cite{NJG-BS-ChP-Sl1,NJG-BS-ChP-KMT99} computed for Bethe
states in the framework of the algebraic Bethe ansatz. Note that this formula is not
restricted to the case in which one of the two states is an eigenstate of
the transfer matrix. It is also interesting to remark that the previous
scalar product formulae allow to prove directly, as in the case of the sine-Gordon model, that the action of a transfer matrix
eigencovector on an eigenvector corresponding to different eigenvalue is
zero.

\begin{corollary}
Let $t_{h}(\lambda )$ and $t_{h}^{\prime }(\lambda )\in \Sigma _{\tau
_{2}}^{h}$ and $\langle t_{h}|$ and $|t_{h}^{\prime }\rangle $ the $\tau
_{2} $-eigenstates defined in Section $\ref{NJG-BS-ChP-LReigenstates}$, then for $%
t_{h}(\lambda )\neq t_{h}^{\prime }(\lambda )$ the $\mathsf{N}\times \mathsf{%
N}$ matrix $\mathcal{M}_{a,b}^{\left( t_{h},t_{h}^{\prime }\right) }$
has rank equal or smaller than $\mathsf{N}-1$. Indeed, the non-zero $\mathsf{%
N}\times 1$ vector V$^{\left( t_{h},t_{h}^{\prime }\right) }$ defined by:%
\begin{equation}
\text{V}_{b}^{\left( t_{h},t_{h}^{\prime }\right) }\equiv c_{b}^{\prime
}-c_{b}\text{\ \ \ }\forall b\in \{1,...,\mathsf{N}\},  \label{NJG-BS-ChP-V-vector}
\end{equation}%
where:
\begin{eqnarray}
t_{h}(\lambda ) &=&\sum_{\epsilon =\pm 1}\left( q^{\epsilon h}a_{\epsilon
}+q^{-\epsilon h}d_{\epsilon }\right) \lambda ^{\epsilon \mathsf{N}%
}+\sum_{b=1}^{\mathsf{N}-1}c_{b}\lambda ^{-\mathsf{N}-2+2b},
\label{NJG-BS-ChP-t-decomp} \\
t_{h}^{\prime }(\lambda ) &=&\sum_{\epsilon =\pm 1}\left( q^{\epsilon
h}a_{\epsilon }+q^{-\epsilon h}d_{\epsilon }\right) \lambda ^{\epsilon 
\mathsf{N}}+\sum_{b=1}^{\mathsf{N}-1}c_{b}^{\prime }\lambda ^{-\mathsf{N}%
-2+2b},  \label{NJG-BS-ChP-t1-decomp}
\end{eqnarray}%
is an eigenvector of $||\mathcal{M}_{a,b}^{\left( t_{h},t_{h}^{\prime
}\right) }||$ corresponding to the eigenvalue zero.
\end{corollary}

\begin{proof}
Note that under the choice (\ref{NJG-BS-ChP-L-gauge}) for the left gauge and (\ref{NJG-BS-ChP-R-gauge}) for the right gauge, it holds:%
\begin{equation}
\omega _{a}({\eta }_{a}^{(h)})=({\eta }_{a}^{(h)})^{\mathsf{N}-2},
\end{equation}%
and then by the definitions $\left( \ref{NJG-BS-ChP-V-vector}\right) $, $\left( \ref{NJG-BS-ChP-t-decomp}\right) $ and $\left(\ref{NJG-BS-ChP-t1-decomp}\right) $ it holds:%
\begin{equation}
\sum_{b=1}^{\mathsf{N}}\mathcal{M}_{a,b}^{\left( t_{h},t_{h}^{\prime
}\right) }\text{V}_{b}^{\left( t_{h},t_{h}^{\prime }\right)
}=\sum_{h=0}^{2s_{a}}Q_{t_{h}^{\prime }}(\eta _{a}^{(h)})\bar{Q}%
_{t_{h}}(\eta _{a}^{(h)})(t_{h}^{\prime }(\eta _{a}^{(h)})-t_{h}(\eta
_{a}^{(h)})).  \label{NJG-BS-ChP-zero-eigenvector-1a}
\end{equation}%
The desired result:%
\begin{equation}
\sum_{b=1}^{\mathsf{N}}\Phi _{a,b}^{\left( t_{h},t_{h}^{\prime }\right) }%
\text{V}_{b}^{\left( t_{h},t_{h}^{\prime }\right) }=0\text{ \ \ \ \ }\forall
a\in \{1,...,\mathsf{N}\},  \label{NJG-BS-ChP-zero-eigenvector}
\end{equation}%
then follows as the Baxter equations $(\ref{NJG-BS-ChP-EQ-Baxter-R})$ and $(\ref{NJG-BS-ChP-EQ-Baxter-L})$ allow to write:%
\begin{align}
Q_{t_{h}^{\prime }}(\eta _{a}^{(k)})\bar{Q}_{t_{h}}(\eta
_{a}^{(k)})(t_{h}^{\prime }(\eta _{a}^{(k)})-t_{h}(\eta _{a}^{(k)}))& =(\bar{%
\text{\textsc{d}}}(\eta _{a}^{(k+1)})Q_{t_{h}^{\prime }}(\eta _{a}^{(k+1)})+%
\bar{\text{\textsc{a}}}(\eta _{a}^{(k-1)})Q_{t_{h}^{\prime }}(\eta
_{a}^{(k-1)}))\bar{Q}_{t}(\eta _{a}^{(k)})  \notag \\
& -(\bar{\text{\textsc{a}}}(\eta _{a}^{(k)})\bar{Q}_{t_{h}}(\eta
_{a}^{(k+1)})+\bar{\text{\textsc{d}}}(\eta _{a}^{(k)})\bar{Q}_{t_{h}}(\eta
_{a}^{(k-1)}))Q_{t_{h}^{\prime }}(\eta _{a}^{(k)}),
\end{align}%
which substituted in (\ref{NJG-BS-ChP-zero-eigenvector-1a}) implies (\ref{NJG-BS-ChP-zero-eigenvector}).
\end{proof}

\subsection{Decomposition of the identity in transfer matrix eigenbasis}

In the representations for which $\tau _{2}(\lambda )$ is diagonalizable
then the simplicity of its spectrum plus the explicit characterizations of
its left and right eigenstates allows to write the following decomposition of
the identity: 
\begin{equation}
\mathbb{I=}\sum_{k=0}^{p-1}\sum_{t(\lambda )\in \Sigma _{\tau _{2}}^{k}}%
\frac{|t_{k}\rangle \langle t_{k}|}{\langle t_{k}|t_{k}\rangle },
\label{NJG-BS-ChP-Id-decomp}
\end{equation}%
where%
\begin{equation}
\langle t_{k}|t_{k}\rangle =\det_{\mathsf{N}-1}||\mathcal{M}_{a,b}^{\left(
t_{k},t_{k}\right) }||\text{ \ with }\mathcal{M}_{a,b}^{\left(
t_{k},t_{k}\right) }\equiv (\eta _{a}^{(0)})^{2(b-1)}\sum_{c=1}^{p}\frac{%
Q_{t_{k}}(\eta _{a}^{(c)})\bar{Q}_{t_{k}}(\eta _{a}^{(c)})}{\omega _{a}(\eta
_{a}^{(c)})}q^{2(b-1)c},
\end{equation}%
is the action of the covector $\langle t_{k}|$ on the vector $|t_{k}\rangle $,
both defined in Section \ref{NJG-BS-ChP-LReigenstates}. Note that in the representations
which define a normal $\tau _{2}(\lambda )$ the simplicity of the spectrum
implies the following identity:%
\begin{equation}
\left( |t_{k}\rangle \right) ^{\dagger }\equiv \alpha _{t_{k}}\langle t_{k}|%
\text{ \ where }\alpha _{t_{k}}=\frac{\left\Vert |t_{k}\rangle \right\Vert
^{2}}{\langle t_{k}|t_{k}\rangle }\in \mathbb{C}
\end{equation}%
for any eigenvector $|t_{k}\rangle $ of $\tau _{2}(\lambda )$. For these
special representations, this stresses the interest in computing the norm 
$\left\Vert |t_{k}\rangle \right\Vert $ as it allows to write left and right 
$\tau _{2}$-eigenstates as one the exact dual of the other.

\section{Propagator for the Bazhanov-Stroganov model}

In this section we construct the propagator operator along the chain of the Bazhanov-Stroganov model for the representations parametrized by points on the chP curves.

\subsection{Fundamental R-matrix of the Bazhanov-Stroganov model}

In the next proposition we report adapting to our notations a fundamental
result of the paper \cite{NJG-BS-ChP-BS}.

\begin{proposition}[\protect\cite{NJG-BS-ChP-BS}]
Let $\mathsf{S}_{(\text{q}_{1},\text{r}_{1}|\text{q}_{2},\text{r}_{2})}$ be
the operator defined on the tensor product of two $p$-dimensional spaces by:%
\begin{equation}
\langle z_{1},z_{2}|\mathsf{S}_{(\text{q}_{1},\text{r}_{1}|\text{q}_{2},%
\text{r}_{2})}|z_{1}^{\prime },z_{2}^{\prime }\rangle \equiv \bar{W}_{\text{q%
}_{2}\text{q}_{1}}(z_{1}/z_{2}^{\prime })W_{\text{r}_{2}\text{q}%
_{1}}(z_{1}^{\prime }/z_{2}^{\prime })\bar{W}_{\text{r}_{2}\text{r}%
_{1}}(z_{2}/z_{1}^{\prime })W_{\text{q}_{2}\text{r}_{1}}(z_{2}/z_{1}),
\end{equation}%
Then, $\mathsf{S}_{(\text{q}_{1},\text{r}_{1}|\text{q}_{2},\text{r}_{2})}$
is the fundamental R-matrix intertwining the Bazhanov-Stroganov  Lax operator in the
quantum space, i.e. it holds:%
\begin{equation}
\mathsf{L}_{0\mathsf{2}}(\lambda |\text{q}_{2},\text{r}_{2})\mathsf{L}_{0%
\mathsf{1}}(\lambda |\text{q}_{1},\text{r}_{1})\mathsf{S}_{(\text{q}_{1},%
\text{r}_{1}|\text{q}_{2},\text{r}_{2})}=\mathsf{S}_{(\text{q}_{1},\text{r}%
_{1}|\text{q}_{2},\text{r}_{2})}\mathsf{L}_{0\mathsf{1}}(\lambda |\text{q}%
_{1},\text{r}_{1})\mathsf{L}_{0\mathsf{2}}(\lambda |\text{q}_{2},\text{r}%
_{2}).
\end{equation}
\end{proposition}

\begin{proof}
Let us just point out that the proof can be obtained by proving it for any
matrix element ($i_{1},i_{2}$)$\in \{1,2\}\times \{1,2\}$. Indeed, taking
the matrix elements on the quantum states $\langle z_{1},z_{2}|$ and $%
|z_{1}^{\prime \prime },z_{2}^{\prime \prime }\rangle $, the proposition
simply follows from the identities: 
\begin{align}
& \sum_{z_{2}^{\prime },z_{2}^{\prime }\in \mathbb{S}_{p},j=1,2}\left( 
\mathsf{L}_{0\mathsf{2}}\right) _{z_{2}z_{2}^{\prime }}^{i_{2},j}(\lambda |%
\text{q}_{2},\text{r}_{2})\left( \mathsf{L}_{0\mathsf{1}}\right)
_{z_{1}z_{1}^{\prime }}^{j,i_{1}}(\lambda |\text{q}_{1},\text{r}_{1})\langle
z_{1}^{\prime },z_{2}^{\prime }|\mathsf{S}_{(\text{q}_{1},\text{r}_{1}|\text{%
q}_{2},\text{r}_{2})}|z_{1}^{\prime \prime },z_{2}^{\prime \prime }\rangle
\left. =\right.  \notag \\
& \text{ \ \ \ \ \ \ \ \ \ \ \ \ \ \ \ \ \ \ \ \ \ \ \ \ \ \ \ \ \ }%
\sum_{z_{2}^{\prime },z_{2}^{\prime }\in \mathbb{S}_{p},j=1,2}\langle
z_{1},z_{2}|\mathsf{S}_{(\text{q}_{1},\text{r}_{1}|\text{q}_{2},\text{r}%
_{2})}|z_{1}^{\prime },z_{2}^{\prime }\rangle \left( \mathsf{L}_{0\mathsf{1}%
}\right) _{z_{1}^{\prime }z_{1}^{\prime \prime }}^{i_{1},j}(\lambda |\text{q}%
_{1},\text{r}_{1})\left( \mathsf{L}_{0\mathsf{2}}\right) _{z_{2}^{\prime
}z_{2}^{\prime \prime }}^{j,i_{2}}(\lambda |\text{q}_{2},\text{r}_{2}),
\label{NJG-BS-ChP-LLS=SLL}
\end{align}%
once the elements of $\mathsf{L}_{0\mathsf{i}}$ are rewritten in terms of
the points of $\mathcal{C}_{k}$ and we use the definition of the functions $%
W $ and $\bar{W}$.
\end{proof}

\subsection{Propagator for the Bazhanov-Stroganov model}

The first transfer matrix of the chP-model has been defined in (\ref{NJG-BS-ChP-kernel}%
) while the second chP-transfer matrix reads:%
\begin{equation}
\mathsf{\hat{T}}_{\lambda _{\text{p}},(\text{p$|$}\{\text{q}_{n},\text{r}%
_{n}\})}^{{\small \text{chP}}}(\textbf{z},\textbf{z}^{\prime })\equiv \langle 
\textbf{z}|\mathsf{\hat{T}}_{\lambda _{\text{p}},(\text{p$|$}\{\text{q}_{n},%
\text{r}_{n}\})}^{{\small \text{chP}}}|\textbf{z}^{\prime }\rangle
=\prod_{n=1}^{N}W_{\text{r}_{n}\text{p}}(z_{n+1}/z_{n}^{\prime })\bar{W}_{%
\text{q}_{n}\text{p}}(z_{n}/z_{n}^{\prime }).
\end{equation}
Let us recall that the propagator operator $\mathsf{U}_{n}$ along the Bazhanov-Stroganov chain is defined by:%
\begin{equation}
\mathsf{U}_{n}\mathsf{M}_{1,...,\mathsf{N}}(\lambda )\mathsf{U}%
_{n}^{-1}\equiv \mathsf{M}_{n,...,\mathsf{N},1,...,n-1}(\lambda )\equiv 
\mathsf{L}_{n-1}(\lambda )\cdots \mathsf{L}_{1}(\lambda )\mathsf{L}_{\mathsf{%
N}}(\lambda )\cdots \mathsf{L}_{n}(\lambda ),  \label{NJG-BS-ChP-Def-Un}
\end{equation}%
then we can prove:

\begin{proposition}
The propagator operator $\mathsf{U}_{m}$ has the following representation in
terms of the chP-transfer matrices:%
\begin{equation}
\mathsf{U}_{m}^{-1}\equiv \mathsf{T}_{\lambda _{\text{r}_{1}},(\text{r}_{1}%
\text{$|$}\{\text{q}_{n},\text{r}_{n}\})}^{{\small \text{chP}}}\mathsf{\hat{T%
}}_{\lambda _{\text{q}_{1}},(\text{q}_{1}\text{$|$}\{\text{q}_{n},\text{r}%
_{n}\})}^{{\small \text{chP}}}\cdots \mathsf{T}_{\lambda _{\text{r}_{m-1}},(%
\text{r}_{m-1}\text{$|$}\{\text{q}_{n},\text{r}_{n}\})}^{{\small \text{chP}}}%
\mathsf{\hat{T}}_{\lambda _{\text{q}_{m-1}},(\text{q}_{m-1}\text{$|$}\{\text{%
q}_{n},\text{r}_{n}\})}^{{\small \text{chP}}}.
\end{equation}
\end{proposition}

\begin{proof}
The previous proposition implies that the operator $\mathsf{S}_{(\text{q}%
_{1},\text{r}_{1}|\text{q}_{2},\text{r}_{2})}$ satisfies the following
equation%
\begin{equation}
\left( \mathsf{S}_{(\text{q}_{1},\text{r}_{1}|\text{q}_{2},\text{r}%
_{2})}\right) ^{-1}\mathsf{L}_{0\mathsf{2}}(\lambda |\text{q}_{2},\text{r}%
_{2})\mathsf{L}_{0\mathsf{1}}(\lambda |\text{q}_{1},\text{r}_{1})\mathsf{S}%
_{(\text{q}_{1},\text{r}_{1}|\text{q}_{2},\text{r}_{2})}=\mathsf{L}_{0%
\mathsf{1}}(\lambda |\text{q}_{1},\text{r}_{1})\mathsf{L}_{0\mathsf{2}%
}(\lambda |\text{q}_{2},\text{r}_{2}),
\end{equation}%
then it is simple to verify that: 
\begin{align}
& \left( \mathsf{S}_{(\text{q}_{1},\text{r}_{1}|\text{q}_{2},\text{r}_{2})}%
\mathsf{S}_{(\text{q}_{1},\text{r}_{1}|\text{q}_{3},\text{r}_{3})}\cdots 
\mathsf{S}_{(\text{q}_{1},\text{r}_{1}|\text{q}_{\mathsf{N}},\text{r}_{%
\mathsf{N}})}\right) ^{-1}\mathsf{L}_{0\mathsf{N}}(\lambda |\text{q}_{%
\mathsf{N}},\text{r}_{\mathsf{N}})\cdots \mathsf{L}_{0\mathsf{2}}(\lambda |%
\text{q}_{\mathsf{2}},\text{r}_{\mathsf{2}})\mathsf{L}_{0\mathsf{1}}(\lambda
|\text{q}_{\mathsf{1}},\text{r}_{\mathsf{1}})  \notag \\
& \left. \times \right. \left( \mathsf{S}_{(\text{q}_{1},\text{r}_{1}|\text{q%
}_{2},\text{r}_{2})}\mathsf{S}_{(\text{q}_{1},\text{r}_{1}|\text{q}_{3},%
\text{r}_{3})}\cdots \mathsf{S}_{(\text{q}_{1},\text{r}_{1}|\text{q}_{%
\mathsf{N}},\text{r}_{\mathsf{N}})}\right) \left. =\right. \mathsf{L}_{0%
\mathsf{1}}(\lambda |\text{q}_{\mathsf{1}},\text{r}_{\mathsf{1}})\mathsf{L}%
_{0\mathsf{N}}(\lambda |\text{q}_{\mathsf{N}},\text{r}_{\mathsf{N}})\cdots 
\mathsf{L}_{0\mathsf{2}}(\lambda |\text{q}_{\mathsf{2}},\text{r}_{\mathsf{2}%
})  \label{NJG-BS-ChP-Prop-Step2}
\end{align}%
Let us compute the matrix elements:
\begin{equation}
\langle \textbf{z}|\mathsf{T}_{\lambda _{\text{r}_{1}},(\text{r}_{1}\text{$|$}\{\text{%
q}_{n},\text{r}_{n}\})}^{{\small \text{chP}}}\mathsf{\hat{T}}_{\lambda _{%
\text{q}_{1}},(\text{q}_{1}\text{$|$}\{\text{q}_{n},\text{r}_{n}\})}^{%
{\small \text{chP}}}|\textbf{z}^{\prime \prime }\rangle =\sum_{\textbf{z}^{\prime }}\langle \textbf{z}|%
\mathsf{T}_{\lambda _{\text{r}_{1}},(\text{r}_{1}\text{$|$}\{\text{q}_{n},%
\text{r}_{n}\})}^{{\small \text{chP}}}|\textbf{z}^{\prime }\rangle \langle \textbf{z}^{\prime
}|\mathsf{\hat{T}}_{\lambda _{\text{q}_{1}},(\text{q}_{1}\text{$|$}\{\text{q}%
_{n},\text{r}_{n}\})}^{{\small \text{chP}}}|\textbf{z}^{\prime \prime }\rangle
\end{equation}%
Using the relations $\bar{W}_{pp}(z/z^{\prime })=\delta _{z,z^{\prime }}$
and $W_{pq}(z)W_{qp}(z)=1$, we get:
\begin{align}
\langle \textbf{z}|\mathsf{T}_{\lambda _{\text{r}_{1}},(\text{r}_{1}\text{$|$}\{\text{%
q}_{n},\text{r}_{n}\})}^{{\small \text{chP}}}\mathsf{\hat{T}}_{\lambda _{%
\text{q}_{1}},(\text{q}_{1}\text{$|$}\{\text{q}_{n},\text{r}_{n}\})}^{%
{\small \text{chP}}}|\textbf{z}^{\prime \prime }\rangle &= & \sum_{\textbf{z}^{\prime }}\delta
_{z_{1},z_{2}^{\prime }}\delta _{z_{1}^{\prime },z_{1}^{\prime \prime
}}\prod_{n\geq 2}\langle z_{n}^{\prime },z_{n}|\mathsf{S}_{(\text{q}_{1},%
\text{r}_{1}|\text{q}_{n},\text{r}_{n})}|z_{n+1}^{\prime },z_{n}^{\prime
\prime }\rangle \\
&= & \langle z_{1},\dots ,z_{\mathsf{N}}|\mathsf{S}_{(\text{q}_{1},\text{r}%
_{1}|\text{q}_{2},\text{r}_{2})}\cdots \mathsf{S}_{(\text{q}_{1},\text{r}%
_{1}|\text{q}_{\mathsf{N}},\text{r}_{\mathsf{N}})}|z_{1}^{\prime \prime
},\dots ,z_{\mathsf{N}}^{\prime \prime }\rangle
\end{align}%
Let us use the notation $\bar{\mathsf{S}}_{i}=\mathsf{T}_{\lambda _{\text{r}_{i}},(%
\text{r}_{i}\text{$|$}\{\text{q}_{n},\text{r}_{n}\})}^{{\small \text{chP}}}%
\mathsf{\hat{T}}_{\lambda _{\text{q}_{i}},(\text{q}_{i}\text{$|$}\{\text{q}%
_{n},\text{r}_{n}\})}^{{\small \text{chP}}}$, then $\left( \ref{NJG-BS-ChP-Prop-Step2}%
\right) $ can be rewritten as it follows:%
\begin{equation}
\bar{\mathsf{S}}_{1}^{-1}\mathsf{L}_{0\mathsf{N}}(\lambda |\text{q}_{\mathsf{%
N}},\text{r}_{\mathsf{N}})\cdots \mathsf{L}_{0\mathsf{2}}(\lambda |\text{q}_{%
\mathsf{2}},\text{r}_{\mathsf{2}})\mathsf{L}_{0\mathsf{1}}(\lambda |\text{q}%
_{\mathsf{1}},\text{r}_{\mathsf{1}})\bar{\mathsf{S}}_{1}=\mathsf{L}_{0%
\mathsf{1}}(\lambda |\text{q}_{\mathsf{1}},\text{r}_{\mathsf{1}})\mathsf{L}%
_{0\mathsf{N}}(\lambda |\text{q}_{\mathsf{N}},\text{r}_{\mathsf{N}})\cdots 
\mathsf{L}_{0\mathsf{2}}(\lambda |\text{q}_{\mathsf{2}},\text{r}_{\mathsf{2}%
})
\end{equation}%
and acting similarly with the others $\bar{\mathsf{S}}_{n}$ with $n>1$ it
holds: 
\begin{align}
& \bar{\mathsf{S}}_{n}^{-1}\mathsf{L}_{0n-1}(\lambda |\text{q}_{n-1},\text{r}%
_{n-1})\cdots \mathsf{L}_{0\mathsf{1}}(\lambda |\text{q}_{\mathsf{1}},\text{r%
}_{\mathsf{1}})\mathsf{L}_{0\mathsf{N}}(\lambda |\text{q}_{\mathsf{N}},\text{%
r}_{\mathsf{N}})\cdots \mathsf{L}_{0n+1}(\lambda |\text{q}_{n+1},\text{r}%
_{n+1})\mathsf{L}_{0n}(\lambda |\text{q}_{n},\text{r}_{n})\bar{\mathsf{S}}%
_{n}  \notag \\
& =\mathsf{L}_{0n}(\lambda |\text{q}_{n},\text{r}_{n})\cdots \mathsf{L}_{0%
\mathsf{1}}(\lambda |\text{q}_{\mathsf{1}},\text{r}_{\mathsf{1}})\mathsf{L}%
_{0\mathsf{N}}(\lambda |\text{q}_{\mathsf{N}},\text{r}_{\mathsf{N}})\cdots 
\mathsf{L}_{0n+2}(\lambda |\text{q}_{n+2},\text{r}_{n+2})\mathsf{L}%
_{0n+1}(\lambda |\text{q}_{n+1},\text{r}_{n+1}),
\end{align}%
from which defining:%
\begin{eqnarray}
\mathsf{U}_{n}^{-1} &=&\bar{\mathsf{S}}_{1}\bar{\mathsf{S}}_{2}\dots \bar{%
\mathsf{S}}_{n-1} \\
&=&\mathsf{T}_{\lambda _{\text{r}_{1}},(\text{r}_{1}\text{$|$}\{\text{q}_{n},%
\text{r}_{n}\})}^{{\small \text{chP}}}\mathsf{\hat{T}}_{\lambda _{\text{q}%
_{1}},(\text{q}_{1}\text{$|$}\{\text{q}_{n},\text{r}_{n}\})}^{{\small \text{%
chP}}}\dots \mathsf{T}_{\lambda _{\text{r}_{n-1}},(\text{r}_{n-1}\text{$|$}\{%
\text{q}_{n},\text{r}_{n}\})}^{{\small \text{chP}}}\mathsf{\hat{T}}_{\lambda
_{\text{q}_{n-1}},(\text{q}_{n-1}\text{$|$}\{\text{q}_{n},\text{r}_{n}\})}^{%
{\small \text{chP}}}\text{,}
\end{eqnarray}
$\mathsf{U}_{n}$ surely satisfies the equation $\left( \ref{NJG-BS-ChP-Def-Un}\right) $
which defines the propagator.
\end{proof}

It is worth noticing that the eigenvalues of the two chP-transfer matrices
on the eigenstates of the $\tau _{2}$ transfer matrix are characterized
according to the discussion made in Section \ref{NJG-BS-ChP-SOV-Char-chP}, then the
eigenvalues of $\mathsf{U}_{m}$ are also known. Moreover, let us point out
that:%
\begin{equation}
\lambda _{\text{q}_{n}}=i\left( q\frac{\mathbbm{a}_{n}\beta _{n}}{\alpha _{n}%
\mathbbm{b}_{n}}\right) ^{1/2},\text{ \ \ \ }\lambda _{\text{r}_{n}}=i\left(
q\frac{\mathbbm{c}_{n}\beta _{n}}{\alpha _{n}\mathbbm{d}_{n}}\right) ^{1/2},
\end{equation}%
i.e. we are computing the Q-operators, $\mathsf{T}_{\lambda _{\text{r}%
_{n}}}^{{\small \text{chP}}}\mathsf{\hat{T}}_{\lambda _{\text{q}_{n}}}^{%
{\small \text{chP}}}$, in the zeros of the quantum determinant of the $\tau
_{2}$-model. In the case of self-adjoint representations on trivial curves
(like for sine-Gordon model) we have up to an overall constant:%
\begin{equation}
\mathsf{U}_{m}^{-1}\equiv \mathsf{Q}_{\lambda _{\text{r}_{1}}}\mathsf{Q}%
_{\lambda _{\text{q}_{1}}^{\ast }}\cdots \mathsf{Q}_{\lambda _{\text{r}%
_{m-1}}}\mathsf{Q}_{\lambda _{\text{q}_{m-1}}^{\ast }}.
\end{equation}%
The case of Bethe anzatz representations correspond to the case q$_{n}=$r$%
_{n}$, i.e. the two zeros of the quantum determinant coincide up to $p$%
-roots of units. In this case and in the homogeneous case we reproduce the
known result of \cite{NJG-BS-ChP-TTF83} for the propagator.

\section{Representation of local operators by separate variables}

The results on the scalar product formulae define one of the main steps to compute matrix elements of local operators. The other one is to reconstruct local operators by using the generators of the
Yang-Baxter algebra, namely to invert the map from the local operators in the
Lax matrices to the monodromy matrix elements. This inverse problem solution makes possible to compute the action of local operators on transfer matrix eigenstates in this way leading to the determination of form factors of local operators once the scalar product formulae are used.

In \cite{NJG-BS-ChP-KMT99} the first solution of this inverse problem has been obtained for the XXZ spin 1/2 chain and then in \cite{NJG-BS-ChP-MT00} it has been generalized to all fundamental lattice models having isomorphic auxiliary and local quantum spaces characterized by a Lax operator matrix coinciding with the permutation operator for a special value of the spectral parameter. This
reconstruction can be also used for non-fundamental lattice models, as derived in \cite{NJG-BS-ChP-MT00} for the higher spin XXX chains by using the fusion
procedure \cite{NJG-BS-ChP-KRS81}. For the Bazhanov-Stroganov model we still don't know how to achieve this type of
reconstruction and the known results reduce to those given
by T. Oota \cite{NJG-BS-ChP-Oota03}. However, Oota's results lead only to reconstruct some local operators of the Bazhanov-Stroganov model. We will explain in this section how to complete the Oota's reconstruction for all the local operators of the Bazhanov-Stroganov model associated to the most general cyclic representations of the 6-vertex Yang-Baxter algebra. The procedure  developed here is the natural generalization to these representations of the one for the special subclass presented in our previous paper \cite{NJG-BS-ChP-GMN12-SG}. The new technical tools required to handle these general representations will be also introduced in the next subsections.

\subsection{Reconstruction of a class of local operators}

The results of Oota's paper \cite{NJG-BS-ChP-Oota03} are here reproduced for the more general cyclic representations associated to the the Bazhanov-Stroganov model; this leads
to the reconstruction of a subclass of local operators. In terms of quantum projectors, when computed in the zeros $\mu _{n,\pm }$ of the quantum determinant, the Lax operator $\mathsf{L}_{n}(\lambda )$ has the following factorization:
\begin{equation}
\mathsf{L}_{n}(\mu _{n,+})\equiv \left( 
\begin{array}{c}
\left( \mathsf{L}_{n}\right) _{12}\mathsf{u}_{n}^{-1/2}f_{n} \\ 
\left( \mathsf{L}_{n}\right) _{21}\mathsf{u}_{n}^{1/2}f_{n}^{-1}%
\end{array}%
\right) \left( 
\begin{array}{cc}
\mathsf{u}_{n}^{-1/2}f_{n} & \mathsf{u}_{n}^{1/2}f_{n}^{-1}%
\end{array}%
\right) ,  \label{NJG-BS-ChP-L-factorization+}
\end{equation}%
\begin{equation}
\mathsf{L}_{n}(\mu _{n,-})\equiv \left( 
\begin{array}{c}
g_{n}\mathsf{u}_{n}^{1/2} \\ 
g_{n}^{-1}\mathsf{u}_{n}^{-1/2}%
\end{array}%
\right) \left( 
\begin{array}{cc}
g_{n}\mathsf{u}_{n}^{1/2}\left( \mathsf{L}_{n}\right) _{21} & g_{n}^{-1}%
\mathsf{u}_{n}^{-1/2}\left( \mathsf{L}_{n}\right) _{12}%
\end{array}%
\right) ,  \label{NJG-BS-ChP-L-factorization-}
\end{equation}%
where $\left( \mathsf{L}_{n}\right) _{ij}$ stays for the matrix element $i,j$
of the Lax operator and:%
\begin{equation}
f_{n}\equiv \left( -\frac{\alpha _{n}\beta _{n}}{\mathbbm{a}_{n}\mathbbm{b}%
_{n}}\right) ^{1/4},\text{ \ \ \ \ }g_{n}\equiv \left( -\frac{\alpha
_{n}\beta _{n}}{\mathbbm{c}_{n}\mathbbm{d}_{n}}\right) ^{1/4}.
\end{equation}%
These factorizations properties were used by Oota's to reconstruct local operators as it follows:
\begin{proposition}
The following reconstructions of local operators hold: 
\begin{align}
\mathsf{u}_{n}^{-1}& =\left( -\frac{\mathbbm{a}_{n}\mathbbm{b}_{n}}{\alpha
_{n}\beta _{n}}\right) ^{1/2}\mathsf{U}_{n}\mathsf{B}^{-1}(\mu _{n,+})%
\mathsf{A}(\mu _{n,+})\mathsf{U}_{n}^{-1}=\left( -\frac{\mathbbm{a}_{n}%
\mathbbm{b}_{n}}{\alpha _{n}\beta _{n}}\right) ^{1/2}\mathsf{U}_{n}\mathsf{D}%
^{-1}(\mu _{n,+})\mathsf{C}(\mu _{n,+})\mathsf{U}_{n}^{-1},  \label{NJG-BS-ChP-IPS-1} \\
&  \notag \\
\alpha _{0,n}& =\mathsf{U}_{n}\mathsf{A}^{-1}(\mu _{n,-})\mathsf{B}(\mu
_{n,-})\mathsf{U}_{n}^{-1}=\mathsf{U}_{n}\mathsf{C}^{-1}(\mu _{n,-})\mathsf{D%
}(\mu _{n,-})\mathsf{U}_{n}^{-1}.  \label{NJG-BS-ChP-IPS-2}
\end{align}
where we have defined:
\begin{equation}
\alpha _{0,n}\equiv \left( \frac{-\mathbbm{c}_{n}\mathbbm{b}_{n}^{2}}{\alpha
_{n}\beta _{n}\mathbbm{d}_{n}}\right) ^{1/2}\left( \frac{1+q^{-1}(\mathbbm{a}%
_{n}/\mathbbm{b}_{n})\mathsf{v}_{n}^{2}}{1+q^{-1}(\mathbbm{c}_{n}/\mathbbm{d}%
_{n})\mathsf{v}_{n}^{2}}\right) \mathsf{u}_{n}.
\end{equation}
\end{proposition}

Oota's formulae (\ref{NJG-BS-ChP-IPS-1})-(\ref{NJG-BS-ChP-IPS-2}) clearly allow to reconstruct all
the powers $\mathsf{u}_{n}^{-k}=\mathsf{U}_{n}\left( \mathsf{B}^{-1}(\mu
_{n,+})\mathsf{A}(\mu _{n,+})\right) ^{k}\mathsf{U}_{n}^{-1}$; however the local operators $\mathsf{v}_{n}^k$ do not admit direct reconstructions as only rational functions like $\left( 1+q^{-1}(\mathbbm{a}_{n}/\mathbbm{b}_{n})\mathsf{v}_{n}^{2}\right)
/\left( 1+q^{-1}(\mathbbm{c}_{n}/\mathbbm{d}_{n})\mathsf{v}_{n}^{2}\right) $
are reconstructed.

\subsection{Reconstruction of all local operators}

Here, we solve the inverse problem for the local operators $\mathsf{v}_{n}^k$ in this way completing the reconstruction of local operators. The cyclicity of the
representations of the Bazhanov-Stroganov model will be the main property here used. Let us define the
following local operators: 
\begin{equation}
\beta_{k,n} \equiv \left( \mathsf{U}_n \mathsf{A}^{-1}(\mu_{n,+}) \mathsf{B}%
(\mu_{n,+}) \mathsf{U}_n^{-1} \right)^{-k-1} \alpha_{0,n} \left( \mathsf{U}%
_n \mathsf{A}^{-1}(\mu_{n,+}) \mathsf{B}(\mu_{n,+}) \mathsf{U}_n^{-1}
\right)^{k}
\end{equation}
then it holds:

\begin{proposition}
\label{NJG-BS-ChP-IPS}For the cyclic representations of the Bazhanov-Stroganov model we consider, the local
operators $\mathsf{v}_{n}^{2k}$ have the following reconstructions: 
\begin{equation}
\mathsf{v}_{n}^{2k}=\frac{1}{p}\left( -\frac{\mathbbm{d}_{n}}{\mathbbm{c}_{n}%
}\right) ^{k}\frac{1+(\mathbbm{c}_{n}/\mathbbm{d}_{n})^{p}}{(\mathbbm{b}_{n}%
\mathbbm{c}_{n}/\mathbbm{a}_{n}\mathbbm{d}_{n})^{1/2}-(\mathbbm{a}_{n}%
\mathbbm{d}_{n}/\mathbbm{b}_{n}\mathbbm{c}_{n})^{1/2}}%
\sum_{a=0}^{p-1}q^{k(2a+1)}\beta _{a,n}\,.  \label{NJG-BS-ChP-IPS-3}
\end{equation}
\end{proposition}

\begin{proof}
By definition in our cyclic representations the powers $\mathsf{u}_{n}^{p}$ and $\mathsf{v}_{n}^{p}$ are central elements of the algebra coinciding with 1. Then it holds: 
\begin{equation}
\frac{1+(\mathbbm{c}_{n}/\mathbbm{d}_{n})^{p}}{1+q^{-2k-1}(\mathbbm{c}_{n}/%
\mathbbm{d}_{n})\mathsf{v}_{n}^{2}}=\sum_{i=0}^{p-1}\left( -q^{-2k-1}(%
\mathbbm{c}_{n}/\mathbbm{d}_{n})\mathsf{v}_n^{2}\right) ^{i}\,.
\end{equation}%
The previous formula and the reconstruction (\ref{NJG-BS-ChP-IPS-1})-(\ref{NJG-BS-ChP-IPS-2})
allow to rewrite $\beta _{k,n}$ as the following finite sum in
powers of $\mathsf{v}_{n}^{2}$: 
\begin{eqnarray}
\beta _{k,n} &=&\frac{(\mathbbm{b}_{n}\mathbbm{c}_{n}/\mathbbm{a}_{n}%
\mathbbm{d}_{n})^{1/2}+(\mathbbm{a}_{n}\mathbbm{d}_{n}/\mathbbm{b}_{n}%
\mathbbm{c}_{n})^{1/2}(\mathbbm{c}_{n}/\mathbbm{d}_{n})^{p}}{1+(\mathbbm{c}%
_{n}/\mathbbm{d}_{n})^{p}}  \notag \\
&+&\frac{(\mathbbm{b}_{n}\mathbbm{c}_{n}/\mathbbm{a}_{n}\mathbbm{d}%
_{n})^{1/2}-(\mathbbm{a}_{n}\mathbbm{d}_{n}/\mathbbm{b}_{n}\mathbbm{c}%
_{n})^{1/2}}{1+(\mathbbm{c}_{n}/\mathbbm{d}_{n})^{p}}%
\sum_{a=1}^{p-1}(-1)^{a}q^{-a(2k+1)}\left( \frac{\mathbbm{c}_{n}}{\mathbbm{d}%
_{n}}\right) ^{a}\mathsf{v}_{n}^{2a},
\end{eqnarray}
then, taking a discrete Fourier transformation, the
reconstructions (\ref{NJG-BS-ChP-IPS-3}) is obtained together with the following sum rules 
\begin{equation}
\sum_{a=0}^{p-1}\beta _{a,n}=p\frac{(\mathbbm{b}_{n}\mathbbm{c}_{n}/%
\mathbbm{a}_{n}\mathbbm{d}_{n})^{1/2}+(\mathbbm{a}_{n}\mathbbm{d}_{n}/%
\mathbbm{b}_{n}\mathbbm{c}_{n})^{1/2}(\mathbbm{c}_{n}/\mathbbm{d}_{n})^{p}}{%
1+(\mathbbm{c}_{n}/\mathbbm{d}_{n})^{p}}.
\end{equation}
\end{proof}

The formulae in (\ref{NJG-BS-ChP-IPS-3}) lead to the reconstruction of all the powers $\mathsf{v}_{n}^{k}$ for $k\in \{1,...,p-1\}$ as it follows from the identities $\mathsf{v}_{n}^{k}=\mathsf{v}
_{n}^{2h}$, for $k=2h-p$ odd integer smaller than $p$. Hence, as desired, all the local operators of the cyclic representations of the Bazhanov-Stroganov model are reconstructed by using the above proposition and the Oota's reconstructions.

\subsection{Separate variables representations of all local operators}

To compute the action of the local operators $\mathsf{v}_n^k$ and $%
\mathsf{u}_n^k$ on eigenstates of the transfer matrix and
then their form factors we need before to determine their
SOV-representations. These SOV-representations are obtained from the above solution of the inverse problem. To this
aim we first prove two lemmas that  are important to overcome the combinatorial
problem associated to the computation of the SOV-representations of the local
operators (\ref{NJG-BS-ChP-IPS-1})-(\ref{NJG-BS-ChP-IPS-2}).

Let us introduce, the coordinate operators $\hat{\boldsymbol\eta}_i$ for $i\in
\{1,...,\mathsf{N}\}$, $\hat{\boldsymbol\eta}_{\mathsf{A}}^{(\pm )}$ and $%
\hat{\boldsymbol\eta}_{\mathsf{D}}^{(\pm )}$ such that:%
\begin{equation}
\langle \boldsymbol\eta |{\hat{\boldsymbol\eta}_{i}}\equiv \eta _{i}\langle \boldsymbol\eta |,\text{ \ \ 
}\langle \boldsymbol\eta |\hat{\boldsymbol\eta}_{\mathsf{A}}^{(\pm )}\equiv \eta _{\mathsf{A%
}}^{(\pm )}\langle \boldsymbol\eta |,\text{\ \ }\langle \boldsymbol\eta |\hat{\boldsymbol\eta}_{%
\mathsf{D}}^{(\pm )}\equiv \eta _{\mathsf{D}}^{(\pm )}\langle \boldsymbol\eta |,
\end{equation}%
and the operator $\mathsf{T}_{i}^{\pm }$ are defined on the left anf right
SOV-representations by\footnote{%
It is worth remarking that from the definition of the SOV-representations of
the generators of the Yang-Baxter algebra, given in Section \ref{NJG-BS-ChP-SOV-Gen},
and the definitions in (\ref{NJG-BS-ChP-Def-T}), it follows that the SOV-representation
of the charge $\Theta $ coincides with the operator $\mathsf{T}_{\mathsf{N}%
}^{-}$ .}:%
\begin{equation}
\langle \boldsymbol\eta |\mathsf{T}_{i}^{\pm }\equiv \langle q^{\pm \delta _{i}}\boldsymbol\eta |,%
\text{ \ \ \ \ }\mathsf{T}_{i}^{\pm }|\boldsymbol\eta \rangle \equiv |q^{\mp \delta
_{i}}\boldsymbol\eta \rangle  \label{NJG-BS-ChP-Def-T}
\end{equation}%
and clearly the commutation relations hold:%
\begin{equation}
\mathsf{T}_{i}^{\pm }\hat{\boldsymbol\eta}_{j}=q^{\pm \delta _{i,j}}%
\hat{\boldsymbol\eta}_{j}\mathsf{T}_{i}^{\pm }.
\end{equation}

\begin{lemma}
We have the expansion 
\begin{equation}
\left( \hat{\boldsymbol\Omega} (f)\right) ^{k}=\sum_{\substack{ \vec{\alpha}=\{\alpha
_{1}\dots \alpha _{N-1}\}  \\ \sum \alpha _{i}=k}}{\left[ \!\!\!%
\begin{array}{c}
k \\ 
\vec{\alpha}%
\end{array}%
\!\!\!\right] }\prod_{i=1}^{\mathsf{N}-1}\left( \prod_{h=0}^{\alpha _{i}-1}f(q^{-h}%
\hat{\boldsymbol\eta}_{i})\prod_{j\neq i}\frac{1}{q^{\alpha _{j}-h}%
\hat{\boldsymbol\eta}_{i}/\hat{\boldsymbol\eta}_{j}-q^{-\alpha _{j}+h}%
\hat{\boldsymbol\eta}_{j}/\hat{\boldsymbol\eta}_{i}}\right) \prod_{i=1}^{\mathsf{N}-1}\left( 
\mathsf{T}_{i}^{-}\right) ^{\alpha _{i}}
\end{equation}%
for the operator
\begin{equation}
\hat{\boldsymbol\Omega} (f)=\sum_{a=1}^{\mathsf{N}-1}\prod_{b\neq a}\frac{1}{\hat{\boldsymbol\eta}_{a}/%
\hat{\boldsymbol\eta}_{b}-\hat{\boldsymbol\eta}_{b}/\hat{\boldsymbol\eta}_{a}}f(%
\hat{\boldsymbol\eta}_{a})\mathsf{T}_{a}^{-},
\end{equation}%
with
\begin{equation}
\left[ \!\!\!%
\begin{array}{c}
k \\ 
\vec{\alpha}%
\end{array}%
\!\!\!\right] \equiv \frac{\lbrack k]!}{\prod_{j=1}^{\mathsf{N}-1}\left[
\alpha _{j}\right] !},\text{ }[k]!\equiv \lbrack k][k-1]\cdots \lbrack 1],%
\text{ }[a]\equiv \frac{q^{a}-q^{-a}}{q-q^{-1}}.
\end{equation}
\end{lemma}

\begin{proof}
The lemma holds for $k=1$ and we prove it by induction for $k>1$. Let us
take $\mathsf{N}-1$ integers $\alpha _{i}$: 
\begin{equation}
\sum_{i=1}^{\mathsf{N}-1}\alpha _{i}=k,
\end{equation}%
from which we define the set of integers $I=\{i\in \{1,...,\mathsf{N}-1
\}:\alpha _{i}\neq 0\}$ and $\hat{\bold C}_{{\vec{\alpha}}}^{(k)}$\ as the operator coefficient
of $\prod \mathsf{T}_{i}^{-\alpha _{i}}$ (put to the left) in the expansion of the $k$-th
power of $\hat{\boldsymbol\Omega} (f)$. By writing $(\hat{\boldsymbol\Omega} (f))^{k}=(\hat{\boldsymbol\Omega} (f))^{k-1}\hat{\boldsymbol\Omega}
(f)$ and by using the induction hypothesis for the power $k-1$ of $\hat{\boldsymbol\Omega}
(f) $, we have:
\begin{align}
\hat{\bold C}_{{\vec{\alpha}}}^{(k)}& =\sum_{a\in I}{\left[ \!\!\!%
\begin{array}{c}
{k-1} \\ 
{\vec{\alpha}-\vec{\delta}_{a}}%
\end{array}%
\!\!\!\right] }\prod_{j=1}^{\mathsf{N}-1}\prod_{h=0}^{\alpha _{j}-\delta
_{a,j}-1}\left( f(q^{-h}\hat{\boldsymbol\eta}_{j})\times \prod_{i\neq
j,i=1}^{\mathsf{N}-1}\frac{1}{{q^{\alpha _{i}-\delta _{a,i}-h}\hat{\boldsymbol\eta}_{j}}/{%
\hat{\boldsymbol\eta}_{i}}-{\hat{\boldsymbol\eta}_{i}}/{q^{\alpha _{i}-\delta
_{a,i}-h}\hat{\boldsymbol\eta}_{j}}}\right)  \notag \\
& \times f(\hat{\boldsymbol\eta}_{a}q^{-\alpha _{a}+1})\prod_{i\in I\backslash
\{a\}}\frac{1}{{q^{\alpha _{a}-\alpha _{i}-1}\hat{\boldsymbol\eta}_{i}}/{%
\hat{\boldsymbol\eta}_{a}}-{\hat{\boldsymbol\eta}_{a}}/{q^{\alpha _{a}-\alpha _{i}-1}%
\hat{\boldsymbol\eta}_{i}}},
\end{align}%
with $\vec{\delta}_{a}\equiv (\delta _{1,a},\dots ,\delta _{\mathsf{N},a})$.
The first term in r.h.s. is the coefficient of $\prod \mathsf{T}%
_{i}^{-\alpha _{i}+\delta _{a,i}}$ in $(\hat{\boldsymbol\Omega} (f))^{k-1}$ and the second is
the coefficient of $\mathsf{T}_{a}^{-1}$ in $\hat{\boldsymbol\Omega} (f)$ once the
commutations between $\prod \mathsf{T}_{i}^{-\alpha _{i}+\delta _{a,i}}$ and
the $\hat{\boldsymbol\eta}_{i}$ have been performed. Hence we get: 
\begin{align}
\hat{\bold C}_{{\vec{\alpha}}}^{(k)}&=\frac{[k-1]!}{\prod [\alpha _{i}]!} \left(
\prod_{j=1}^{\mathsf{N}-1}\prod_{h=0}^{\alpha _{j}-1}(\prod_{i\neq j,i=1}^{\mathsf{N}-1}\frac{1%
}{{q^{\alpha _{i}-h}\hat{\boldsymbol\eta}_{j}}/{\hat{\boldsymbol\eta}_{i}}-{%
\hat{\boldsymbol\eta}_{i}}/{q^{\alpha _{i}-h}\hat{\boldsymbol\eta}_{j}}})f(q^{-h}%
\hat{\boldsymbol\eta}_{j})\right)  \notag \\
& \times \sum_{a\in I}([\alpha _{a}]\prod_{i\in I\backslash \{a\}}\frac{{%
q^{\alpha _{a}}\hat{\boldsymbol\eta}_{i}}/{\hat{\boldsymbol\eta}_{a}}-{%
\hat{\boldsymbol\eta}_{a}}/{q^{\alpha _{a}}\hat{\boldsymbol\eta}_{i}}}{{q^{\alpha
_{a}-\alpha _{i}}\hat{\boldsymbol\eta}_{i}}/{\hat{\boldsymbol\eta}_{a}}-{%
\hat{\boldsymbol\eta}_{a}}/{q^{\alpha _{a}-\alpha _{i}}\hat{\boldsymbol\eta}_{i}}}),
\end{align}%
which leads to our result by using the relation:%
\begin{equation}
\sum_{a=1}^{n}[\alpha _{a}]\prod_{i\neq a}\frac{{q^{\alpha _{a}}\eta _{i}}/{%
\eta _{a}}-{\eta _{a}}/{q^{\alpha _{a}}\eta _{i}}}{{q^{\alpha _{a}-\alpha
_{i}}\eta _{i}}/{\eta _{a}}-{\eta _{a}}/{q^{\alpha _{a}-\alpha _{i}}\eta _{i}%
}}=\left[ \sum_{a=1}^{n}\alpha _{a}\right] .
\end{equation}%
Note that the above formula holds for any $n$, for any set of numbers $\eta _{i}$ and
for any non-negative integers $\alpha _{i}$. This is proven by studying  the analytical properties of the function 
\begin{equation}
g(z)=\frac{1}{z}\prod \frac{z-\eta _{i}^{2}}{z-q^{-2\alpha _{i}}\eta _{i}^{2}%
}.
\end{equation}
\end{proof}

\begin{lemma}
The SOV-representation of the powers of $B^{-1}(\lambda )A(\lambda )$ are
given by 
\begin{equation}
\left( B^{-1}(\lambda )A(\lambda )\right) ^{m}=\sum_{i+j+k=m}\frac{(-1)^{j}}{%
\hat{\boldsymbol\eta}_{\mathsf{N}}^{m}}\left( \lambda \prod_{a=1}^{\mathsf{N}-1}\hat{\boldsymbol\eta}_{a}\right) ^{i-j}a_{+}^{i}a_{-}^{j}q^{\frac{i(i-1)-j(j-1)}{2}}\left[ \!\!\!%
\begin{array}{c}
m \\ 
i,j,k%
\end{array}%
\!\!\!\right] \hat{\boldsymbol\sigma} (\lambda )^{k}\mathsf{T}_{\mathsf{N}}^{j-i}
\end{equation}%
with 
\begin{equation}
\hat{\boldsymbol\sigma} (\lambda )=\sum_{a=1}^{\mathsf{N}-1}\prod_{b\neq a}\frac{1}{\hat{\boldsymbol\eta}_{a}/\hat{\boldsymbol\eta}_{b}-\hat{\boldsymbol\eta}_{b}/\hat{\boldsymbol\eta}_{a}}\frac{%
\mathtt{a}^{(SOV)}(\hat{\boldsymbol\eta}_{a})}{\lambda /\hat{\boldsymbol\eta}_{a}-%
\hat{\boldsymbol\eta}_{a}/\lambda }\mathsf{T}_{a}^{-},
\end{equation}%
where the powers of $\hat{\boldsymbol\sigma} (\lambda )$ are given by the previous lemma.
\end{lemma}

\begin{proof}
Let $\hat{\bold a}$, $\hat{\bold b}$ and $\hat{\bold c}$ be three operators satisfying the relations 
\begin{equation}
\hat{\bold b} \hat{\bold a}=q^{-2}\hat{\bold a} \hat{\bold b}, \hat{\bold c} \hat{\bold b}=q^{2}\hat{\bold b} \hat{\bold c}, \hat{\bold c} \hat{\bold a}=q^{-2}\hat{\bold a} \hat{\bold c}  \label{NJG-BS-ChP-Comb. power requ.}
\end{equation}%
It is easy to prove by induction that 
\begin{equation}
\left( \hat{\bold a}+\hat{\bold b}+\hat{\bold c}\right) ^{m}=\sum_{i+j+k=m}q^{k(j-i)-ij}\left[ \!\!\!%
\begin{array}{c}
m \\ 
i,j,k%
\end{array}%
\!\!\!\right] \hat{\bold a}^{i}\hat{\bold b}^{j}\hat{\bold c}^{k}  \label{NJG-BS-ChP-Comb. power}
\end{equation}%
The SOV-representation of $B^{-1}(\lambda )A(\lambda )$ is the sum of three
main terms, 
\begin{eqnarray}
\hat{\bold a} &=&\frac{\prod_{i=1}^{\mathsf{N}-1}\hat{\boldsymbol\eta}_{i}}{\hat{\boldsymbol\eta}_{\mathsf{N}}}\lambda a_{+}\mathsf{T}_{\mathsf{N}}^{-} \\
\hat{\bold b} &=&-\frac{\prod_{i=1}^{\mathsf{N}-1}\hat{\boldsymbol\eta}_{i}^{-1}}{\hat{\boldsymbol\eta}_{\mathsf{N}}}\lambda ^{-1}a_{-}\mathsf{T}_{\mathsf{N}}^{+} \\
\hat{\bold c} &=&\frac{1}{\hat{\boldsymbol\eta}_{\mathsf{N}}}\sum_{a=1}^{\mathsf{N}-1}\prod_{b\neq a}\frac{1}{%
\hat{\boldsymbol\eta}_{a}/\hat{\boldsymbol\eta}_{b}-
\hat{\boldsymbol\eta}_{b}/\hat{\boldsymbol\eta}_{a}}
\frac{\mathtt{a}^{(SOV)}(\hat{\boldsymbol\eta}_{a})}{\lambda /\hat{\boldsymbol\eta}_{a}-\hat{\boldsymbol\eta}_{a}/\lambda }\mathsf{T}%
_{a}^{-}
\end{eqnarray}%
Since they satisfy the commutation relations (\ref{NJG-BS-ChP-Comb. power requ.}), the
power of $B^{-1}(\lambda )A(\lambda )$ can be computed using the formula (%
\ref{NJG-BS-ChP-Comb. power}), which ends the proof.
\end{proof}

\textbf{Remark 1.} The quantum multinomials have the property 
\begin{equation}
{\left[ \!\!\!%
\begin{array}{c}
p \\ 
{\vec{\alpha}}%
\end{array}%
\!\!\!\right] }=\left\{ 
\begin{array}{l}
1\text{ \ \ if }\exists i\in \{1,...,\mathsf{N}-1\}:\alpha _{i}=p\delta _{a,i}%
\text{\ }\forall a\in \{1,...,\mathsf{N}-1\}, \\ 
0\text{ \ \ otherwise,}%
\end{array}%
\right.
\end{equation}%
This property yields that the power $p$ of $\mathsf{B}^{-1}(\lambda )\mathsf{%
A}(\lambda )$\ is a central element of the Yang-Baxter algebra and it reads: 
\begin{equation}
(\mathsf{B}^{-1}(\lambda )\mathsf{A}(\lambda ))^{p}=\mathcal{B}(\Lambda
)^{-1}\mathcal{A}(\Lambda ),  \label{NJG-BS-ChP-A/B^p-a}
\end{equation}%
result which is consistent with the commutations relations: 
\begin{equation}
\mathsf{B}^{-1}(q\lambda )\mathsf{A}(q\lambda )=\mathsf{A}(\lambda )\mathsf{B%
}^{-1}(\lambda ).
\end{equation}

The two previous lemmas allow to expand the SOV-representation of the
operators $\mathsf{u}_n^k$. However, they do not apply directly to the
expansion of $\mathsf{v}_n$. The aim of the following lemma is to transform
the operators $\beta_{k,n}$, whose linear combination gives the powers of $%
\mathsf{v}_n$.

\begin{lemma}
The operator $\beta_{k,n}$ has the following expansion: 
\begin{eqnarray}
\beta_{k,n}&=&\frac{\mathcal{B}(\mu_{n,-}^p)}{\mathcal{A}(\mu_{n,-}^p) 
\mathcal{B}(\mu_{n,+}^p)} \frac{\mu_{n,+}/\mu_{n,-}-\mu_{n,-}/\mu_{n,+}}{q^k
\mu_{n,+}/\mu_{n,-}-q^{-k}\mu_{n,-}/\mu_{n,+}}  \notag \\
&&\times B^{-1}(\mu_{n,+}) A(\mu_{n,+}) \prod_{i=1}^{p-k} B(q^{-i}
\mu_{n,+}) \left(B^{-1}(\mu_{n,-}) A(\mu_{n,-}) \right)^{p-1}
\prod_{i=p-k+1}^{p} B(q^{-i} \mu_{n,+})  \notag \\
&+&\frac{q^k-q^{-k}}{q^k \mu_{n,+}/\mu_{n,-}-q^{-k}\mu_{n,-}/\mu_{n,+}}
\end{eqnarray}
\end{lemma}

\begin{proof}
A simple induction on the Yang-Baxter relation $B(\lambda )A(q^{-1}\lambda
)=A(\lambda )B(q^{-1}\lambda )$ shows that 
\begin{equation}
\left( A^{-1}(\lambda )B(\lambda )\right)
^{k}=\prod_{i=1}^{k}B(q^{-i}\lambda )\prod_{i=1}^{k}A^{-1}(q^{-i}\lambda
)=\prod_{i=0}^{k-1}A^{-1}(q^{i}\lambda )\prod_{i=0}^{k-1}B(q^{i}\lambda )\,.
\end{equation}%
From the definition of the average values of operators, we get 
\begin{eqnarray}
\left( A^{-1}(\lambda )B(\lambda )\right) ^{k} &=&\mathcal{A}(\Lambda
)^{-1}\prod_{i=1}^{p-k}A(q^{-i}\lambda )\prod_{i=p-k+1}^{p}B(q^{-i}\lambda )\,,
\\
\left( A^{-1}(\lambda )B(\lambda )\right) ^{-k} &=&\mathcal{B}(\Lambda
)^{-1}\prod_{i=1}^{p-k}B(q^{-i}\lambda )\prod_{i=p-k+1}^{p}A(q^{-i}\lambda )\,.
\end{eqnarray}%
It also yields 
\begin{equation}
\left( A^{-1}(\lambda )B(\lambda )\right) ^{p}=\mathcal{A}^{-1}(\Lambda )%
\mathcal{B}(\Lambda )
\end{equation}%
and 
\begin{equation}
A^{-1}(\lambda )B(\lambda )=\mathcal{A}^{-1}(\Lambda )\mathcal{B}(\Lambda
)\left( B^{-1}(\lambda )A(\lambda )\right) ^{p-1}\,.
\end{equation}%
Standard arguments give the relation 
\begin{eqnarray}
B(\mu _{n,-})\prod_{i=1}^{p-k}A(q^{-i}\mu _{n,+}) &=&\frac{q^{k}-q^{-k}}{q^{k}\mu _{n,+}/\mu _{n,-}-q^{-k}\mu _{n,-}/\mu _{n,+}%
}A(\mu _{n,-})\prod_{i=1}^{p-k-1}A(q^{-i}\mu _{n,+})B(q^{k}\mu _{n,+}) 
\notag\\
&+&\frac{\mu _{n,+}/\mu
_{n,-}-\mu _{n,-}/\mu _{n,+}}{q^{k}\mu _{n,+}/\mu _{n,-}-q^{-k}\mu
_{n,-}/\mu _{n,+}}\prod_{i=1}^{p-k}A(q^{-i}\mu _{n,+})B(\mu _{n,-})\,. 
\end{eqnarray}%
Eventually, the use of these relations proves the lemma.
\end{proof}

\section{Form factors of local operators}
In this section we present the main results of our paper on the form factors of the local operators. One of the main peculiarities emerging in quantum separate variables is a feature of universality in the representation of these dynamical observables. In fact, the comparison between the results presented here for the most general cyclic representations of the 6-vertex Yang-Baxter algebra and those previously derived in our paper \cite{NJG-BS-ChP-GMN12-SG} defines one peculiar and evident instance of this universality.
\subsection{Form factors of $\mathsf{u}_{n}^{-1}$ and $\alpha _{0,n}^{-1}$}
The form factors of some local operators written as single determinants are here provided.

\begin{proposition}
\label{NJG-BS-ChP-FF-Prop1}Let us denote with $\varphi _{n}^{(t_{k})}$ and $\varphi _{n}^{(t_{k^{\prime }}^{\prime
})}$ the eigenvalues of the shift operator $\mathsf{U}_{n}$ respectively on the left $\langle t_{k}|$ and right $|t_{k^{\prime }}^{\prime }\rangle $ eigenstates of the transfer matrix $\tau _{2}(\lambda )$, then the following determinant formula is verified:
\begin{equation}
\langle t_{k}|\mathsf{u}_{n}^{-1}|t_{k^{\prime }}^{\prime }\rangle =\left( -%
\frac{\mathbbm{a}_{n}\mathbbm{b}_{n}}{\alpha _{n}\beta _{n}}\right) ^{1/2}%
\frac{\varphi _{n}^{(t_{k})}}{\varphi _{n}^{(t_{k^{\prime }}^{\prime })}}%
\delta _{k,k^{\prime }-1}\det_{\mathsf{N}-1}(||\mathcal{U}%
_{a,b}^{(t_{k},t_{k^{\prime }}^{\prime })}(\mu _{n,+})||). \label{NJG-BS-ChP-u}
\end{equation}Here, $||%
\mathcal{U}_{a,b}^{(t_{k},t_{k^{\prime }}^{\prime })}(\lambda )||$ is the $(\mathsf{N}-1)\times (\mathsf{N}-1)$ matrix defined by:
\begin{align}
\mathcal{U}_{a,b}^{(t_{k},t_{k^{\prime }}^{\prime })}(\lambda )& \equiv 
\mathcal{M}_{a,b+1/2}^{(t_{k},t_{k^{\prime }}^{\prime })}\text{ \ for \ }%
b\in \{1,...,\mathsf{N}-2\}, \\
\mathcal{U}_{a,\mathsf{N}-1}^{(t_{k},t_{k^{\prime }}^{\prime })}(\lambda )&
\equiv \frac{1}{{\eta }_{\mathsf{N}}^{(0)}}\sum_{h=1}^{p}\frac{\left( {\eta }%
_{a}^{(h)}\right) ^{\mathsf{N}-2}Q_{t_{k^{\prime }}^{\prime }}({\eta }%
_{a}^{(h)})}{\omega _{a}({\eta }_{a}^{(h)})}\left[ \frac{\bar{Q}_{t_{k}}({%
\eta }_{a}^{(h+1)})}{(\lambda /{\eta }_{a}^{(h+1)}-{\eta }%
_{a}^{(h+1)}/\lambda )}\mathtt{\bar{a}}^{(SOV)}({\eta }_{a}^{(h)})\right. 
\notag \\
& +\left. \bar{Q}_{t_{k}}({\eta }_{a}^{(h)})\left( a_{+}\lambda \left( {\eta 
}_{a}^{(h)}\right) ^{\mathsf{N}-1}q^{k^{\prime }}-\frac{a_{-}}{\lambda }%
\left( {\eta }_{a}^{(h)}\right) ^{-(\mathsf{N}-1)}q^{-k^{\prime }}\right) %
\right] .
\end{align}
\end{proposition}

\begin{proof}

The operator $\mathsf{B}^{-1}(\lambda )\mathsf{A}
(\lambda )$ admits the following SOV-representation:%
\begin{equation}
\mathsf{B}^{-1}(\lambda )\mathsf{A}(\lambda )=\frac{1}{\hat{\boldsymbol\eta}_{%
\mathsf{N}}}\left( \lambda \hat{\boldsymbol\eta}_{\mathsf{A}}^{(+)}\mathsf{T}_{%
\mathsf{N}}^{-}+\frac{\hat{\boldsymbol\eta}_{\mathsf{A}}^{(-)}}{\lambda }\mathsf{%
T}_{\mathsf{N}}^{+}\right) +\sum_{a=1}^{\mathsf{N}-1}\mathsf{T}_{a}^{-}\frac{%
\mathtt{\bar{a}}^{(SOV)}(\hat{\boldsymbol\eta}_{a})}{\hat{\boldsymbol\eta}_{\mathsf{N%
}}(\lambda /\hat{\boldsymbol\eta}_{a}q-\hat{\boldsymbol\eta}_{a}q/\lambda )}%
\prod_{b\neq a}\frac{1}{(\hat{\boldsymbol\eta}_{a}/\hat{\boldsymbol\eta}_{b}-%
\hat{\boldsymbol\eta}_{b}/\hat{\boldsymbol\eta}_{a})}.  \label{NJG-BS-ChP-ExB-1A}
\end{equation}
For brevity we denote with $[\mathsf{B}^{-1}(\lambda )\mathsf{A}(\lambda )]$ the
sum on the r.h.s. of (\ref{NJG-BS-ChP-ExB-1A}). Then, from the SOV-decomposition of the $\tau _{2}$-eigenstates, it holds:
\begin{align}
\langle t_{k}|[\mathsf{B}^{-1}(\lambda )\mathsf{A}(\lambda )]|t_{k^{\prime
}}^{\prime }\rangle & =\frac{\sum_{h_{\mathsf{N}}=1}^{p}q^{(k+1-k^{\prime
})h_{\mathsf{N}}}}{p{\eta }_{\mathsf{N}}^{(0)}}\sum_{a=1}^{\mathsf{N}%
-1}\sum_{h_{1},...,h_{\mathsf{N}-1}=1}^{p}V(\left( {\eta }%
_{1}^{(h_{1})}\right) ^{2},...,\left( {\eta }_{\mathsf{N}-1}^{(h_{\mathsf{N}%
-1})}\right) ^{2})  \notag \\
& \times \prod_{b\neq a,b=1}^{\mathsf{N}-1}\frac{{\eta }%
_{b}^{(h_{b})}Q_{t_{k^{\prime }}^{\prime }}({\eta }_{b}^{(h_{b})})\bar{Q}%
_{t_{k}}({\eta }_{b}^{(h_{b})})}{\omega _{b}({\eta }_{b}^{(h_{b})})(({\eta }%
_{a}^{(h_{a})})^{2}-({\eta }_{b}^{(h_{b})})^{2})}  \notag \\
& \times \frac{\bar{Q}_{t_{k}}({\eta }_{a}^{(h+1)})Q_{t_{k^{\prime
}}^{\prime }}({\eta }_{a}^{(h_{a})})}{\omega _{a}({\eta }_{a}^{(h_{a})})}%
\frac{\left( {\eta }_{a}^{(h_{a})}\right) ^{(\mathsf{N}-2)}\mathtt{\bar{a}}%
^{(SOV)}({\eta }_{a}^{(h_{a})})}{(\lambda /{\eta }_{a}^{(0)}q^{h_{a}+1}-{%
\eta }_{a}^{(0)}q^{h_{a}+1}/\lambda )},
\end{align}%
and so:%
\begin{align}
\langle t_{k}|[\mathsf{B}^{-1}(\lambda )\mathsf{A}(\lambda )]|t_{k^{\prime
}}^{\prime }\rangle & =\frac{\delta _{k,k^{\prime }-1}}{{\eta }_{\mathsf{N}%
}^{(0)}}\sum_{a=1}^{\mathsf{N}-1}\sum_{\substack{ h_{1},...,h_{\mathsf{N}}=1 
\\ {\small \overbrace{h_{a}\text{ is missing.}}}}}^{p}\underset{\ \ \ 
{\small \overbrace{(\text{The row } a \text{ is removed.}})}}{\hat{V}%
_{a}(\left( {\eta }_{1}^{(h_{1})}\right) ^{2},...,\left( {\eta }_{\mathsf{N}%
-1}^{(h_{\mathsf{N}-1})}\right) ^{2})}  \notag \\
& \times \prod_{b\neq a,b=1}^{\mathsf{N}-1}\frac{{\eta }%
_{b}^{(h_{b})}Q_{t_{k^{\prime }}^{\prime }}({\eta }_{b}^{(h_{b})})\bar{Q}%
_{t_{k}}({\eta }_{b}^{(h_{b})})}{\omega _{b}({\eta }_{b}^{(h_{b})})}  \notag
\\
& \times (-1)^{(\mathsf{N}-1+a)}\sum_{h_{a}=1}^{p}\frac{\bar{Q}_{t_{k}}({%
\eta }_{a}^{(0)}q^{h_{a}+1})Q_{t_{k^{\prime }}^{\prime }}({\eta }%
_{a}^{(h_{a})})\left( {\eta }_{a}^{(h_{a})}\right) ^{(\mathsf{N}-2)}\mathtt{%
\bar{a}}^{(SOV)}({\eta }_{a}^{(h_{a})})}{\omega _{a}({\eta }%
_{a}^{(h_{a})})(\lambda /{\eta }_{a}^{(h_{a}+1)}-{\eta }_{a}^{(h_{a}+1)}/%
\lambda )},
\end{align}
inserting the sum over ($h_{1},...,\widehat{h_{a}},...,h_{\mathsf{N}-1}$)
in the Vandermonde determinant $\hat{V}_{a}$, the above
expression reduces to the expansion of the following determinant:
\begin{equation}
\langle t_{k}|[\mathsf{B}^{-1}(\lambda )\mathsf{A}(\lambda )]|t_{k^{\prime
}}^{\prime }\rangle =\delta _{k,k^{\prime }-1}\det_{\mathsf{N}-1}(||\left[ 
\mathcal{U}_{a,b}^{(t_{k},t_{k^{\prime }}^{\prime })}(\lambda )\right] ||),
\label{NJG-BS-ChP-FF-Odd-part}
\end{equation}
where $\left[ \mathcal{U}_{a,b}^{(t_{k},t_{k^{\prime }}^{\prime })}(\lambda )%
\right] $ is just  $\mathcal{M}_{a,b+1/2}^{(t_{k},t_{k^{\prime
}}^{\prime })}$ for \ $b\in \{1,...,\mathsf{N}-2\}$, while:%
\begin{equation}
\left[ \mathcal{U}_{a,\mathsf{N}-1}^{(t_{k},t_{k^{\prime }}^{\prime
})}(\lambda )\right] \equiv \frac{\left( {\eta }_{a}^{(0)}\right) ^{\mathsf{N%
}-2}}{{\eta }_{\mathsf{N}}^{(0)}}\sum_{h=1}^{p}\frac{q^{(\mathsf{N}%
-2)h}Q_{t_{k^{\prime }}^{\prime }}({\eta }_{a}^{(h)})\bar{Q}_{t_{k}}({\eta }%
_{a}^{(h_{a}+1)})}{\omega _{a}({\eta }_{a}^{(h)})(\lambda /{\eta }%
_{a}^{(h_{a}+1)}-{\eta }_{a}^{(h_{a}+1)}/\lambda )}\mathtt{\bar{a}}^{(SOV)}({%
\eta }_{a}^{(h)}).
\end{equation}%
We compute now the matrix elements:
\begin{align}
\langle t_{k}|\hat{\boldsymbol\eta}_{\mathsf{N}}^{-1}\hat{\boldsymbol\eta}_{\mathsf{A%
}}^{(\pm )}\mathsf{T}_{\mathsf{N}}^{\mp }|t_{k^{\prime }}^{\prime }\rangle &
=\frac{\pm a_{\pm }q^{\pm k^{\prime }}\sum_{h_{\mathsf{N}}=1}^{p}q^{(k+1-k^{%
\prime })h_{\mathsf{N}}}}{p{\eta }_{\mathsf{N}}^{(0)}}\sum_{h_{1},...,h_{%
\mathsf{N}-1}=1}^{p}V(\left( {\eta }_{1}^{(h_{1})}\right) ^{2},...,\left( {%
\eta }_{\mathsf{N}-1}^{(h_{\mathsf{N}-1})}\right) ^{2})  \notag \\
& \times \prod_{b=1}^{\mathsf{N}-1}\frac{\left( {\eta }_{b}^{(h_{b})}\right)
^{\pm 1}Q_{t_{k^{\prime }}^{\prime }}({\eta }_{b}^{(h_{b})})\bar{Q}_{t_{k}}({%
\eta }_{b}^{(h_{b})})}{\omega _{b}({\eta }_{b}^{(h_{b})})},
\end{align}%
hence leading to:%
\begin{equation}
\langle t_{k}|\hat{\boldsymbol\eta}_{\mathsf{N}}^{-1}\hat{\boldsymbol\eta}_{\mathsf{A}}^{(\pm )}\mathsf{T}_{\mathsf{N}}^{\mp }|t_{k^{\prime }}^{\prime }\rangle =%
\frac{\pm a_{\pm }q^{\pm k^{\prime }}\delta _{k,k^{\prime }-1}}{{\eta }_{%
\mathsf{N}}^{(0)}}\det_{\mathsf{N}-1}(||\mathcal{M}_{a,b\pm
1/2}^{(t_{k},t_{k^{\prime }}^{\prime })}||).  \label{NJG-BS-ChP-FF-Even-part}
\end{equation}%
Then our result follows as the matrices of
formula (\ref{NJG-BS-ChP-FF-Odd-part}) and (\ref{NJG-BS-ChP-FF-Even-part}) have $\mathsf{N}-2$ common columns. Let us note that the above formula holds for any value of $%
\lambda $.
\end{proof}

\textbf{Remark 2.}

\textsf{I) } The matrix elements $\langle t_{k}|\alpha
_{0,n}^{-1}|t_{k^{\prime }}^{\prime }\rangle $ of the local operators $\alpha _{0,n}^{-1}$ are given by:%
\begin{equation}
\langle t_{k}|\alpha _{0,n}^{-1}|t_{k^{\prime }}^{\prime }\rangle =\frac{%
\varphi _{n}^{(t_{k})}}{\varphi _{n}^{(t_{k^{\prime }}^{\prime })}}\delta
_{k,k^{\prime }-1}\det_{\mathsf{N}-1}(||\mathcal{U}_{a,b}^{(t_{k},t_{k^{%
\prime }}^{\prime })}(\mu _{n,-})||).  \label{NJG-BS-ChP-alfa}
\end{equation}

\textsf{II)} In the case of general representations $\mathcal{R}_{\mathsf{N}%
}$ the matrix elements $\langle t_{k}|\mathsf{u}_{n}|t_{k^{\prime }}^{\prime
}\rangle $ can be computed by using the reconstruction: 
\begin{equation}
\mathsf{u}_{n}=\left( -\frac{\alpha _{n}\beta _{n}}{\mathbbm{a}_{n}%
\mathbbm{b}_{n}}\right) ^{1/2}\mathsf{U}_{n}\mathsf{C}^{-1}(\mu _{n,+})%
\mathsf{D}(\mu _{n,+})\mathsf{U}_{n}^{-1},
\end{equation}%
in the SOV$\ \mathsf{C}$-representation. Here, we do not make this
explicitly as the result will have the same type of form presented for $%
\langle t_{k}|\mathsf{u}_{n}^{-1}|t_{k^{\prime }}^{\prime }\rangle $; the
difference will be that all the quantities will be written in the SOV$\ 
\mathsf{C}$-representation.

\subsection{Determinant representations of form factors for a suitable basis of operators}

In this section we construct an operator basis for which the form factors of any operator in this basis is written by a one determinant formula. For this reason we will refer to it as the basis of elementary operators. The idea of the construction goes back to the sine-Gordon case \cite{NJG-BS-ChP-GMN12-SG}. 

\subsubsection{Introduction of the basis of elementary operators}

\begin{lemma}
Let us define the operators:%
\begin{equation}
\mathcal{O}_{a,k}\equiv \frac{\mathsf{B}({\eta }_{a}^{(p+k-1)})\mathsf{B}({%
\eta }_{a}^{(p+k-2)})\cdots \mathsf{B}({\eta }_{a}^{(k+1)})\mathsf{A}({\eta }%
_{a}^{(k)})}{p\hat{\boldsymbol\eta}_{\mathsf{N}}^{p-1}\prod_{b\neq a,b=1}^{\mathsf{N}%
-1}(Z_{a}/Z_{b}-Z_{b}/Z_{a})}\text{ \ with }k\in \{0,...,p-1\},
\end{equation}
then they satisfy the following properties:
\begin{equation}
\mathcal{O}_{a,k}\mathcal{O}_{a,h}\text{ is non-zero if and only if }h=k-1,
\label{NJG-BS-ChP-Prod-O-zeros}
\end{equation}%
and%
\begin{equation}
\mathcal{O}_{a,k}\mathcal{O}_{a,k-1}\cdots \mathcal{O}_{a,k+1-p}\mathcal{O}%
_{a,k-p}=\frac{\mathcal{A}(Z_{a})}{\prod_{b\neq a,b=1}^{\mathsf{N}%
-1}(Z_{a}/Z_{b}-Z_{b}/Z_{a})}\mathcal{O}_{a,k}.  \label{NJG-BS-ChP-O-mean-value}
\end{equation}%
The following commutation relations are furthermore satisfied:
\begin{equation}
\hat{\boldsymbol\eta} _{\mathsf{A}}^{(\pm )}\mathcal{O}_{a,k}=q^{\mp 1}\mathcal{O}_{a,k}\hat{\boldsymbol\eta}_{\mathsf{A}}^{(\pm )},\text{ \ }[\hat{\boldsymbol\eta}_{\mathsf{N}},\mathcal{O}_{a,k}]=[%
\mathsf{T}_{\mathsf{N}}^{-} ,\mathcal{O}_{a,k}]=0,
\end{equation}%
and 
\begin{equation}
\mathcal{O}_{a,k}\mathcal{O}_{b,h}=\frac{({\eta }_{a}^{(k-h+1)}/{\eta }%
_{b}^{(0)}-{\eta }_{b}^{(0)}/{\eta }_{a}^{(k-h+1)})}{({\eta }_{a}^{(k-h-1)}/{%
\eta }_{b}^{(0)}-{\eta }_{b}^{(0)}/{\eta }_{a}^{(k-h-1)})}\text{ }\mathcal{O}%
_{b,h}\mathcal{O}_{a,k}  \label{NJG-BS-ChP-Com-O}
\end{equation}%
for $a\neq b\in \{1,...,\mathsf{N}-1\}$.
\end{lemma}

\begin{proof}
Since $\mathcal{B}(Z_{a})=0$ with $\mathcal{B}(\Lambda )$ the average value of the operator $\mathsf{B}(\lambda )$, the first is quite immediate. Moreover, the following identity:%
\begin{equation}
\langle \eta _{1},...,{\eta }_{a}^{(h)},...,\eta _{\mathsf{N}}|\mathcal{O}%
_{a,k}=\frac{a({\eta }_{a}^{(k)})\delta _{h,k}}{\prod_{b\neq a,b=1}^{\mathsf{%
N}}({\eta }_{a}^{(k)}/\eta _{b}-\eta _{b}/{\eta }_{a}^{(k)})}\langle \eta
_{1},...,{\eta }_{a}^{(k-1)},...,\eta _{\mathsf{N}}|,  \label{NJG-BS-ChP-O-Action}
\end{equation}
is a direct consequence of the definition of the operators $\mathcal{O}_{a,k}$ so that the second identity of the lemma follows.
Now by using the following Yang-Baxter commutation relation:
\begin{equation}
(\lambda /\mu -\mu /\lambda )\mathsf{A}(\lambda )\mathsf{B}(\mu )=(\lambda
/q\mu -\mu q/\lambda )\mathsf{B}(\mu )\mathsf{A}(\lambda )+(q-q^{-1})\mathsf{%
B}(\lambda )\mathsf{A}(\mu )  \label{NJG-BS-ChP-AB-Yang-Baxter}
\end{equation}%
and moving the $\mathsf{A}({\eta }_{a}^{(k)})$ to the right
through all the $\mathsf{B}({\eta }_{b}^{(j)})$, for $j\neq h$, remarking
that only the first term of the r.h.s of (\ref{NJG-BS-ChP-AB-Yang-Baxter}) survives,
and after moving the $\mathsf{A}({\eta }_{b}^{(h)})$ to the left, we get the last identity of the lemma. 
\end{proof}

Now we define \textit{elementary operators} by the following monomials:
\begin{equation}
\mathcal{E}_{k,k_{0},(a_{1},k_{1}),...,(a_{r},k_{r})}^{(\alpha
_{1},...,\alpha _{r})}\equiv \hat{\boldsymbol\eta}_{\mathsf{N}}^{-k}\left( %
\hat{\boldsymbol\eta}_{\mathsf{A}}^{(+)}\mathsf{T}_{\mathsf{N}}^{-} \right)
^{k_{0}}\mathcal{O}_{a_{1},k_{1}}^{(\alpha _{1})}\cdots \mathcal{O}%
_{a_{r},k_{r}}^{(\alpha _{r})},  \label{NJG-BS-ChP-O-basis}
\end{equation}%
where $\sum_{h=1}^{r}\alpha _{h}\leq p,$ $k,k_{i}\in \{0,...,p-1\},\ \
a_{i}<a_{j}\in \{1,...,\mathsf{N}-1\}$ \ for $i<j\in \{1,...,\mathsf{N}-1\}$
and:%
\begin{equation}
\mathcal{O}_{a,k}^{(\alpha )}\equiv \mathcal{O}_{a,k}\mathcal{O}%
_{a,k-1}\cdots \mathcal{O}_{a,k+1-\alpha },\text{ with }\alpha \in
\{1,...,p\}.
\end{equation}

\begin{lemma}
Once the set of the elementary operators
is dressed by the shift operator $\mathsf{U}_{n}$ as it follows:%
\begin{equation}
\mathsf{U}_{n}\mathcal{E}_{k,k_{0},(a_{1},k_{1}),...,(a_{r},k_{r})}^{(\alpha
_{1},...,\alpha _{r})}\mathsf{U}_{n}^{-1},
\end{equation}
a basis is defined in the space of the local operators at the quantum site $n$, $\forall n\in \{1,...,\mathsf{N}\}$.
\end{lemma}

\begin{proof}
In order to prove the lemma the local operators in site $n$ generated by $\mathsf{u}%
_{n}^{k}$ and $\mathsf{v}_{n}^{k}$ for $k\in \{1,...,p-1\}$ have to be written as linear combinations of the dressed elementary operators (\ref{NJG-BS-ChP-O-basis}) and thanks to Proposition \ref{NJG-BS-ChP-IPS} this is equivalent to prove the same statement for the following basis of local operators:
\begin{align}
\mathsf{u}_{n}^{-k}& =\mathsf{U}_{n}\left( \mathsf{B}^{-1}(\mu _{n,+})%
\mathsf{A}(\mu _{n,+})\right) ^{k}\mathsf{U}_{n}^{-1},  \label{NJG-BS-ChP-U-B-basis1} \\
\tilde{\beta}_{k,n}& =\mathsf{U}_{n}\left( \mathsf{B}^{-1}(\mu _{n,+})%
\mathsf{A}(\mu _{n,+})\right) ^{k}\mathsf{B}^{-1}(\mu _{n,-})\mathsf{A}(\mu
_{n,-})\left( \mathsf{B}^{-1}(\mu _{n,+})\mathsf{A}(\mu _{n,+})\right)
^{p-1-k}\mathsf{U}_{n}^{-1}  \label{NJG-BS-ChP-U-B-basis2}
\end{align}
The operator $\mathsf{B}^{-1}(\lambda )$ is invertible for $\lambda ^{p}\neq Z_{a}$ with $a\in \{1,...,\mathsf{N}-1\}$ so that the
centrality of the average values implies: 
\begin{equation}
\mathsf{B}^{-1}(\lambda )\mathsf{A}(\lambda )=\frac{\mathsf{B}(\lambda
q^{p-1})\mathsf{B}(\lambda q^{p-2})\cdots \mathsf{B}(\lambda q)\mathsf{A}%
(\lambda )}{\mathcal{B}(\Lambda )}.
\end{equation}
The monomial $\mathsf{B}(\lambda q^{p-1})\mathsf{B}(\lambda
q^{p-2})\cdots \mathsf{B}(\lambda q)\mathsf{A}(\lambda )$ is an even Laurent
polynomial of degree $p(\mathsf{N}-1)+1$ in $\lambda $ and so we can write:
\begin{equation}
\mathsf{B}^{-1}(\lambda )\mathsf{A}(\lambda )=\frac{1}{\hat{\boldsymbol\eta}_{%
\mathsf{N}}}\left( \lambda \hat{\boldsymbol\eta}_{\mathsf{A}}^{(+)}\mathsf{T}_{%
\mathsf{N}}^{-} +\frac{\hat{\boldsymbol\eta}_{\mathsf{A}}^{(-)}}{\lambda }
\mathsf{T}_{\mathsf{N}}^{+}\right) +\frac{1}{\hat{\boldsymbol\eta}_{\mathsf{N}}}%
\sum_{a=1}^{\mathsf{N}-1}\sum_{k=0}^{p-1}\frac{\mathcal{O}_{a,k}}{(\lambda /{\eta}_{a}^{(k)}-{\eta }_{a}^{(k)}/\lambda )}.
\end{equation}
It is then clear that the local operators $\mathsf{u}_{n}^{-k}$ and $\tilde{\beta}_{k,n}$ are linear combinations of the monomials:
\begin{equation}
\mathsf{U}_{n} \hat{\boldsymbol\eta} _{\mathsf{N}}^{-h}\left( \hat{\boldsymbol\eta}_{\mathsf{A}}^{(+)}\mathsf{T}_{\mathsf{N}}^{-} \right) ^{h_{0}}\mathcal{O}%
_{a_{1},h_{1}}\cdots \mathcal{O}_{a_{s},h_{s}}\mathsf{U}_{n}^{-1}
\end{equation}
for $s\leq p$, $a_{i}\in \{1,...,\mathsf{N}-1\}$ and $h,h_{i}\in
\{0,...,p-1\}$. The commutation rules (\ref{NJG-BS-ChP-Com-O}) allow to rewrite any monomial $\mathcal{O}_{a_{1},h_{1}}\cdots \mathcal{O}%
_{a_{s},h_{s}}$ in a way that operators with the same index $a$ are adjacent and those with different $a$ are ordered in a way $a_{i}<a_{j}$\ for $i<j\in \{1,...,\mathsf{N}-1\}$. Then the rule (\ref{NJG-BS-ChP-Prod-O-zeros}) tell us if the monomial is zero or 
non-zero. The property (\ref{NJG-BS-ChP-O-mean-value}) finally implies: 
\begin{equation}
\mathcal{O}_{a,k}^{(p+\alpha )}=\frac{\mathcal{A}(Z_{a})}{\prod_{b\neq
a,b=1}^{\mathsf{N}}(Z_{a}/Z_{b}-Z_{b}/Z_{a})}\mathcal{O}_{a,k}^{(\alpha )},
\end{equation}%
and so that all the non-zero monomials $\mathcal{O}%
_{a_{1},h_{1}}\cdots \mathcal{O}_{a_{s},h_{s}}$ are rewritable in the form (\ref{NJG-BS-ChP-O-basis}).
\end{proof}

\subsubsection{Determinant representation of elementary operator form factors}

\begin{lemma}
The elementary operators admit the following simple characterizations for their form factors: 
\begin{equation}
\langle t_{k}|\mathcal{E}_{(h,h_{0},(a_{1},h_{1}),...,(a_{r},h_{r})}^{(%
\alpha _{1},...,\alpha _{r})}|t_{k^{\prime }}^{\prime }\rangle =\frac{\delta
_{k,k^{\prime }+h}a_{+}^{h_{0}}q^{h_{0}k^{\prime }}}{({\eta }_{\mathsf{N}%
}^{(0)})^{h}}\,\mathsf{f}_{(h_{0},\{\alpha \},\{a\})}\det_{\mathsf{N}%
-1+rp-g}(||\text{\textsc{O}}_{a,b}^{(h_{0},\{\alpha \},\{a\})}||).
\label{NJG-BS-ChP-E-op}
\end{equation}
Here, $\langle t_{k}|$ and $|t_{k^{\prime }}^{\prime }\rangle $ are two
eigenstates of the transfer matrix $\tau _{2}(\lambda )$ and $||$\textsc{O}$_{a,b}^{(h_{0},\{\alpha \},\{a\})}||$ is the $\left( 
\mathsf{N}-1+rp-g\right) \times \left( \mathsf{N}-1+rp-g\right) $ matrix of
elements:
\begin{align}
\text{\textsc{O}}_{a,\sum_{h=1}^{m-1}(p-\alpha
_{h}+1)+j_{m}}^{(h_{0},\{\alpha \},\{a\})}& \equiv \left( \eta
_{a_{m}}^{2}q^{2j_{m}}\right) ^{2(a-1)}\text{ \ for }j_{m}\in
\{0,...,p-\alpha _{m}\},\text{ \ }m\in \{1,...,r\},\text{ } \\
\text{\textsc{O}}_{a,\sum_{h=1}^{r}(p-\alpha _{h}+1)+i}^{(h_{0},\{\alpha
\},\{a\})}& \equiv \Phi _{b_{i},a+(h_{0}+g)/2}^{\left( t,t^{\prime }\right)
},\text{ \ \ \ \ \ \ \ \ \ \ for }i\in \{1,...,\mathsf{N}-1-r\},\text{ \ \ \ 
}g\equiv \sum_{h=1}^{r}\alpha _{h},
\end{align}%
for any $a\in \{1,...,\mathsf{N}-1+rp-g\}$. Moreover, we have used the following notations $%
\{b_{1},...,b_{\mathsf{N}-1-r}\}\equiv \{1,...,\mathsf{N}-1\}\backslash
\{a_{1},...,a_{r}\}$ where the elements are ordered by $b_{i}<b_{j}$ for $i<j$,
\begin{align}
\mathsf{f}_{(h_{0},\{\alpha \},\{a\})}& \equiv \frac{\prod_{i=1}^{r}Q_{t^{%
\prime }}(\eta _{a_{i}}q^{-\alpha _{i}})\bar{Q}_{t}(\eta _{a_{i}})\left(
\eta _{a_{i}}^{h_{0}+\alpha _{i}\left( \mathsf{N}-1-r\right) }/\omega
_{a_{i}}(\eta _{a_{i}})\right) \prod_{h=0}^{\alpha _{i}-1}a(\eta _{i}q^{-h})%
}{\prod_{i=1}^{r}\prod_{h=0}^{\alpha _{i}-1}\prod_{j=1}^{i-1}(q^{\alpha
_{j}-h}\eta _{a_{i}}/\eta _{a_{j}}-\eta _{a_{j}}/q^{\alpha _{j}-h}\eta
_{a_{i}})\prod_{j=i+1}^{r}(\eta _{a_{i}}/q^{h}\eta _{a_{j}}-\eta
_{a_{j}}q^{h}/\eta _{a_{i}})}  \notag \\
& \times \frac{(-1)^{\sum_{i=1}^{r}(a_{i}-i)}\prod_{i=1}^{r}q^{-\left( 
\mathsf{N}-1-r\right) \alpha _{i}(\alpha _{i}-1)/2}V(\eta
_{a_{1}}^{2},...,\eta _{a_{r}}^{2})}{\prod_{i=1}^{r}\prod_{j=1}^{\mathsf{N}%
-1-r}(Z_{a_{i}}^{2}-Z_{b_{j}}^{2})V(\eta _{a_{1}}^{2},\eta
_{a_{1}}^{2}q^{2},...,\eta _{a_{1}}^{2}q^{2(p-\alpha _{1})},...,\eta
_{a_{r}}^{2},\eta _{a_{r}}^{2}q^{2},...,\eta _{a_{r}}^{2}q^{2(p-\alpha
_{r})})},
\end{align}
$V(x_{1},...,x_{\mathsf{N}})\equiv \prod_{1\leq a<b\leq \mathsf{N}%
}(x_{a}-x_{b})$ is the Vandermonde determinant and for brevity:%
\begin{equation}
\eta _{a_{m}}\equiv \eta _{a_{m}}^{(h_{m})}.  \label{NJG-BS-ChP-eta-a_m}
\end{equation}
\end{lemma}

\begin{proof}
The following actions hold:
\begin{align}
\langle t_{k}|
\hat{\boldsymbol\eta}_{\mathsf{N}}^{-h}
(\hat{\boldsymbol\eta}_{\mathsf{A}}^{(+)}
\mathsf{T}_{\mathsf{N}}^{-})^{h_{0}}&
=\frac{a_{+}^{h_{0}}q^{h_{0}(k-h)}}{({\eta }_{\mathsf{N}}^{(0)})^{h}}%
\sum_{h_{1},...,h_{\mathsf{N}}=1}^{p}\frac{q^{(k-h)h_{\mathsf{N}}}}{p^{1/2}}%
\prod_{a=1}^{\mathsf{N}-1}({\eta }_{a}^{(h_{a})})^{h_{0}}\bar{Q}_{t}({\eta }%
_{a}^{(h_{a})})\notag\\
&\times \prod_{1\leq a<b\leq \mathsf{N}-1}(({\eta }_{a}^{(h_{a})})^{2}-({%
\eta }_{b}^{(h_{b})})^{2})\frac{\langle {\eta }_{1}^{(h_{1})},...,{\eta }_{%
\mathsf{N}}^{(h_{\mathsf{N}})}|}{\prod_{b=1}^{\mathsf{N}-1}\omega _{b}({\eta 
}_{b}^{(h_{b})})}.
\end{align}
From the formula (\ref{NJG-BS-ChP-O-Action}), it follows:%
\begin{equation}
\langle \eta _{1},...,{\eta }_{a_{i}}^{(f)},...,\eta _{\mathsf{N}}|\mathcal{O%
}_{a_{i},h_{i}}^{(\alpha _{i})}=\frac{\prod_{h=0}^{\alpha _{i}-1}a(\eta
_{a_{i}}q^{-h})\delta _{f,h_{i}}\langle \eta _{1},...,\eta
_{a_{i}}q^{-\alpha _{i}},...,\eta _{\mathsf{N}}|}{\prod_{b\neq a_{i},b=1}^{%
\mathsf{N}-1}\prod_{h=0}^{\alpha _{i}-1}(\eta _{a_{i}}q^{-h}/\eta _{b}-\eta
_{b}/\eta _{a_{i}}q^{-h})},  \label{NJG-BS-ChP-O-a-Action}
\end{equation}%
where $\eta _{a_{i}}$ is defined in (\ref{NJG-BS-ChP-eta-a_m}). The action of $\mathcal{O}_{a_{1},h_{1}}^{(\alpha _{1})}\cdots \mathcal{O}%
_{a_{r},h_{r}}^{(\alpha _{r})}$ can be computed now taking into account the order of the
operators which appear in the monomial then by using the scalar product formula we get:
\begin{align}
\langle t_{k}|\mathcal{E}_{(h,h_{0},(a_{1},h_{1}),...,(a_{r},h_{r})}^{(%
\alpha _{1},...,\alpha _{r})}|t_{k^{\prime }}^{\prime }\rangle & =\frac{%
a_{+}^{h_{0}}q^{h_{0}(k-h)}}{({\eta }_{\mathsf{N}}^{(0)})^{h}}%
\sum_{k_{1},...,k_{\mathsf{N}}=1}^{p}\frac{q^{\left[ (k-h)-k^{\prime }\right]
k_{\mathsf{N}}}}{p}\prod_{a=1}^{\mathsf{N}-1}({\eta }_{a}^{(h_{a})})^{h_{0}}
\notag \\
& \times \prod_{i=1}^{r}\frac{\prod_{h=0}^{\alpha _{i}-1}a(\eta
_{a_{i}}q^{-h})\delta _{k_{a_{i}},h_{i}}}{\prod_{j=1}^{\mathsf{N}%
-1-r}\prod_{h=0}^{\alpha _{i}-1}(\eta _{a_{i}}q^{-h}/\eta
_{b_{j}}^{(k_{b_{j}})}-\eta _{b_{j}}^{(k_{b_{j}})}/\eta _{a_{i}}q^{-h})} 
\notag \\
& \times \prod_{i=1}^{r}\prod_{h=0}^{\alpha _{i}-1}\frac{\prod_{j=i+1}^{r}(%
\eta _{a_{i}}q^{-h}/\eta _{a_{j}}-\eta _{a_{j}}/\eta _{a_{i}}q^{-h})^{-1}}{%
\prod_{j=1}^{i-1}(\eta _{a_{i}}q^{\alpha _{j}-h}/\eta _{a_{j}}-\eta
_{a_{j}}/\eta _{a_{i}}q^{\alpha _{j}-h})}  \notag \\
& \times \prod_{j=1}^{\mathsf{N}-1-r}\frac{Q_{t^{\prime }}(\eta
_{b_{j}}^{(k_{b_{j}})})\bar{Q}_{t}(\eta _{b_{j}}^{(k_{b_{j}})})}{\omega
_{b_{j}}(\eta _{b_{j}}^{(k_{b_{j}})})}\prod_{i=1}^{r}\frac{Q_{t^{\prime
}}(\eta _{a_{i}}q^{-\alpha _{i}})\bar{Q}_{t}(\eta _{a_{i}})}{\omega
_{a_{i}}(\eta _{a_{i}})}V(\eta _{1}^{2},...,\eta _{\mathsf{N}-1}^{2}).
\end{align}
The presence of the $\prod_{i=1}^{r}\delta
_{k_{a_{i}},h_{i}}$ reduces the sum $\sum_{k_{1},...,k_{\mathsf{N}}=1}^{p}$ to $\delta _{k,k^{\prime }+h}$ times the sum $\sum_{k_{b_{1}},...,k_{b_{%
\mathsf{N}-(r+1)}}=1}^{p}$ where:%
\begin{equation}
\{a_{1},...,a_{r}\}\cup \{b_{1},...,b_{\mathsf{N}-(r+1)}\}=\{1,...,\mathsf{N}%
-1\}.
\end{equation}
We get our formula (\ref{NJG-BS-ChP-E-op}) multiplying each term of the sum by: 
\begin{align}
1& =\prod_{\epsilon =\pm 1}\prod_{i=1}^{r}\prod_{j=1}^{\mathsf{N}%
-1-r}\prod_{h=-p+\alpha _{i}}^{-1}(\eta {}_{a_{i}}^{2}q^{-2h}-(\eta
_{b_{j}}^{(k_{b_{j}})})^{2})^{\epsilon }  \notag \\
& \times \left( \frac{V(\eta _{a_{1}}^{2},\eta _{a_{1}}^{2}q^{2},...,\eta
_{a_{1}}^{2}q^{2(p-\alpha _{1})},...,\eta _{a_{r}}^{2},\eta
_{a_{r}}^{2}q^{2},...,\eta _{a_{r}}^{2}q^{2(p-\alpha _{r})})}{V(\eta
_{a_{1}}^{2},...,\eta _{a_{r}}^{2})}\right) ^{\epsilon }\,.
\end{align}%
Indeed, the power $+1$ leads to the construction of the Vandermonde determinant:
\begin{equation}
V(\underset{p-\alpha _{1}+1\text{ columns}}{\underbrace{\eta
_{a_{1}}^{2},...,\eta _{a_{1}}^{2}q^{2(p-\alpha _{1})}}},...,\underset{%
p-\alpha _{r}+1\text{ columns}}{\underbrace{\eta _{a_{r}}^{2},...,\eta
_{a_{r}}^{2}q^{2(p-\alpha _{r})}}},\underset{(\mathsf{N}-1)-r\text{ columns}}%
{\underbrace{(\eta _{b_{1}}^{(k_{b_{1}})})^{2},...,(\eta _{b_{b_{(\mathsf{N}%
-1)-r}}}^{(k_{b_{b_{(\mathsf{N}-1)-r}}})})^{2}}}),
\end{equation}
and the sum $\sum_{k_{b_{1}},...,k_{b_{\mathsf{N}-(r+1)}}=1}^{p}$ becomes
sum over columns which can be brought inside the determinant.
\end{proof}
\subsection{The chiral Potts model order parameters}

The results presented in the previous subsections are
as well results for the matrix elements of local operators in the
inhomogeneous chiral Potts model. In particular, let $|t_{k}\rangle $ and $|t_{k^{\prime }}^{\prime }\rangle 
$ be two eigenstates of the chiral Potts transfer matrix, then the matrix
elements:%
\begin{equation*}
\langle t_{k}|\mathsf{u}_{n}^{-1}|t_{k^{\prime }}^{\prime }\rangle ,\text{ \
\ \ }\langle t_{k}|\alpha _{0,n}^{-1}|t_{k^{\prime }}^{\prime }\rangle \text{%
\ \ \ and\ \ }\langle t_{k}|\mathcal{E}%
_{(h,h_{0},(a_{1},h_{1}),...,(a_{r},h_{r})}^{(\alpha _{1},...,\alpha
_{r})}|t_{k^{\prime }}^{\prime }\rangle
\end{equation*}%
are given respectively by the formulae $(\ref{NJG-BS-ChP-u}),$ $(\ref{NJG-BS-ChP-alfa})$ and $(\ref{NJG-BS-ChP-E-op})$. Furthermore, in the representations $\mathcal{R}_{\mathsf{N}}^{{\small 
\text{chP,S-adj}}}$ the formulae $(\ref{NJG-BS-ChP-u}),$ $(\ref{NJG-BS-ChP-alfa})$ and $(\ref{NJG-BS-ChP-E-op}%
)$ are always matrix elements of the corresponding local operators on chiral
Potts eigenstates. As clarified below, some of these matrix elements generate the chiral Potts order parameters under the homogeneous and thermodynamic limits.

\subsubsection{Local Hamiltonians and order parameters}

It is worth recalling that the following local quantum Hamiltonians:%
\begin{align}
H\left. \equiv \right. H_{0}+kH_{1},\text{ \ }H_{0}& \equiv \sum_{n=1}^{%
\mathsf{N}}\left[ \sum_{r=1}^{p-1}f_{r}(\theta )\mathsf{u}_{n}^{r}\mathsf{u}%
_{n+1}^{-r}\right] ,\text{ \ }H_{1}\equiv \sum_{n=1}^{\mathsf{N}}\left[
\sum_{r=1}^{p-1}f_{r}(\bar{\theta})\mathsf{v}_{n}^{r}\right] ,
\label{NJG-BS-ChP-local hamiltonians} \\
f_{r}(\theta )& \equiv \frac{e^{i(2r-p)\theta /p}}{\sin \pi r/p},\text{ \ \ }%
\cos \bar{\theta}=\frac{\cos \theta }{k},\text{ \ }e^{i(2\theta -\pi
)/p}\equiv \frac{x_{\text{q}_{n}}}{y_{\text{q}_{n}}}=\frac{x_{\text{r}_{n}}}{%
y_{\text{r}_{n}}},
\end{align}%
first constructed by von Gehlen and Rittenberg \cite{NJG-BS-ChP-vGR85}, commute with
the homogeneous Z$_{p}$ chP transfer matrices. Indeed, they are generated by
derivative of these transfer matrices w.r.t. the spectral parameter, see for
example \cite{NJG-BS-ChP-AMcP0} for a derivation. Then the order parameters associated
to the homogeneous Z$_{p}$ chP models:%
\begin{equation}
\mathcal{M}_{r}\equiv \frac{\langle g.s.|\mathsf{u}_{1}^{r}|g.s.\rangle }{%
\langle g.s.|g.s.\rangle },\text{ \ \ }\forall r\in \{1,...,p-1\}
\end{equation}%
admit a natural interpretation as spontaneous magnetizations in terms of the
spin chain formulation associated to these local Hamiltonians. They have
been mainly analyzed in the special representations associated to the
super-integrable Z$_{p}$ chP model, characterized by the following
constrains:%
\begin{equation}
x_{\text{q}_{n}}^{p}=y_{\text{q}_{n}}^{p}=x_{\text{r}_{n}}^{p}=y_{\text{r}%
_{n}}^{p}=\frac{1+k^{\prime }}{k},\text{ }\forall n\in \{1,...,\mathsf{N}\}%
\text{ \ }\rightarrow \text{ \ }\bar{\theta}=\theta =\pi /2.
\end{equation}%
In these special representations the Z$_{p}$ chP model also has an
underlying Onsager algebra \cite{NJG-BS-ChP-auYMcCPTY87} generated by the two
components $H_{0}$ and $H_{1}$ of the local quantum Hamiltonians. The
following thermodynamic limits:%
\begin{equation}
\mathcal{M}_{r}=(1-k^{2})^{\frac{r(p-r)}{2p^{2}}},\text{ \ \ }\forall r\in
\{1,...,p-1\}  \label{NJG-BS-ChP-order paprameters}
\end{equation}%
have been first argued by perturbative computations in \cite{NJG-BS-ChP-AMcCPT89} and
then proven with techniques\footnote{%
See Section 1.1 for an historical recall.} which apply only starting from
finite lattice computations in the super-integrable case. Nevertheless, as
argued in \cite{NJG-BS-ChP-auYP11}, the formulae $\left( \ref{NJG-BS-ChP-order paprameters}\right) 
$ should hold true for the general homogeneous Z$_{p}$ chP models.
It is then relevant pointing out that our approach should give us the
possibility to prove this statement for general representations without the 
need to be restricted to the super-integrable case and our SOV results already provide simple determinant formulae for the matrix elements associated to $\mathcal{M}_{p-1}$ in the finite size and inhomogeneous regime. 

\section{Conclusion and outlook}
\subsection{Conclusions}
In this article we have considered general cyclic representations of the
6-vertex Yang-Baxter algebra on \textsf{N}-sites finite lattices and
analyzed the associated Bazhanov-Stroganov model and consequently the chiral Potts
model. We have derived a reconstruction for all local operators in terms of standard Sklyanin's quantum separate variables and characterized by one determinant formulae of \textsf{N}$\times $\textsf{N} matrices the scalar products of separate states.
These findings imply that the action of any local operator on transfer matrix eigenstates reduces to a finite sum of separate states which
allows to characterize matrix elements of any local operator as finite sum of determinants of the scalar product type. Moreover, we
have obtained: form factors of the local operators $\mathsf{u}_{n}^{-1}$ and $\alpha _{0,n}^{-1}$ expressed by one determinant formulae
obtained by modifying a single row in the scalar product matrices; form factors of a basis of operators expressed by one determinant
formulae obtained by modifying the scalar product matrices by introducing rows which coincide with those of Vandermonde's matrix
computed in the spectrum of the separate variables.

Let us comment that it would be desirable to get also for the generators $%
\mathsf{v}_{n}$ of the local Weyl algebras simple one determinant formulae
as for the generators $\mathsf{u}_{n}$ (at this moment we have expressed its form factors as finite sums of determinants); this interesting
issue is currently under investigation. One important motivation to derive
form factors of local operators by simple determinant formulae is for their
use as efficient tools for the computations of correlation
functions. The decomposition of
the identity $\left( \ref{NJG-BS-ChP-Id-decomp}\right) $ allows to write correlation
functions in spectral series of form factors and so it allows to analyzed numerically them mainly by the same tools developed in \cite{NJG-BS-ChP-CM05}
in the ABA framework and used in the series of works\footnote{By this numerical approach, relevant physical observables (like the so
called dynamical structure factors) were evaluated and successfully compared with the measurements accessible by neutron
scattering experiments \cite{NJG-BS-ChP-Bloch36}-\cite{NJG-BS-ChP-Balescu75}.}
\cite{NJG-BS-ChP-CM05}-\cite{NJG-BS-ChP-CCS07}. Indeed, in our SOV framework we have determinant
representations of the form factors and eventually complete characterization of the
transfer matrix spectrum in terms of the solutions of a system of Bethe
equations type. Let us mention that in a recent series of papers \cite{NJG-BS-ChP-KKMST09}-\cite{NJG-BS-ChP-KP12} the problem to compute the asymptotic behavior of
correlation functions has been successfully addressed\footnote{These results have been also successfully compared with those obtained previously with a method relying mainly on the Riemann-Hilbert analysis of
related Fredholm determinants \cite{NJG-BS-ChP-KKMST09++}-\cite{NJG-BS-ChP-K1011}.} with a method which is in principle susceptible to be extended to
any (integrable) quantum model possessing determinant representations for
the form factors of local operators \cite{NJG-BS-ChP-KKMST1110} and so also to the
models analyzed by our approach in the SOV framework.

Finally, let us remark that the originality and interest of our current results are also due to the fact that matrix elements of local operators were so far mainly
confined to the special class of super-integrable representations of Z$_{p}$
chiral Potts model. As these representations can be obtained by taking well
defined limits on the parameters of a generic (non-super-integrable)
representation to which SOV applies, it is then an interesting issue to
investigate how from our form factor results one can reproduce also those
known in the super-integrable case. About this point it is worth mentioning
that in the special case ($p$=2) of the generalized Ising model, it was
already remarked in \cite{NJG-BS-ChP-IPSTvG09} that the matrix elements of the local
spin operators obtained in the SOV framework in \cite{NJG-BS-ChP-GIPST08} admit
factorized forms similar to those conjectured in \cite{NJG-BS-ChP-Ba09} and proven in 
\cite{NJG-BS-ChP-IPSTvG09} for the super-integrable Z$_{p}$ cases for general $p\geq 2$.

In a future paper, we will analyze the homogeneous and thermodynamic limits
focusing the attention on the derivation of the order parameter formulae for
the general homogeneous Z$_{p}$ chiral Potts models. These formulae were
proven with techniques working only in the super-integrable case but they
are expected to be true \cite{NJG-BS-ChP-auYP11} for the general homogeneous Z$_{p}$
chiral Potts models. Our approach should give access to a proof of this
statement from the finite lattice in general representations and we find
encouraging the fact that the matrix element describing the order parameter:%
\begin{equation}
\mathcal{M}_{p-1}\equiv \frac{\langle g.s.|\mathsf{u}_{1}^{-1}|g.s.\rangle }{%
\langle g.s.|g.s.\rangle }
\end{equation}%
admits simple determinant formula in our approach.

\subsection{Outlook}
It is worth recalling that in the literature of quantum integrable models there exist some results on form factors derived by different applications of separation of variable methods. For a more detailed analysis of the most relevant preexisting results and an explicit comparison
with those obtained by our method in separation of variables we address the reader to \cite{NJG-BS-ChP-GMN12-SG}. Here, we
want just recall the Smirnov's results \cite{NJG-BS-ChP-Sm98}, in the case of the integrable quantum Toda chain \cite{NJG-BS-ChP-Sk1}, \cite{NJG-BS-ChP-GM}-\cite{NJG-BS-ChP-KL99}
and those of Babelon, Bernard and Smirnov \cite{NJG-BS-ChP-BBS96,NJG-BS-ChP-BBS97}, in the case of the restricted sine-Gordon at the reflectionless points.
In both these cases form factors of local operators were argued\footnote{The absence of a direct reconstruction of the local operators in terms of the Sklyanin's quantum separate variables was the motivation in \cite{NJG-BS-ChP-Sm98,NJG-BS-ChP-BBS97} to use some well-educated guess relying on counting arguments for the characterization of local operators basis and to use semi-classical arguments relying on the classical SOV-reconstruction for the identification of primary fields \cite{NJG-BS-ChP-BBS96,NJG-BS-ChP-Sm98-0}. Note that a reconstructions of local operators in the lattice Toda model have been achieved in \cite{NJG-BS-ChP-OB-04} in terms of a set of quantum separate variables defined by a change of variables in terms of the original Sklyanin's quantum separate variables. Recent analysis of this reconstruction problem for the lattice Toda model appear also in \cite{NJG-BS-ChP-Sk13} and \cite{NJG-BS-ChP-K13}.} to have a determinant form. A strong similarity in the form of the results appears: the elements of the matrices whose determinants give the form factors are expressed as
\textquotedblleft convolutions\textquotedblright, over the spectrum of each separate variable, of the product of the corresponding
separate components of the wave functions times contributions associated to the action of local operators. It is then remarkable that also our results fall in this general form. This observation and the potential generality of the SOV method leads to the expectation of an
universality in the SOV characterization of form factors.

A natural project is then to develop explicitly our method for a set of fundamental integrable quantum models providing determinant
representations for form factors. This SOV method is not restricted to the case of cyclic
representation and applies to a large class of integrable quantum models which were not tractable with other methods and in particular by algebraic Bethe ansatz. There exist already several key integrable quantum
models associated by QISM to highest weight representations of the Yang-Baxter algebras and generalization of it for which this program has been developed. In
\cite{NJG-BS-ChP-N12-0,NJG-BS-ChP-N12-1,NJG-BS-ChP-N12-2,NJG-BS-ChP-FKN13,NJG-BS-ChP-FN13} and \cite{NJG-BS-ChP-N12-3} our approach has been respectively implemented for the spin-1/2 XXZ and the spin-s XXX
inhomogeneous quantum chains with antiperiodic boundary conditions, for the spin-1/2 XXZ and XYZ open quantum chains with general non-diagonal integrable boundary conditions \cite{NJG-BS-ChP-Skly88}-\cite{NJG-BS-ChP-GZ94} and finally for the spin-1/2 representations of highest weight type of the dynamical 6-vertex
Yang-Baxter algebra. In all these models the universality we just discussed in the structure of the matrix elements of local operator has been verified.

\section*{Acknowledgements}
N. G. is supported by the University of Cergy-Pontoise. He acknowledges the support of ENS Lyon and ANR grant ANR-10-BLAN-0120-04-DIADEMS during his thesis when most of this work was done. N. G. would also like to thank the YITP Institute of Stony Brook for hospitality. J. M. M. is supported by CNRS, ENS Lyon and by the grant ANR-10-BLAN-0120-04-DIADEMS. G. N. gratefully thanks Barry McCoy for his teachings on pre-existing results on the Bazhanov-Stroganov model, in particular about the order parameters. G. N. is supported by National Science Foundation grants PHY-0969739. G. N. gratefully acknowledge the YITP Institute of Stony Brook for the opportunity to develop his research programs. G. N. would like to thank the Theoretical Physics Group of the Laboratory of Physics at ENS-Lyon for hospitality, under support of ANR-10-BLAN- 0120-04-DIADEMS, which made possible this collaboration.



\begin{small}

\end{small}

\end{document}